\documentclass[11pt]{amsart}

\usepackage[text={450pt,660pt},centering]{geometry}

\usepackage{color}
\usepackage{esint,amssymb}
\usepackage{graphicx}
\usepackage{MnSymbol}
\usepackage{amsmath, mathtools}
\usepackage[colorlinks=true, pdfstartview=FitV, linkcolor=blue, citecolor=blue, urlcolor=blue,pagebackref=false]{hyperref}
\usepackage{microtype}

\usepackage{bm}
\usepackage{scalerel} 
\usepackage{dsfont}
\usepackage{mathrsfs}
\usepackage[font={footnotesize}]{caption}
\usepackage{stmaryrd}
\usepackage{tikz}
\usepackage{xcolor}
\usepackage{contour}
\usetikzlibrary{calc,patterns}
\usepackage{upgreek}

\usepackage{enumitem}

\definecolor{darkgreen}{rgb}{0,0.5,0}
\definecolor{darkblue}{rgb}{0,0,0.7}
\definecolor{darkred}{rgb}{0.9,0.1,0.1}

\makeatletter
\newtheorem*{rep@theorem}{\rep@title}
\newcommand{\newreptheorem}[2]{%
\newenvironment{rep#1}[1]{%
 \def\rep@title{#2 \ref{##1}}%
 \begin{rep@theorem}}%
 {\end{rep@theorem}}}
\makeatother

\newtheorem{theorem}{Theorem}
\newreptheorem{theorem}{Theorem}
\newreptheorem{proposition}{Proposition}
\newreptheorem{lemma}{Lemma}
\newtheorem{proposition}{Proposition}
\newtheorem{lemma}[proposition]{Lemma}
\newtheorem{corollary}[proposition]{Corollary}

\theoremstyle{remark}

\theoremstyle{definition}
\newtheorem{definition}[proposition]{Definition}
\newtheorem{remark}[proposition]{Remark}
\newtheorem{conjecture}[proposition]{Conjecture}

\numberwithin{equation}{section}
\numberwithin{proposition}{section}

\newcommand{\Z}{\mathbb{Z}}
\newcommand{\N}{\mathbb{N}}
\newcommand{\R}{\mathbb{R}}

\newcommand{\E}{\mathbb{E}}
\renewcommand{\P}{\mathbb{P}}
\newcommand{\T}{\mathcal{T}}

\newcommand{\Zd}{\mathbb{Z}^d}

\newcommand{\inte}[1]{%
 {\kern0pt#1}^{\mathrm{o}}%
}

\newcommand{\ep}{\varepsilon}

\newcommand{\per}{\mathrm{per}}

\renewcommand{\subset}{\subseteq}

\DeclareMathOperator{\dist}{dist}

\DeclareMathOperator{\var}{var}

\DeclareMathOperator{\FE}{FE}

\DeclareMathOperator{\Leb}{Leb}

\DeclareMathOperator{\fluc}{Fluc}

\renewcommand{\bar}{\overline}
\renewcommand{\tilde}{\widetilde}

\newcommand{\indc}{1}

\newcommand{\Q}{\mathcal{Q}}
\renewcommand{\S}{\mathcal{S}}

\pgfdeclarepatternformonly{south east lines}{\pgfqpoint{-0pt}{-0pt}}{\pgfqpoint{3pt}{3pt}}{\pgfqpoint{3pt}{3pt}}{
    \pgfsetlinewidth{0.4pt}
    \pgfpathmoveto{\pgfqpoint{0pt}{3pt}}
    \pgfpathlineto{\pgfqpoint{3pt}{0pt}}
    \pgfpathmoveto{\pgfqpoint{.2pt}{-.2pt}}
    \pgfpathlineto{\pgfqpoint{-.2pt}{.2pt}}
    \pgfpathmoveto{\pgfqpoint{3.2pt}{2.8pt}}
    \pgfpathlineto{\pgfqpoint{2.8pt}{3.2pt}}
    \pgfusepath{stroke}}

\begin{document}

\title{Quantitative disorder effects in low-dimensional spin systems}

\author[P. Dario]{Paul Dario}
\address[P. Dario]{CNRS and Laboratoire d’Analyse et de Mathématiques Appliquées (LAMA), Universit\'e Paris-Est Cr\'eteil, Cr\'eteil, France}
\email{paul.dario@u-pec.fr}

\author[M. Harel]{Matan Harel}
\address[M. Harel]{Northeastern University, 360 Huntington Avenue, Boston, MA 02115}
\email{m.harel@northeastern.edu}

\author[R. Peled]{Ron Peled}
\address[R. Peled]{\newline School of Mathematical Sciences, Tel Aviv University, Ramat Aviv, Tel Aviv 69978, Israel. \newline
School of Mathematics, Institute for Advanced Study and Department of Mathematics,
Princeton University, New Jersey, United States}
\email{peledron@tauex.tau.ac.il}

\maketitle


\begin{abstract}
   The Imry--Ma phenomenon, predicted in 1975 by Imry and Ma and rigorously established in 1989 by Aizenman and Wehr, states that first-order phase transitions of low-dimensional spin systems are `rounded' by the addition of a quenched random field coupled to the quantity undergoing the transition. The phenomenon applies to a wide class of spin systems in dimensions $d\le 2$ and to spin systems possessing a continuous symmetry in dimensions $d\le 4$.

   This work provides quantitative estimates for the Imry--Ma phenomenon: in a cubic domain of side length $L$, we study the effect of the boundary conditions on the spatial and thermal average of the quantity coupled to the random field. We show that the boundary effect diminishes at least as fast as an inverse power of $\log\log L$ in general two-dimensional spin systems. For systems possessing a continuous symmetry, we show that the boundary effect diminishes at least as fast as an inverse power of $L$ in two and three dimensions and at least as fast as an inverse power of $\log\log L$ in four dimensions. Finally, we establish a partial uniqueness results for translation-covariant Gibbs states, and prove that, for almost every realization of the random field, all such states must agree on the thermally-averaged value of the quantity coupled to the random field.
   
   Specific models of interest for the obtained results include the random-field $q$-state Potts model, the Edwards-Anderson spin glass model, and the random-field spin $O(n)$ models.

\end{abstract}

\section{Introduction}\label{sec:introduction}

The large-scale properties of equilibrium statistical physics systems with quenched (frozen-in) disorder can differ significantly from those of the corresponding non-disordered systems~\cite{Bo06,stein2013spin,T03B,T03B2}. The present understanding of such phenomena is still lacking, in both the physical and mathematical literature, but some cases are better understood. This work focuses on the so-called Imry--Ma phenomenon, by which first-order phase transitions of low-dimensional spin systems are `rounded' by the addition of a quenched random field to the quantity undergoing the transition. Imry and Ma~\cite{IM75} studied the Ising and spin $O(n)$ models and predicted that, in low dimensions, the addition of a random independent magnetic field causes the systems to lose their characteristic low-temperature ordered states, even when the strength of the added field is arbitrarily weak. Specifically, they predicted this effect for the random-field Ising model at all temperatures (including zero temperature) in dimensions $d \leq 2$, and for the random-field spin $O(n)$-model, with $n \geq 2$, at all temperatures in dimensions $d \leq 4$. Their predictions were confirmed, and greatly extended, in the seminal work of Aizenman and Wehr~\cite{AW1989, AW89}, who rigorously established the rounding phenomenon for a general class of spin systems in dimensions $d \leq 2$ and for spin systems with suitable continuous symmetry in dimensions $d \leq 4$.

Let us informally describe the Aizenman--Wehr result. Consider a spin system on $\Zd$, with a formal translation-invariant Hamiltonian $H$. Let $\eta=(\eta_v)_{v\in\Zd}$ be independent and identically distributed random variables (or random vectors). Construct a  disordered spin system by modifying the Hamiltonian $H$ to
\begin{equation*}
    H^{\eta, \lambda}(\sigma) := H(\sigma) - \lambda \sum_{v\in\Zd} \eta_v\cdot f(\T_{-v}(\sigma))
\end{equation*}
where $\T$ denotes the translation operator defined by $\T_v(\sigma)_u=\sigma_{u-v}$, where $f$ is the observable to which the random field $\eta$ is coupled, and where $\lambda$ is the disorder strength. Aizenman and Wehr prove that, under suitable assumptions on $H$, $\eta$ and $f$, in dimensions $d\le 2$, at all temperatures and positive disorder strengths, the limit
\begin{equation}\label{eq:AW limit}
    \lim_{L \to \infty} \frac{1}{\left| \Lambda_L \right|} \sum_{v \in \Lambda_L} \left\langle f (\mathcal{T}_v (\sigma)) \right\rangle_{\mu}
\end{equation}
exists with probability one (in terms of $\eta$), and takes the same value for all infinite-volume Gibbs measures $\mu$ of the disordered system. The notation $\left\langle\cdot\right\rangle_\mu$ denotes the thermal average over the Gibbs measure $\mu$; it should be noted that, while the set of Gibbs measures of the disordered system depends on $\eta$, Aizenman and Wehr further prove that the common limit~\eqref{eq:AW limit} does not depend on $\eta$. An analogous statement is proved to hold in dimensions $d\le 4$ when the spin system and the added disorder satisfy suitable hypotheses of continuous symmetry. Additional information on the Aizenman--Wehr setup and result, and their comparison with those of the present work, is provided in Section~\ref{sectioncommentsandopenproblems}.

The Aizenman--Wehr theorem does not provide \emph{quantitative} estimates on the rate of convergence in~\eqref{eq:AW limit} and it is the goal of this work to obtain such a quantified result. Specifically, the quantity that we study is
\begin{equation}\label{eq:quantity to be studied}
    \sup_{\tau_1, \tau_2} \left| \frac{1}{|\Lambda_L|} \sum_{v \in \Lambda_L} \left( \left\langle f (\mathcal{T}_v (\sigma)) \right\rangle_{\Lambda_L}^{\tau_1} - \left\langle f (\mathcal{T}_v (\sigma)) \right\rangle_{\Lambda_L}^{\tau_2} \right) \right|
\end{equation}
where $\left\langle\cdot\right\rangle_\Lambda^\tau$ is the thermal average under the finite-volume Gibbs measure (of the disordered system) in the volume $\Lambda\subset\Zd$ with boundary conditions $\tau$, where $\Lambda_L := \left\{ -L,  \ldots, L \right\}^d\subset\Zd$ and $\left| \Lambda_L \right|$ is its cardinality, and where the supremum is taken over all choices of boundary conditions. We are thus measuring the largest discrepancy in the value of the spatially and thermally averaged observable, which may arise when changing the boundary conditions in finite volume. This quantity may be thought of as a particular kind of correlation decay rate. We emphasize that the boundary conditions $\tau_1,\tau_2$ appearing in~\eqref{eq:quantity to be studied} are allowed to depend on $L$ and on the specific realization of the disorder $\eta$.

Our main results are that the following holds with high probability over $\eta$: (1) For a general class of two-dimensional spin systems the quantity~\eqref{eq:quantity to be studied} decays at least as fast as an inverse power of $\log\log L$, (2) For a general class of spin systems with continuous symmetry, the quantity~\eqref{eq:quantity to be studied} decays at least as fast as an inverse power of $L$ in dimensions $d=2,3$ and at least as fast as an inverse power of $\log\log L$ in dimension $d=4$. These results are shown to hold at all temperatures and positive disorder strengths.

Quantitative estimates of the type obtained here have recently been developed for the two-dimensional random-field Ising model~\cite{AHP20, AP19, C18, ding20192exponential}, where it was shown that correlations decay at an exponential rate (in the nearest-neighbor case) at all temperatures and positive disorder strength; see Section~\ref{relatedresults} for more details. However, to our knowledge, this is the only spin system for which quantitative estimates have previously been developed, and their derivation appears to crucially rely on monotonicity properties of the Ising model which are absent for general spin systems. Specific examples of interest for the results obtained here include the two-dimensional random-field $q$-state Potts (with $q\ge 3$) and Edwards-Anderson spin glass models, and the $d$-dimensional random-field spin $O(n)$ models (with $n\ge 2$) in dimensions $d\le 4$.

We further point out that, while the assumptions placed by our results on the spin system, disorder distribution and noised observable are different, and in many ways more restrictive than those imposed by Aizenman--Wehr, an advantage of our methods is that they apply also to systems which are not invariant under translations (see also Section~\ref{sectioncommentsandopenproblems}).

In spin systems satisfying suitable monotonicity properties, such as the random-field Ising model, the Aizenman--Wehr result allows to conclude that the disordered system possesses a unique Gibbs measure at all temperatures. Such a conclusion is false for general spin systems, but we provide a related conjecture pertaining to the possible general behavior (Section~\ref{sec.conj.uniq}). The conjecture extends the well-known belief that the two-dimensional Edwards-Anderson spin glass model possesses a unique ground-state pair (see, e.g.,~\cite{ADN10,NS00, NS96Spa, NS01} and the references therein).

The rest of the paper is structured as follows: Section~\ref{sec:2d systems} presents the setup and results for general two-dimensional spin systems. The setup and results for spin systems with continuous symmetry are presented in Section~\ref{sec:continuous symmetry results}. Section~\ref{relatedresults} is devoted to an overview of related results. An outline of our proofs is given in Section~\ref{Secstratproof}. Section~\ref{section2not} collects preliminary properties of the spin systems that we study. The proofs of the results pertaining to general two-dimensional spin systems are presented in Section~\ref{sectionDSS} while the proofs pertaining to spin systems with continuous symmetry are presented in Section~\ref{SectionCSS}. Section~\ref{sectioncommentsandopenproblems} discusses additional points and highlights several open problems.

\section{Two-dimensional disordered spin systems}\label{sec:2d systems}
In this section we describe the quantitative decay rates that are obtained for a general class of two-dimensional disordered spin systems. We first describe the spin systems to which the results apply (Section~\ref{sec:general setup}), then proceed to a list of specific examples which clarify and add interest to the general definitions (Section~\ref{sec:examples}), and finally describe the results themselves (Section~\ref{sec:results}).

\subsection{The general setup} \label{sec:general setup}
We work on the standard integer lattice $\Zd$, in which we denote the standard orthonormal basis by $e_1 , \ldots, e_d$; we write $u \sim v$ if two vertices $u , v$ are nearest neighbors. We let $\|\cdot\|_\infty$ denote the $\ell_\infty$ distance on $\Zd$ and, for $v\in\Zd$ and an integer $R\ge 0$, set $B(v,R):=\{w\in\Zd\,\colon\,\|w-v\|_\infty\le R\}$ to be the ball of radius $R$ around $v$ in the $\ell_\infty$ distance. For integer $R\ge 0$, denote the external vertex boundary to distance $R$ of a set $\Lambda \subseteq \Zd$ by $\partial^R \Lambda := \left\{ v \in \Zd \setminus \Lambda \,:\, B(v,R)\cap\Lambda\neq\emptyset \right\} $ and set $\Lambda^{+R} := \Lambda \cup \partial^R \Lambda$. We also abbreviate $\partial\Lambda := \partial^1\Lambda$ and $\Lambda^+:=\Lambda^{+1}$.

Our general results (which do not rely on continuous symmetry) apply to disordered spin systems defined via~\eqref{eq:13001511} and~\eqref{eq:Gibbs measure general setup} below and built as follows.

\noindent {\bf The base system:} A non-disordered spin system is constructed from the following elements.
\begin{enumerate}[leftmargin=*]
    \item \emph{State space and configuration space:} let $(\S, \mathcal{A}, \kappa)$ be a probability space. Configurations of the spin system are functions $\sigma:\Zd\to\S$. Their restriction to a subset $\Lambda\subset\Zd$ is denoted $\sigma_{\Lambda}$.
    \item \emph{Hamiltonian:} for each finite $\Lambda \subseteq \Zd$, let $H_{\Lambda} : \S^{\Zd} \to \R$ be a bounded measurable map. We assume that this family of Hamiltonians $(H_\Lambda)_{\Lambda \subseteq \Zd}$ satisfies the following two properties:
        \begin{enumerate}[leftmargin=*]
            \item \emph{Consistency:} for each finite $\Lambda \subseteq \Zd$, inverse temperature $\beta>0$, and boundary condition $\tau:\Zd\setminus\Lambda\to\S$, define the finite-volume Gibbs measure
\begin{equation*}
    \mu_{\beta , \Lambda, \tau} \left( d\sigma \right) :=  \frac 1{Z_{\beta , \Lambda, \tau}} \exp \left( - \beta  H_{\Lambda}(\sigma) \right) \prod_{v \in \Lambda} \kappa(d \sigma_v)\prod_{v\in\Zd\setminus\Lambda}\delta_{\tau_v} \left( d\sigma_v \right),
\end{equation*}
    where $Z_{\beta , \Lambda, \tau}$, called the partition function, normalizes $\mu_{\beta, \Lambda, \tau}$ to be a probability measure, and $\delta_{\tau_v}$ is the Dirac delta measure at $\tau_v$. We denote by $\left\langle \cdot \right\rangle_{\beta,\Lambda}^{\tau}$ the expectation operator with respect to $\mu_{\beta , \Lambda, \tau}$. We omit $\beta$ from the notation when it is clear from context.

    We assume that the family of Hamiltonians $(H_\Lambda)_{\Lambda \subseteq \Zd}$ satisfies the following consistency relation (finite-volume Gibbs property): for each pair of finite subsets $\Lambda' \subseteq \Lambda$, each inverse temperature $\beta > 0$, each boundary condition $\tau_0:\Zd \setminus \Lambda \to \S$, and each bounded measurable $g : \S^{\Zd} \to \R$,
    \begin{equation} \label{eq:17391211}
        \left\langle  \left\langle g(\sigma) \right\rangle_{\beta,\Lambda'}^{\tau} \right\rangle_{\beta,\Lambda}^{\tau_0} =  \left\langle g(\sigma) \right\rangle_{\beta,\Lambda}^{\tau_0},
    \end{equation}
    where the spin $\sigma$ in the left-hand side is distributed according to the measure $\mu_{\beta , \Lambda', \tau}$, and the boundary condition $\tau$ is random and is distributed according to the restriction of $\mu_{\beta , \Lambda, \tau_0}$ to the set $\Zd \setminus \Lambda'$.  The identity~\eqref{eq:17391211} is known as the Dobrushin, Lanford and
Ruelle (DLR) equation for Gibbs measures (see~\cite[Chapters 1 and 2]{georgii2011gibbs}). 
    
    If $\Lambda = \{-L,\dots, L\}^d$, for some $L$, we can also define a periodic measure $\mu_{\beta , \Lambda, \mathrm{Per}} \left( d\sigma \right)$. Let $\mathrm{Per}_{\Lambda}$ be the set of configurations $\sigma$ which are $(2L+1)$ periodic --- i.e. $\sigma_u = \sigma_{u +v}$ for every $u \in \Lambda$ and $v \in (2L+1) \cdot \Zd$. There is a natural bijection between the spaces $\mathcal{S}^{\Lambda}$ and $\mathrm{Per}_{\Lambda}$ obtained by extending periodically functions defined on $\Lambda$ to $\Zd$. We then let $\kappa_{\mathrm{Per}, \Lambda}$ be the pushforward of the measure $\prod_{v \in \Lambda} \kappa(d\sigma_v)$ by the extension mapping, and define 
    \begin{equation*}
    \mu_{\beta , \Lambda, \mathrm{Per}} \left( d\sigma \right) :=  \frac 1{Z_{\beta , \Lambda, \mathrm{Per}}} \exp \left( - \beta  H_{\Lambda}(\sigma) \right) \kappa_{\mathrm{Per},  \Lambda}(d \sigma) 
\end{equation*}
on the set of periodic configurations $\mathrm{Per}_{\Lambda}$. The periodic measure must also satisfy the same consistency conditions described above. Below, whenever we take suprema over possible boundary conditions, we include the periodic measure. 
    \item \emph{Bounded boundary effect:} We assume that there exists a constant $C_H \geq 0$ such that for each finite $\Lambda \subseteq \Zd$,
    \begin{equation} \label{def.cteCH}
       \left| H_{\Lambda } (\sigma) - H_{\Lambda } (\sigma')\right| \leq C_H \left| \partial \Lambda\right|\quad\text{for $\sigma , \sigma':\Zd\to\S$ satisfying $\sigma_\Lambda = \sigma'_\Lambda$}.
    \end{equation}
    \end{enumerate}
\end{enumerate}
\noindent {\bf The disordered system:} To form the disordered system from the base system, we add a term of the form $\sum_{v \in\Lambda} \eta_v\cdot f_v(\sigma)$ to the Hamiltonian, where $(\eta_v)$ is a family of $m$-dimensional random vectors (the quenched disorder) and $(f_v)$ is a family of $m$-dimensional functions of the configuration such that $f_v(\sigma)$ depends only on the restriction of $\sigma$ to a neighborhood of $v$ of some fixed radius $R$.
\begin{enumerate}[leftmargin=*]
    \item \emph{Disorder:} we let $\eta=\left( \eta_v \right)_{v \in \Zd}$ be a collection of independent standard $m$-dimensional Gaussian vectors; we use the symbols $\P$ and $\E$ to refer to the law and the expectation operator with respect to $\eta$.
    \item \emph{Noised observables:} fix integers $m \geq 1$ and $R\ge 0$. For each $v \in \Zd$, we let $f_v : \S^{\Zd} \to \R^m$ be a measurable function satisfying
    \begin{align*}
        &\text{Boundedness:}&&\left| f_v(\sigma)\right| \leq 1\text{ for $\sigma\in\S^{\Zd}$},\\
        &\text{Finite range:}&&f_v(\sigma) = f_v(\sigma')\text{ when $\sigma,\sigma'\in\S^{\Zd}$ satisfy $\sigma_{B(v,R)}=\sigma'_{B(v,R)}$}.
    \end{align*}
    \item \emph{Disordered Hamiltonian:}
given a finite $\Lambda \subseteq \Zd$, a fixed realization of $\eta:\Zd\to\R$ and a disorder strength $\lambda > 0$, define the disordered Hamiltonian $H_\Lambda^{\eta,\lambda}:\S^{\Zd}\to\R$ by
\begin{equation} \label{eq:13001511}
    H^{\eta ,\lambda}_{\Lambda} \left( \sigma \right) := H_{\Lambda} \left( \sigma \right) - \lambda \sum_{v \in \Lambda} \eta_{v}\cdot f_v \left( \sigma \right)
\end{equation}
(where the dot product denotes the Euclidean scalar product on $\R^m$) and, for an inverse temperature $\beta >0$ and boundary condition $\tau \in \S^{\Zd \setminus \Lambda}$, define the finite-volume Gibbs measure
\begin{equation}\label{eq:Gibbs measure general setup}
    \mu_{\beta , \Lambda, \tau}^{\eta , \lambda} \left( d\sigma \right) :=  \frac 1{Z_{\beta , \Lambda, \tau}^{\eta, \lambda}} \exp \left( - \beta  H^{\eta, \lambda}_{\Lambda}(\sigma) \right) \prod_{v \in \Lambda} \kappa(d \sigma_v)\prod_{v\in\Zd\setminus\Lambda}\delta_{\tau_v} \left( d\sigma_v \right),
\end{equation}
where $Z_{\beta , \Lambda, \tau}^{\eta, \lambda}$ is the normalization constant which makes the measure $ \mu_{\beta, \Lambda, \tau}^{\eta, \lambda}$ a probability measure. We denote by $\left\langle \cdot \right\rangle_{\beta,\Lambda}^{\tau, \eta, \lambda}$ the expectation with respect to the measure $\mu_{\beta , \Lambda, \tau}^{\eta, \lambda}$ and refer to it as \emph{the thermal expectation}. When $\beta, \eta$ and $\lambda$ are clear from the context, we will omit them from the notation.

We note that the consistency relation~\eqref{eq:17391211} implies the same identity for the disordered system: Given a pair of finite subsets $\Lambda'\subset\Lambda$, $\eta:\Zd\to\R$, $\lambda > 0$, $\beta>0$, $\tau \in \S^{\Zd \setminus \Lambda}$ and any bounded measurable $g : \S^{\Zd} \to \R$,
\begin{equation} \label{eq:14361411}
    \left\langle  \left\langle g(\sigma) \right\rangle_{\Lambda'}^{\tau} \right\rangle_{\Lambda}^{\tau_0} =  \left\langle g(\sigma) \right\rangle_{\Lambda}^{\tau_0},
\end{equation}
where $\sigma$ in the left-hand side is distributed as $\mu_{\beta , \Lambda', \tau}^{\eta, \lambda}$, and $\tau$ is random and distributed according to the restriction to $\Zd \setminus \Lambda'$ of a random spin configuration distributed as $\mu_{\beta , \Lambda, \tau_0}^{\eta, \lambda}$.
\end{enumerate}

We point out that neither the base system nor the noised observables $(f_v)$ are required to be periodic with respect to translations. Still, it is very natural to work in a \emph{translation-invariant setup}, by which we mean that, for all $v\in\Zd$ and configurations $\sigma$, we have (i) $H_\Lambda(\sigma) = H_{\Lambda+v}(\T_v(\sigma))$ for all finite $\Lambda$, where $\T_v$ is the translation by $v$ operation: $\T_v(\sigma)_u = \sigma_{u-v}$, and (ii) $f_v(\sigma) = f_{\mathbf 0}(\T_{-v}(\sigma))$ (where $\mathbf 0$ is the zero vector in $\Zd$). Indeed all of the examples presented in the next section are of this type; however, we stress that such invariance is not required for our results (see Section~\ref{sec:results}). The additional flexibility of the general definitions allows for inhomogeneities in the base system and further allows to vary the noised observables $(f_v)$ periodically along a sublattice of $\Zd$ (this may make sense, e.g., in antiferromagnetic systems, where the ordered state of the base system is not invariant to all translations), or even to choose $(f_v)$ arbitrarily. 

\subsection{Examples}\label{sec:examples}
We now describe several classical examples of disordered systems which fit within the general class of systems discussed in the previous section. To the best of our knowledge, our general results are already new when specialized to these examples except for the example of the random-field ferromagnetic Ising model where stronger results are known (see Section~\ref{relatedresults}).
\begin{itemize}[leftmargin=*]
    \item \emph{Random-field Ising model:} the state space is $\S := \{-1 , 1\}$ equipped with the counting measure and the Hamiltonian of the base system (the ferromagnetic Ising model) is
    \begin{equation*}
         H_{\Lambda}(\sigma):= -\sum_{\substack{ \{u, v\}\cap \Lambda\neq\emptyset\\ u\sim v}} \sigma_u \sigma_v - \sum_{v \in \Lambda} h \sigma_v,
    \end{equation*}
    for a fixed $h \in \R$.
    The noised observable is $f_v(\sigma) := \sigma_v$ so that $m=1$, $R=0$ and the disordered Hamiltonian is
    \begin{equation*}
         H^{\eta , \lambda }_{\Lambda} \left( \sigma \right) := -\sum_{\substack{ \{u, v\}\cap \Lambda\neq\emptyset\\ u\sim v}} \sigma_u \sigma_v - \sum_{v \in \Lambda} \left( \lambda \eta_{v} + h \right) \sigma_v
    \end{equation*}
    corresponding to the addition of a quenched random field to the Ising model. One similarly forms the random-field \emph{antiferromagnetic} Ising model by removing the minus sign in front of the term $\sum \sigma_u \sigma_v$.
    
    \item \emph{Random-field $q$-state Potts model:} let $q\ge 3$ be an integer. The state space is $\S := \{1, \ldots, q\}$ equipped with the counting measure and the Hamiltonian of the base system (the ferromagnetic $q$-state Potts model) is
    \begin{equation*}
     H_\Lambda(\sigma):=-\sum_{\substack{ \{u, v\}\cap \Lambda\neq\emptyset\\ u\sim v}} \indc_{\left\{ \sigma_u = \sigma_v \right\}} - \sum_{v \in \Lambda} \sum_{k = 1}^q h_k \indc_{\{ \sigma_v = k\}},
    \end{equation*}
    for a fixed $h := (h_1 , \ldots, h_q) \in \R^q$.
    The noised observable is $f_v(\sigma) := \left( \indc_{\{\sigma_v = 1\}}, \ldots, \indc_{\{\sigma_v = q\}} \right)$ so that $m=q$, $R=0$ and the disordered Hamiltonian is
    \begin{equation*}
         H^{\eta , \lambda, h}_{\Lambda} \left( \sigma \right) := -\sum_{\substack{ \{u, v\}\cap \Lambda\neq\emptyset\\ u\sim v}} \indc_{\{ \sigma_u = \sigma_v \}} - \sum_{v \in \Lambda} \sum_{k = 1}^q \left( \lambda \eta_{v, k} + h_k \right) \indc_{\{ \sigma_v = k\}}
    \end{equation*}
    where we write $\eta_v = \left( \eta_{v,1} , \ldots , \eta_{v,q} \right) \in \R^q$ and $h = \left( h_1 , \ldots, h_q \right) \in \R^q$. The effect of the disorder is to add an energetic bonus or penalty to each spin state at each vertex according to quenched random vectors of length $q$ which are assigned independently to the vertices.
    \item \emph{Edwards-Anderson spin glass}: the state space is $\S := \{-1 , 1\}$ equipped with the counting measure and the Hamiltonian of the base system (a uniform system) is
    \begin{equation*}
        H_\Lambda(\sigma):=0.
    \end{equation*}
    The noised observable is $f_v(\sigma):=\left(\sigma_v \sigma_{v+e_i} \right)_{i \in \{ 1 , \ldots, d \}}$ so that $m=d$, $R=1$ and we have, in effect, a noise term $\eta$ assigned to every edge of $\Zd$. The disordered Hamiltonian is (up to a term which does not depend on the spins in $\Lambda$)
    \begin{equation*}
         H^{\eta , \lambda ,h}_{\Lambda} \left( \sigma \right) := -\sum_{\substack{ \{u, v\}\cap \Lambda\neq\emptyset\\ u\sim v}} (\lambda\eta_{u,v} + h) \sigma_u \sigma_v
    \end{equation*}
    corresponding to adding a random energetic bonus/penalty to satisfied edges (edges with the same spin state assigned to both endpoints) independently between the different edges.
    \item \emph{Random-field spin $O(n)$-model:} let $n\ge 2$ be an integer. The state space is $\S := \mathbb{S}^{n-1} \subseteq \R^n$ equipped with the uniform measure and the Hamiltonian of the base system (the spin $O(n)$ model) is
    \begin{equation*}
        H_\Lambda(\sigma):=-\sum_{\substack{ \{u, v\}\cap \Lambda\neq\emptyset\\ u\sim v}} \sigma_u \cdot \sigma_v.
    \end{equation*}
    The noised observable is $f_v(\sigma) = \sigma_v\in\R^n$ so that $R=0$, $m = n$ and the disordered Hamiltonian is
    \begin{equation*}
         H^{\eta ,\lambda, h}_{\Lambda} \left( \sigma \right) := -\sum_{\substack{ \{u, v\}\cap \Lambda\neq\emptyset\\ u\sim v}} \sigma_u \cdot \sigma_v - \sum_{v \in \Lambda} \left( \lambda \eta_{v} + h \right) \cdot \sigma_v
    \end{equation*}
    corresponding to the addition of a quenched random field in a random direction, independently at each vertex, to the spin $O(n)$-model.

    We point out that while the results of the next section apply to the random-field spin $O(n)$-model, stronger results are obtained for it in Section~\ref{sec:continuous symmetry results} by relying on its continuous symmetry (at $h=0$).

\end{itemize}

\subsection{Results}\label{sec:results}

\subsubsection{Main result}
Throughout the section, we fix a disordered spin system of the type defined in Section~\ref{sec:general setup}, described by a state space $(\S, \mathcal{A}, \kappa)$, a family of Hamiltonians for the base system $(H_\Lambda)_{\Lambda \subseteq \Zd}$ satisfying the bounded boundary effect assumption with constant $C_H$ and a family of observables $(f_v)$ with a range of $R$, taking values in $\R^m$. The disorder $(\eta_v)$ is a collection of independent standard $m$-dimensional Gaussian vectors. The results of this section apply to two-dimensional spin systems so we further fix $d=2$.

In the first theorem, we present estimates on the effect of the boundary condition on the thermal expectation of the spatially-averaged noised observables in a finite box, in the presence of the disorder. For the second part of the theorem, recall that the notion of translation-invariant setup was defined in Section~\ref{sec:general setup}.

Here and later, we let $\Lambda_L := \left\{ -L,  \ldots, L \right\}^d\subset\Zd$, let $\left| \Lambda_L \right|$ be its cardinality, and denote by $\left| \cdot \right|$ the Euclidean norm on $\R^m$.
\begin{theorem} \label{thm1prop2.31708main}
Let $\beta > 0$ be the inverse temperature and $\lambda>0$ be the disorder strength. There exist constants $C,c>0$ depending only on $\lambda$, $C_H$, $m$ and $R$ such that, for each integer $L\ge 3$,
\begin{equation} \label{eq:11191006}
    \P \left( \sup_{\tau_1, \tau_2 \in \S^{\Z^2 \setminus \Lambda_L}} \left| \frac{1}{|\Lambda_L|} \sum_{v \in \Lambda_L} \left( \left\langle f_v \left( \sigma \right) \right\rangle_{\Lambda_L}^{\tau_1} - \left\langle f_v \left( \sigma \right) \right\rangle_{\Lambda_L}^{\tau_2} \right) \right|  > \frac{C}{\sqrt[4]{\ln \ln L}}  \right) \le\exp \left( - c L \right)
\end{equation}
and moreover, in a translation-invariant setup,
\begin{equation}\label{eq:fluctuations around limiting value}
    \P \left( \sup_{\tau\in \S^{\Z^2 \setminus \Lambda_L}} \left| \alpha - \frac{1}{|\Lambda_L|} \sum_{v \in \Lambda_L} \left\langle f_{\mathbf 0} \left( \mathcal{T}_{-v} \sigma \right) \right\rangle_{\Lambda_L}^{\tau} \right|  > \frac{C}{\sqrt[4]{\ln \ln L}}  \right) \leq \exp \left( - c L \right)
\end{equation}
where $\alpha\in\R^m$ depends only on the spin system considered, on the inverse temperature $\beta$ and on the disorder strength $\lambda$ (but does not depend on the disorder $\eta$ and on $L$).
\end{theorem}

We mention that, for any fixed $L$, any $v \in \Lambda_L$, and any boundary condition $\tau$, the function $\eta \mapsto \left\langle f_v \left( \sigma \right) \right\rangle_{\Lambda_L}^{\tau}$ is Lipschitz continuous, uniformly in $v$ and $\tau$. Thus, the quantities inside the probabilities in the previous theorems and the following ones are indeed measurable.

We remark that, while we only prove the result at positive temperature, the same argument yields the following zero-temperature version of the theorem: for a given side length $L$ and a realization of the disorder $\eta$, define an $L$-ground configuration as a configuration $\sigma:\Z^2 \to\S$ satisfying
\begin{equation*}
    H_{\Lambda_L}^{\eta,\lambda}(\sigma) = \inf_{\substack{\sigma':\Z^2 \to\S \\ \sigma'\equiv\sigma \, \mathrm{in}\,\Z^2 \setminus\Lambda_L }} H_{\Lambda_L}^{\eta,\lambda}(\sigma'),
\end{equation*}
and let $G_L^\eta$ be the set of $L$-ground configurations. Then we have
\begin{equation}
    \P \left( \sup_{\sigma_1, \sigma_2 \in G_L^\eta} \left| \frac{1}{|\Lambda_L|} \sum_{v \in \Lambda_L} \left( f_v \left( \sigma_1 \right) -  f_v \left( \sigma_2 \right) \right) \right|  > \frac{C}{\sqrt[4]{\ln \ln L}}  \right) \le\exp \left( - c L \right)
\end{equation}
where, if $G_L^\eta$ is empty, then the supremum is assumed to take the value negative infinity. Additionally, if we assume that the state space is compact and both the base Hamiltonian and noised observables are continuous, we can show that, for all $\eta$ and all $L$, the set $G_L^\eta$ of $L$-ground configurations is non-empty. It is straightforward that this is the case in all the example systems of Section~\ref{sec:examples}.

In the translation-invariant setup, the value $\alpha$ is explicit for some of the models described in Section~\ref{sec:examples}: for the $q$-state Potts model with external field $h = 0$, one has $\alpha = \left( 1/q , \ldots , 1/q\right)$. In the case of the Edwards-Anderson spin glass, a gauge symmetry (obtained by flipping the spins on the even sublattice and changing the sign of $\eta$ on all edges) implies that the density of satisfied edges is about $1/2$, i.e.,
\begin{equation*}
    \lim_{L \to \infty} \frac{1}{\left| E \left(\Lambda_L\right) \right|} \sum_{\substack{x , y \in \Lambda_L \\ x \sim y}}  \left\langle \indc_{\{\sigma_x = \sigma_y \}} \right\rangle_{\Lambda_L}^{\tau_L}   \underset{L \to \infty}{\longrightarrow} \frac 12 \hspace{5mm}\P\mbox{-almost-surely}
    \end{equation*}
for any collection of random boundary condition $\eta \mapsto \tau_L(\eta) \in \S^{\Z^2 \setminus \Lambda_L}$ and $L \geq 1$, where $\left| E \left(\Lambda_L\right) \right|$ denotes the number of edges in the box $\Lambda_L$. In both these examples, our results are novel.

We finally remark that, while our result is stated and proved in dimension $d = 2$, it is possible to adapt the argument to treat the case of the one-dimensional disordered spin systems. In this setting, the system is subcritical and power law decays can be obtained on the expectation of the supremum over pairs of boundary conditions of the difference of the thermally and spatially averaged noised observables. We record below (and without proof) the results one obtains by adapting the techniques developed in the proof of Theorem~\ref{thm1prop2.31708main}:
\begin{equation*}
    \E \left[ \sup_{\tau_1, \tau_2 \in \S^{\Z \setminus \Lambda_L}} \left| \frac{1}{|\Lambda_L|} \sum_{v \in \Lambda_L} \left( \left\langle f_v \left( \sigma \right) \right\rangle_{\Lambda_L}^{\tau_1} - \left\langle f_v \left( \sigma \right) \right\rangle_{\Lambda_L}^{\tau_2} \right) \right| \right] \leq \frac{C}{L^{1/4}},
\end{equation*}
and, in the translation-invariant setup,
\begin{equation*}
    \E \left[ \sup_{\tau\in \S^{\Z \setminus \Lambda_L}} \left| \alpha - \frac{1}{|\Lambda_L|} \sum_{v \in \Lambda_L} \left\langle f_{\mathbf 0} \left( \mathcal{T}_{-v} \sigma \right) \right\rangle_{\Lambda_L}^{\tau} \right| \right] \leq \frac{C}{L^{1/6}}.
\end{equation*}

\subsubsection{Uniqueness conjecture and additional results} \label{sec.conj.uniq}
Theorem~\ref{thm1prop2.31708main} states that the spatial average over $\Lambda_L$ of the thermal expectations $\left\langle f_v \left( \sigma \right) \right\rangle_{\Lambda_L}^{\tau}$ does not depend strongly on the boundary condition $\tau$. It is natural to ask for a more detailed result: to what extent can the value of $\left\langle f_v \left( \sigma \right) \right\rangle_{\Lambda_L}^{\tau}$ at a given vertex $v$ be affected by $\tau$? Can the values at specific vertices $v$ not too close to the boundary of $\Lambda_L$, or even at most $v\in\Lambda_L$, be significantly altered by changing $\tau$ (keeping in mind that the overall spatial average does not depend strongly on $\tau$ in the sense of Theorem~\ref{thm1prop2.31708main})? We conjecture that such a phenomenon cannot occur. We first state this as a conjecture, and then explain additional motivation as well as additional rigorous support (we remind that throughout the section we have fixed a disordered spin system of the type defined in Section~\ref{sec:general setup} and we work in dimension $d=2$)

\begin{conjecture}\label{conj:uniqueness}
Let $\beta > 0$ be the inverse temperature and $\lambda>0$ be the disorder strength. Then, $\P$-almost-surely,
\begin{equation}\label{eq:averaged uniqueness conjecture}
    \lim_{L\to\infty}\sup_{\tau_1, \tau_2 \in \S^{\Z^2 \setminus \Lambda_L}} \frac{1}{|\Lambda_L|} \sum_{v \in \Lambda_L} \left|\left\langle f_v \left( \sigma \right) \right\rangle_{\Lambda_L}^{\tau_1} - \left\langle f_v \left( \sigma \right) \right\rangle_{\Lambda_L}^{\tau_2} \right| = 0
\end{equation}
and moreover, in a translation-invariant setup, $\P$-almost-surely,
\begin{equation}\label{eq:pointwise uniqueness conjecture}
    \lim_{L\to\infty}\sup_{\tau_1, \tau_2 \in \S^{\Z^2 \setminus \Lambda_L}} \left|\left\langle f_{\mathbf 0} \left( \sigma \right) \right\rangle_{\Lambda_L}^{\tau_1} - \left\langle f_{\mathbf 0} \left( \sigma \right) \right\rangle_{\Lambda_L}^{\tau_2} \right| = 0.
\end{equation}
\end{conjecture}
We make several remarks about the conjecture:

In the translation-invariant setup, the pointwise statement~\eqref{eq:pointwise uniqueness conjecture} implies the averaged statement~\eqref{eq:averaged uniqueness conjecture}. This can be argued, for instance, by applying Birkhoff's ergodic theorem to the functions $$g_\ell(\eta):=\sup_{\tau_1, \tau_2 \in \S^{\Z^2 \setminus \Lambda_\ell}} \left|\left\langle f_{\mathbf 0} \left( \sigma \right) \right\rangle_{\Lambda_\ell}^{\tau_1} - \left\langle f_{\mathbf 0} \left( \sigma \right) \right\rangle_{\Lambda_\ell}^{\tau_2} \right|.$$

We also note that the pointwise statement~\eqref{eq:pointwise uniqueness conjecture} may be reformulated as a statement on the ($\eta$-dependent) set of Gibbs measures of the disordered system. Indeed, let us assume that, in addition to having a translation-invariant setup, the state space $\S$ is Polish and compact, with $\mathcal{A}$ the Borel sigma algebra, and the based Hamiltonian and noised observables $f_v$ are continuous (with respect to the product topology). These assumptions allow to extract a subsequential limiting Gibbs state $\mu$ out of a sequence of finite-volume Gibbs states $\mu_k$ on increasing domains so that $\left\langle f_{\mathbf 0} \left( \sigma \right) \right\rangle_{\mu_k}\to\left\langle f_{\mathbf 0} \left( \sigma \right) \right\rangle_\mu$ (where $\left\langle\cdot\right\rangle_\mu$ is the expectation under $\mu$). Then, the statement~\eqref{eq:pointwise uniqueness conjecture} is equivalent to the claim that, at any $\beta, \lambda>0$, it holds $\P$-almost surely that $\left\langle f_{\mathbf 0} \left( \sigma \right) \right\rangle_\mu$ takes the same value for all Gibbs measures $\mu$ of the disordered system.

In monotonic systems such that, for each $i \in \{1 , \ldots, d  \}$ and $L \in \mathbb{N}$, there exist boundary conditions $\tau_{\text{min}},\tau_{\text{max}}$ such that $\left\langle f_{v,i} \left( \sigma \right) \right\rangle_{\Lambda_L}^{\tau_{\text{min}}}\le \left\langle f_{v,i} \left( \sigma \right) \right\rangle_{\Lambda_L}^{\tau}\le \left\langle f_{v,i} \left( \sigma \right) \right\rangle_{\Lambda_L}^{\tau_{\text{max}}}$ for all $\tau$ and $v$ (such as the plus and minus boundary conditions in the random-field ferromagnetic Ising model), the averaged statement~\eqref{eq:averaged uniqueness conjecture} follows immediately from Theorem~\ref{thm1prop2.31708main}. If we additionally assume that we are in the translation-invariant setup described above, the pointwise bound~\eqref{eq:pointwise uniqueness conjecture} also follows.

Another statement which would be of interest to prove, and which is implied by the averaged statement~\eqref{eq:averaged uniqueness conjecture}, is the following: $\P$-almost-surely,
\begin{equation}\label{eq:weaker averaged uniqueness conjecture}
    \lim_{\ell\to\infty}\lim_{L\to\infty}\sup_{\tau_1, \tau_2 \in \S^{\Z^2 \setminus \Lambda_L}} \frac{1}{|\Lambda_\ell|} \sum_{v \in \Lambda_\ell} \left|\left\langle f_v \left( \sigma \right) \right\rangle_{\Lambda_L}^{\tau_1} - \left\langle f_v \left( \sigma \right) \right\rangle_{\Lambda_L}^{\tau_2} \right| = 0.
\end{equation}

We also expect the following zero-temperature version of Conjecture~\ref{conj:uniqueness} to hold: in both~\eqref{eq:averaged uniqueness conjecture} and~\eqref{eq:pointwise uniqueness conjecture}, replace the supremum over $\tau_1,\tau_2\in\S^{\Zd \setminus \Lambda_L}$ by the supremum over $\sigma_1,\sigma_2\in G_L^\eta$, where $G_L^\eta$ is the set of $L$-ground configurations and replace $\left\langle f_v \left( \sigma \right) \right\rangle_{\Lambda_L}^{\tau_i}$ by $f_v(\sigma_i)$. The zero-temperature pointwise statement would imply that, $\P$-almost-surely, all ground configurations $\sigma$ (i.e., all $\sigma\in\cap_L G_L^\eta$) agree on $f_v(\sigma)$ for all $v\in\Zd$. In the specific case of the Edwards-Anderson spin glass model this is the same as the well-known belief that the two-dimensional model has a unique ground-state pair.

\medskip

One way in which typical configurations of the system can fulfill Theorem~\ref{thm1prop2.31708main} but avoid the uniqueness statement in~\eqref{eq:averaged uniqueness conjecture} is if the configurations are differently ordered on the two bipartite classes of $\Z^2$ (by the two bipartite classes, we mean the vertices with even sum of coordinates and the vertices with odd sum of coordinates). Such a situation is familiar from the (non-disordered) antiferromagnetic Ising model at low temperature, in which typical configurations have a chessboard-like pattern (and the boundary conditions can decide which of the two chessboard patterns will emerge). Our next result shows, as a special case, that such behavior cannot arise for the disordered systems considered herein, by showing that the boundary conditions cannot significantly influence the average value of $\left\langle f_v \left( \sigma \right) \right\rangle$ on any \emph{deterministic} set of positive density.

\begin{theorem} \label{prop:1109}
Let $\beta > 0$ be the inverse temperature and $\lambda>0$ be the disorder strength. There exist constants $C,c>0$ depending only on $\lambda$, $C_H$, $m$ and $R$ such that for each integer $L\ge 3$ and for each weight function $w:\Lambda_L\to[-1,1]^m$,
\begin{equation} \label{eq:13562610}
\P \left( \sup_{\tau_1, \tau_2 \in \S^{\Z^2 \setminus \Lambda_L}}  \left|\frac{1}{|\Lambda_L|}\sum_{v \in \Lambda_L} w(v) \cdot \left(\left\langle f_v \left( \sigma \right) \right\rangle_{\Lambda_L}^{\tau_1} - \left\langle f_v \left( \sigma \right) \right\rangle_{\Lambda_L}^{\tau_2} \right) \right|  \leq \frac{C}{\sqrt[8]{\ln \ln L}}  \right) \geq 1 -  \exp \left( - c L \right).
\end{equation}
\end{theorem}

A second motivation for Theorem~\ref{prop:1109} comes from Parseval's identity, which, in our setting, reads
\begin{equation} \label{15000301}
    \frac{1}{|\Lambda_L|} \sum_{v \in \Lambda_L} \left|\left\langle f_{v} \left( \sigma \right) \right\rangle_{\Lambda_L}^{\tau_1} - \left\langle f_v \left( \sigma \right) \right\rangle_{\Lambda_L}^{\tau_2} \right|^2 = \sum_{ \mathbf{k} \in \Lambda_L} \left| \frac{1}{\left| \Lambda_L \right|}\sum_{v \in \Lambda_L} e^{\frac{2\pi i\,\mathbf{k} \cdot v}{2L+1}} \left( \left\langle f_{v} \left( \sigma \right) \right\rangle_{\Lambda_L}^{\tau_1} - \left\langle f_{v} \left( \sigma \right) \right\rangle_{\Lambda_L}^{\tau_2} \right) \right|^2.
\end{equation}

Consequently, if one can upgrade the statement of Theorem~\ref{prop:1109} and obtain a rate of convergence faster than $1/L$ (instead of the $(\log\log L)^{-1/8}$ term), then it would imply that the right-hand side of~\eqref{15000301} is small, which would then yield~\eqref{eq:averaged uniqueness conjecture}.

As a corollary of Theorem~\ref{prop:1109}, we obtain a quantitative estimate in the spirit of the averaged statement~\eqref{eq:averaged uniqueness conjecture}. The obtained result is weaker than~\eqref{eq:averaged uniqueness conjecture} as it involves the expected value (in the random field) of~$\left\langle f_v \left(  \sigma \right) \right\rangle_{\Lambda_L}^{\tau(\eta)}$.

\begin{corollary} \label{prop:1109cor}
Let $\beta > 0$ be the inverse temperature and $\lambda>0$ be the disorder strength. There exist constants $C,c>0$ depending only on $\lambda$, $C_H$, $m$ and $R$ such that, for each integer $L\ge 3$ and each random (i.e., measurable) pair of boundary conditions $\eta \mapsto \tau_1(\eta), \tau_2(\eta) \in \S^{\Z^2 \setminus \Lambda_L}$,
\begin{equation} \label{eq:15432}
    \frac{1}{\left| \Lambda_L \right|} \sum_{v \in \Lambda_L} \left| \E \left[\left\langle f_v \left(  \sigma \right) \right\rangle_{\Lambda_L}^{\tau_1(\eta)}- \left\langle f_v \left( \sigma \right) \right\rangle_{\Lambda_L}^{\tau_2(\eta)} \right]\right| \leq \frac{C}{\sqrt[8]{\ln \ln L}}.
\end{equation}
 \end{corollary}

If the weights $w(v)$ were allowed to depend on $\eta$, Theorem~\ref{prop:1109} would imply Conjecture~\ref{conj:uniqueness}. A partial result in this direction is formulated in Proposition~\ref{prop:7.8august}, which strengthens Theorem~\ref{prop:1109} to allow $w(v)$ to have a restricted dependence on the disorder $\eta$.

As a second consequence of Theorem~\ref{prop:1109}, we obtain a partial uniqueness result related to Conjecture~\ref{conj:uniqueness}. As mentioned above, in a translation-invariant setup, when the state space $\S$ is Polish and compact, when $\mathcal{A}$ is the corresponding Borel sigma algebra, and when the Hamiltonian and noised observables are continuous, Conjecture~\ref{conj:uniqueness} is equivalent to the claim that the expectation $\left\langle f_{\mathbf{0}}(\sigma)\right\rangle_{\mu}$ takes the same value for every infinite volume Gibbs measures and for almost every realization of the disorder. In Theorem~\ref{theoremuniqueness} below, we establish this result (in the setup of general spin systems introduced in Section~\ref{sec:general setup}) in the specific case of \emph{translation-covariant} Gibbs states defined in the following paragraph.

Let us fix a disorder strength $\lambda>0$ and an inverse temperature $\beta > 0$. For any realization of the random field $\left( \eta_v \right)_{v \in \Zd}$, an \emph{infinite-volume} Gibbs measure with disorder $\eta$ is a probability distribution $\mu$ on the space of configurations $\sigma : \Zd \to \R$ satisfying the consistency relations: for any finite subset $\Lambda \subseteq \Zd$ and any bounded measurable $g : \mathcal{S}^{\Zd} \to \R$, 
    \begin{equation} \label{eq:14042507}
        \left\langle  \left\langle g(\sigma) \right\rangle_{\beta,\Lambda}^{\tau} \right\rangle_{\mu} =  \left\langle g(\sigma) \right\rangle_{\mu},
    \end{equation}
where $\left\langle \cdot \right\rangle_{\mu}$ denotes the expectation with respect to the measure $\mu$, the configuration $\sigma$ in the left-hand side is distributed as $\mu_{\beta , \Lambda, \tau}^{\eta, \lambda}$, and $\tau$ is random and distributed according to the restriction to $\Zd \setminus \Lambda$ of a random spin configuration distributed as $\mu$. We denote by $
\mathcal{P}^\eta(\mathcal{S}^{\Zd})$ the set of infinite-volume Gibbs measures with disorder $\eta$.

A map $\eta \mapsto \mu^\eta \in \mathcal{P}^\eta(\mathcal{S}^{\Zd})$ is measurable if, for any bounded measurable function $g : \S^{\Zd} \to \R$, the map $\eta \mapsto \left\langle g(\sigma) \right\rangle_{\mu^\eta}$ is measurable (as a real-valued function).
 
A \emph{translation-covariant} Gibbs state is a measurable map $\mu : \eta \mapsto \mu^\eta \in \mathcal{P}^\eta(\mathcal{S}^{\Zd})$ satisfying the following property:
\begin{equation*}
    \forall v \in \Zd,~(\mathcal{T}_v)_* \mu^\eta = \mu^{\mathcal{T}_v \eta} \hspace{3mm} \P-\mbox{almost-surely},
\end{equation*}
where $(\mathcal{T}_v)_* \mu^\eta$ denotes the pushforward of the measure $\mu^\eta$ by the translation operator $\mathcal{T}_v$ and $\left( \mathcal{T}_v \eta \right)_u := \eta_{u - v}.$
This definitions generalizes the notion of \emph{translation-invariant} Gibbs states to the setup of disordered systems.

\begin{theorem}[Uniqueness for translation-covariant Gibbs states] \label{theoremuniqueness}
    In the translation-invariant setup, for any inverse temperature $\beta >0$, any disorder strength $\lambda >0$, and any pair $\mu_1 , \mu_2$ of infinite-volume translation-covariant Gibbs measures, one has the identity
    \begin{equation*}
        \left\langle f_{\mathbf{0}}(\sigma) \right\rangle_{\mu_1^\eta} = \left\langle f_{\mathbf{0}}(\sigma) \right\rangle_{\mu_2^\eta} \hspace{3mm} \P-\mbox{almost-surely.}
    \end{equation*}
\end{theorem}
We note that we can produce a translation-covariant Gibbs state via compactness arguments (see Appendix~I of Aizenman--Wehr~\cite{AW89}). However, even in the zero temperature case, it is difficult to obtain good control on the support of this measure for fixed $\eta$. An interesting question is to construct an {\em extremal} translation-covariant Gibbs state. Partial progress has been made in this direction: Cotar--Jahnel--Kulske~\cite{CJK18} have shown the existence and measurability of a translation-invariant, extremal decomposition of translation-covariant random Gibbs states. 

An example where such a construction could have far-reaching consequences is the zero-temperature Edwards--Anderson spin glass model. If one could show that there exists a translation-covariant Gibbs state supported on pairs of configurations that differ by a global sign change, then Theorem~\ref{theoremuniqueness} implies that {\em all translation-covariant} Gibbs states are supported on such ground-state pairs.

As before, versions of Theorem~\ref{prop:1109}, Corollary~\ref{prop:1109cor} and Theorem~\ref{theoremuniqueness} hold at zero temperature. In Theorem~\ref{prop:1109}, the supremum over boundary conditions should be replaced by a supremum over $\sigma_1,\sigma_2\in G_L^\eta$ and $\left\langle f_{v} \left( \sigma \right) \right\rangle_{\Lambda_L}^{\tau_i}$ should be replaced by $f_v(\sigma_i)$. In Corollary~\ref{prop:1109cor}, the random boundary conditions should be replaced by random $\sigma_1,\sigma_2\in G_L^\eta$ and $\left\langle f_v \left(  \sigma \right) \right\rangle_{\Lambda_L}^{\tau_i(\eta)}$ should be replaced by $f_v(\sigma_i)$. In Theorem~\ref{theoremuniqueness}, the translation-covariant Gibbs states $\mu_1 , \mu_2$ should be replaced by translation-covariant ground states $\sigma_1 , \sigma_2$, that is, measurable maps $\sigma_i : \eta \to \sigma_i(\eta) \in \cap_L G_L^\eta$ satisfying $\sigma_i(\mathcal{T}_v \eta) = \mathcal{T}_v \sigma_i(\eta)$ for any $v \in \Zd$, and $\left\langle f_{\mathbf{0}} \left( \sigma \right) \right\rangle_{\mu^\eta_i}$ should be replaced by $f_{\mathbf{0}}(\sigma_i)$.

\section{Spin systems with continuous symmetry}\label{sec:continuous symmetry results}

In this section, we present quantitative results for spin systems with continuous symmetry. The class of systems that we study are versions of the random-field spin $O(n)$ model and are described similarly to Section~\ref{sec:general setup} with the following additional assumptions:
\begin{enumerate}[leftmargin=*]
    \item \emph{State space:} We let $\S$ be the sphere $\mathbb{S}^{n-1}$ for some integer $n \geq 1$, equipped with its Borel sigma algebra and the uniform measure $\kappa$.
    \item \emph{Base Hamiltonian:} we assume that, for any finite subset $\Lambda \subseteq \Zd$, the Hamiltonian $H_{\Lambda}$ takes the form
    \begin{equation} \label{eq12592444}
    H_{\Lambda }(\sigma) = \sum_{\substack{ \{u, v\}\cap \Lambda\neq\emptyset \\ u\sim v}} \Psi \left( \sigma_u , \sigma_v \right),
    \end{equation}
    where the map $\Psi : \mathbb{S}^{n - 1} \times \mathbb{S}^{n - 1} \to \R$ is twice continuously differentiable and rotationally invariant: for any $R \in O(n)$ and any $\sigma_1 , \sigma_2 \in \mathbb{S}^{n-1}$, $\Psi \left(R \sigma_1 , R \sigma_2 \right) = \Psi \left(\sigma_1 ,  \sigma_2 \right)$.
    \item \emph{Disorder:} We let $\eta=\left( \eta_v \right)_{v \in \Zd}$ be a collection of independent standard $n$-dimensional Gaussian vectors, and denote by $\P$ and $\E$  the law and the expectation operator with respect to $\eta$.
    \item \emph{The noised observable:} we assume that $m = n$, and that $f_v\left( \sigma \right) = \sigma_v$. In particular the observables have range $R = 0$.
    \item \emph{Disordered Hamiltonian:} it will turn out to be useful to incorporate a deterministic external magnetic field $h\in\R^n$ in the disordered Hamiltonian. Thus, given additionally a finite $\Lambda \subseteq \Zd$, a fixed realization of $\eta:\Zd\to\R$ and a disorder strength $\lambda > 0$, we define
    \begin{equation} \label{eq125924}
     H^{\eta , \lambda,  h }_{\Lambda} \left( \sigma \right) := H_{\Lambda} \left( \sigma \right) - \sum_{v \in \Lambda} \left( \lambda \eta_{v} + h \right) \cdot \sigma_v.
\end{equation}
As we work with nearest-neighbor systems, it suffices to specify the boundary condition on the external boundary $\partial \Lambda$ of the domain $\Lambda$ and we will do so in the sequel (instead of specifying the boundary condition on $\Zd \setminus \Lambda$). Given an inverse temperature $\beta >0$ and a boundary condition $\tau \in \S^{\partial \Lambda}$, we define the finite-volume Gibbs measure
\begin{equation}\label{eq:Gibbs measure general setup continuous}
    \mu_{\beta , \Lambda, \tau}^{\eta , \lambda, h} \left( d\sigma \right) :=  \frac 1{Z_{\beta , \Lambda, \tau}^{\eta, \lambda, h}} \exp \left( - \beta  H^{\eta, \lambda, h}_{\Lambda}(\sigma) \right) \prod_{v \in \Lambda} \kappa(d \sigma_v)\prod_{v\in\Zd\setminus\Lambda}\delta_{\tau_v} \left( d\sigma_v \right),
\end{equation}
where $Z_{\beta , \Lambda, \tau}^{\eta, \lambda,h}$ is the normalization constant. We denote by $\left\langle \cdot \right\rangle_{\Lambda}^{\tau,h}$ the expectation with respect to the measure $\mu_{\beta , \Lambda, \tau}^{\eta, \lambda,h}$ (omitting the parameters $\beta, \eta$ and $\lambda$ in the notation as was done in Section~\ref{sec:2d systems}).

We remark that the setting is sufficiently flexible to include periodic boundary conditions (i.e., taking $\Lambda$ to be the discrete torus), as well as free boundary conditions. These are referred to by the notations per and free.
\end{enumerate}

The following results show that, in dimensions $1\le d\le 3$, the spatially and thermally averaged magnetization is close to zero when the external field $h$ is close to zero, uniformly in the boundary condition. The results depend on $L$ through power laws. The first statement applies to deterministic boundary conditions, and has a better exponent in the power law when compared with the second statement, which is uniform in the (possibly random) boundary conditions (see Theorem~\ref{prop:subcritical2} for more details on the case that $h$ is large).

We use the notation $a \vee b  := \max (a , b)$ and $a \wedge b := \min \left(a , b \right)$ for $a , b \in \R$.

\begin{theorem} \label{prop:subcritical}
Let $n\ge 2$, $d\in\{1,2,3\}$ and $L \geq 1$. Let $\beta > 0$ be the inverse temperature, $\lambda>0$ be the disorder strength and $h \in \R^n$ be the deterministic external field. There exists a constant $C>0$ depending only on $n$, $\Psi$ and $\lambda$ such that, then for each boundary condition $\tau \in \S^{\partial \Lambda_{2L}}$ (allowing also free and periodic boundary conditions) and each $h$ satisfying $|h| \le L^{-2}$,
\begin{equation} \label{eq:TV0904Final}
   \left| \E \left[\frac{1}{\left| \Lambda_L \right|}  \sum_{v \in  \Lambda_L } \left\langle \sigma_v \right\rangle^{\tau , h}_{\Lambda_{2L}} \right] \right|  \leq C L^{-\frac{1}{2}(4-d)}.
\end{equation}
Moreover, for any $h$ satisfying $|h| \le L^{-1}$,
\begin{equation} \label{eq:TV090404final}
   \E \left[ \sup_{\tau \in \S^{\partial \Lambda_L}} \left|\frac{1}{\left| \Lambda_L \right|}  \sum_{v \in  \Lambda_L} \left\langle \sigma_v \right\rangle^{\tau , h}_{\Lambda_L} \right| \right]   \leq C L^{-\frac{4-d}{2(8-d)}}.
\end{equation}
\end{theorem}
The next result presents our quantitative estimates in dimension $d=4$.
\begin{theorem} \label{thm.thm2}
Let $n\ge 2$ and $d = 4$. Let $\beta > 0$ be the inverse temperature, $\lambda>0$ be the disorder strength and $h \in \R^n$ with $|h|\le 1$ be the deterministic external field. There exists a constant $C>0$ depending only on $n$, $\Psi$ and $\lambda$ such that for each integer $L\ge 3$,
\begin{equation} \label{eq:11090401}
     \E \left[ \sup_{\tau \in  \S^{\partial \Lambda_L}}  \left|\frac{1}{\left|\Lambda_L \right|} \sum_{v \in \Lambda_L}\left\langle  \sigma_v \right\rangle^{\tau , h}_{\Lambda_L} \right| \right] \leq \frac{C}{\sqrt{\ln \ln \left( |h|^{-1} \wedge L  \right)}}.
\end{equation}
\end{theorem}
We again remark that, in the second part of Theorem~\ref{prop:subcritical} and Theorem~\ref{thm.thm2}, free and periodic boundary conditions are allowed as one of the options in the supremum.

We mention that versions of Theorem~\ref{prop:subcritical} and~\ref{thm.thm2} hold also at zero temperature. To state the results, we introduce the following notations. Given $L \in \N$ and a boundary condition $\tau \in \S^{\partial \Lambda_L}$, we define the zero-temperature free energy by
\begin{equation*}
    \FE^{\tau , \lambda, h}_{\Lambda, T=0} (\eta):= \inf_{\substack{\sigma': \Lambda_L^+ \to \S \\ \sigma' \equiv \tau \, \mathrm{on} \, \partial \Lambda_L}} H ^{\eta , \lambda,  h }_{\Lambda} \left( \sigma' \right).
\end{equation*}
We then define the set $G_L^{\tau,\eta}$ of $L$-ground configurations with boundary condition $\tau$ as the set of configurations $\sigma : \Lambda_L^+ \to \S$ satisfying $\sigma \equiv \tau$ on $\partial \Lambda_L$ and
\begin{equation*}
    H^{\eta , \lambda,  h }_{\Lambda} \left( \sigma \right) = \FE^{\tau , \lambda, h}_{\Lambda, T=0} (\eta).
\end{equation*}
The continuity of the disordered Hamiltonian and the compactness of the spin space ensure that the set $G_L^{\tau, \eta}$ is non-empty for all realizations of the random field $\eta$. Additionally, we remark that the cardinality of the set $G_L^{\tau, \eta}$ is almost surely equal to $1$. Indeed, this result is a consequence of the following observation: since the map $\eta \mapsto \FE^{\tau , \lambda, h}_{\Lambda, T=0} (\eta)$ is concave, it is differentiable almost everywhere, and, on the corresponding set of full measure, the $L$-ground configuration with boundary condition $\tau$ is uniquely characterised as $-\frac{1}{\lambda}$ times the $\eta$-gradient of $\FE^{\tau , \lambda, h}_{\Lambda, T=0}$.

The zero-temperature version of the results is then obtained by replacing the mapping $\eta \mapsto \left\langle \sigma_v \right\rangle^{\tau , h}_{\Lambda_{2L}}$ in~\eqref{eq:TV0904Final} by the map $\eta \mapsto \sigma_v$ with $\sigma \in G_{2L}^{\tau, \eta}$. In~\eqref{eq:TV090404final} and~\eqref{eq:11090401}, the supremum over boundary conditions should be replaced by a supremum over all $L$-ground configurations $\sigma \in \cup_{\tau \in \S^{\partial \Lambda_L}} G_L^{\tau, \eta}$ and $\left\langle \sigma_v \right\rangle^{\tau , h}_{\Lambda_{L}}$ should be replaced by $\sigma_v$.

Lastly, we mention that a discussion of models with continuous symmetries of higher order appears in Section~\ref{sectioncommentsandopenproblems}.

\section{Background}\label{relatedresults}

This section presents a brief overview of related results.

\emph{The random-field Ising model:} The question of the quantification of the Imry-Ma phenomenon was previously addressed for the two-dimensional random-field Ising model. Chatterjee~\cite{C18} obtained an upper bound at the rate $(\ln \ln L)^{-1/2}$ on the effect of boundary conditions on the magnetization at the center of a box of side length~$L$. Aizenman and the third author~\cite{AP19} studied the same quantity for finite-range random-field Ising models and obtained an algebraic upper bound of the form $L^{-\gamma}$. Exponential decay was established first at high temperature or strong disorder; results in this direction include the ones of Fr\"{o}hlich and Imbrie~\cite{FI84}, Berreti~\cite{B85}, Von Dreifus, Klein and Perez~\cite{DKP95}, and Camia, Jiang and Newman~\cite{CJN18}. Exponential decay at all disorder strengths was established for the nearest-neighbor model by Ding and Xia~\cite{ding20192exponential} at zero temperature, and then extended to all temperatures by Ding and Xia~\cite{ding2019exponential} and by Aizenman, the second and third authors~\cite{AHP20}. Other important aspects of the random-field Ising model which have been the subject of recent developments include the correlation length of the model, which has been successfully identified by Ding and Wirth~\cite{ding2023correlation}, and the absence of replica symmetry breaking which was established by Chatterjee~\cite{chatterjee2015absence}. We further refer to the recent work of Chatterjee~\cite{chatterjee2023features} which investigates the presence, or the absence, of various properties of the mean-field spin glass model in the random-field Ising model (such as replica symmetry breaking, non-self averaging, ultrametricity, and the existence of many pure states).

Bricmont and Kupiainen~\cite{BK88H} introduced a hierarchical approximation to the random-field Ising model and proved that it exhibits spontaneous magnetization in three dimensions and no magnetization in two dimensions. Imbrie~\cite{I85} proved that the three-dimensional ground state of the Ising model in a weak magnetic field exhibits long range order, and this was extended to the low-temperature regime by Bricmont and Kupiainen~\cite{BK87L, BK88} using a rigorous renormalization group argument. Recently, Ding and Zhuang~\cite{DingZhuang} obtained a new and simple proof of the result of~\cite{BK87L, BK88}, extended the result to the $q$-states Potts model, and obtained a lower bound for the correlation length of the two-dimensional random-field Ising model at low temperature matching the one of~\cite{ding2023correlation}. These results were further extended by Ding, Liu, Xia~\cite{ding2022long} who proved that the model exhibits long range ordering in three and higher dimensions for any inverse temperature $\beta > \beta_c$ (where $\beta_c$ is the critical inverse temperature without disorder) as long as the strength of the disorder is sufficiently weak, and by Affonso, Bissacot and Maia~\cite{affonso2023phase} to the long-range Ising model.

Recently in~\cite{BS102}, Bowditch and Sun constructed a continuum version of the two-dimensional random-field Ising model by considering the scaling limit of the suitably renormalized discrete random-field Ising model, and proved that the law of the limiting magnetization field is singular with respect to the one of the continuum pure two-dimensional Ising model constructed and studied by Camia, Garban and Newman~\cite{CGN161, CGN162}.

Additional background on the random-field Ising model can be found in~\cite[Chapter 7]{Bo06} and~\cite{N98T}.

\smallskip
\emph{Random-field induced order in the $XY$-model:} Crawford~\cite{C11,crawford2013random,crawford2014random} studied the $XY$-model in the presence of a random field pointing along the $Y$-axis, and proved that this may lead to the model exhibiting residual ordering along the $X$-axis. We additionally refer to the references therein for a review of the works in the physics literature investigating this phenomenon in different models of statistical physics.

\smallskip
\emph{Other models of statistical physics:} The Imry--Ma phenomenon has also been studied in the context of random surfaces. This line of investigation was initiated by Bovier and K\"{u}lske~\cite{BK94, BK96} with later, qualitative and quantitative, contributions of Cotar, van Enter, K\"{u}lske and Orlandi~\cite{KO06, KO08, VK08, CK12, CK15}. Additional quantitative results are obtained in~\cite{DHP21} which studies ``random-field random surfaces'' of the form
\begin{equation} \label{mod.randomsurface}
    \P(d \phi) := \frac{1}{Z} \exp \left( - \sum_{x \sim y} V \left( \phi_x - \phi_y \right) + \lambda \sum_{x}  \eta_{x} \phi_x \right) \prod_x d \phi_x ,
\end{equation}
where in one case $\phi$ is a mapping defined on the lattice and valued in $\R$, the map $V : \R \to \R$ is a uniformly convex potential, and $d \phi_x$ denotes the Lebesgue measure and in another case $\phi$ is valued in $\Z$, $V(x)=x^2$ and $d \phi_x$ denotes the counting measure. We refer to~\cite{DHP21} for a more detailed review of this line of investigation.

A version of the rounding effect was also studied in models of directed polymers by Giacomin and Toninelli~\cite{G06S, GT06smoot}. They established that, while the non-disordered model may exhibit a first- (or higher-) order phase transition, this transition is at least of second order in the presence of a random disorder. Nevertheless, the mechanism taking place is of different nature than the one underlying the Imry--Ma phenomenon~\cite{G06S}.

\smallskip
\emph{Quantum systems:} The arguments developed in~\cite{AW89} were extended by Aizenman, Greenblatt and Lebowitz~\cite{AGL12} to quantum lattice systems. They established that, at all temperatures, the first-order phase transition of these systems is rounded by the addition of a random field in dimensions $d \leq 2$, and in every dimensions $d \leq 4$ for systems with continuous symmetry.

\section{Strategy of the proof} \label{Secstratproof}
In this section, we outline the proofs of Theorem~\ref{thm1prop2.31708main}, Theorem~\ref{prop:subcritical} and Theorem~\ref{thm.thm2}, which contain the main ideas developed in the article. To simplify the presentation of the arguments, we assume that $m = 1$ in the definition of the noised observable and of the random disorder, that the strength $\lambda$ of the random field is equal to $1$, and that the maps $(f_v)$ have range $0$ (i.e., $f_v$ depends only on the value of the spin $\sigma_v$).

\medskip

\subsection{Outline of the proof of Theorem~\ref{thm1prop2.31708main}} \label{subsecoutkinemain} The proof relies on a thermodynamic approach and requires the introduction of the finite-volume free energy of the system as follows: for each side length $L \geq 2$, each boundary condition $\tau \in \S^{\Zd \setminus \Lambda_L}$, inverse temperature $\beta >0$, and magnetic field $\eta : \Lambda_L \to \R$, we define the finite-volume free energy to be the suitably renormalized logarithm of the partition function
\begin{equation} \label{eq:16110501}
    \FE^{\tau}_{\Lambda_L}(\eta)  := -\frac{1}{\beta \left| \Lambda_L \right|} \ln Z_{\beta , \Lambda_L, \tau}^{\eta ,1}.
\end{equation}
The proofs then rely on the following standard observations:
\begin{itemize}
    \item By decomposing the random field according to $\eta = \left( \hat \eta_L , \eta_L^\perp \right)$ with $\hat \eta_L := \left| \Lambda_L \right|^{-1} \sum_{v \in \Lambda_L} \eta_v$ and $\eta^\perp_L := \eta - \hat \eta_L$, we see that the
    the observable $ \left| \Lambda_L \right|^{-1} \sum_{v \in \Lambda_L} \left\langle f_v \left( \sigma \right) \right\rangle_{\Lambda_L}^{\tau}$ can be characterized as the derivative of the free energy with respect to the averaged external field $\hat \eta_L$, i.e.,
    \begin{equation*}
        \frac{\partial}{\partial \hat \eta_L} \FE^{\tau}_{\Lambda_L}(\eta) = - \frac{1}{\left| \Lambda_L \right|} \sum_{v \in \Lambda_L} \left\langle f_v \left( \sigma \right) \right\rangle_{\Lambda_L}^{\tau}.
    \end{equation*}
    \item For each realization of $\eta_L^\perp$, the mapping $\hat \eta_L \mapsto  -\FE^{\tau}_{\Lambda_L}(\hat \eta_L , \eta^\perp_L)$ is convex, differentiable and $1$-Lipschitz. Moreover, for any pair of boundary conditions $\tau_1 , \tau_2 \in \S^{\Zd \setminus \Lambda_L}$, the finite-volume free energy satisfies the relation
    \begin{equation} \label{eq:07170111}
        \left| \FE^{\tau_1}_{\Lambda_L}(\eta) - \FE^{\tau_2}_{\Lambda_L}(\eta) \right| \leq \frac{C}{L}.
    \end{equation}
    \item The averaged external field $\hat \eta_L$ is a Gaussian random variable whose variance is equal to $\left| \Lambda_L \right|^{-1}$. In two dimensions, its fluctuations are of order~$L^{-1}$, which is the same order of magnitude as the right-hand side of~\eqref{eq:07170111}.
\end{itemize}

The key input of the argument is the following general fact: for each $1$-Lipschitz convex and differentiable function $g : \R \to \R$ and each $\delta > 0$, there exists a set $A_g \subseteq \R$ satisfying the two following properties:
\begin{enumerate}
    \item The Lebesgue measure of the set $\R \setminus A_g$ is finite and satisfies $\Leb \left( \R \setminus A_g \right) \leq C/\delta^2$.
    \item For each point $x \in A_g$, and each differentiable, $1$-Lipschitz and convex function $g_1 : \R \to \R$ satisfying $\sup_{t \in \R}\left| g(t) - g_1(t) \right| \leq 1$, one has
    \begin{equation*}
        \left| g'(x) - g_1'(x) \right| \leq \delta.
    \end{equation*}
\end{enumerate}
This observation is quantified through the notion of $\delta$-stability introduced in Section~\ref{section3.1}.

Applying this property to the free energies $\hat \eta_L \mapsto -\FE^{\tau}_{\Lambda_L}(\hat \eta_L , \eta^\perp_L)$, for a fixed boundary condition~$\tau \in \mathcal{S}^{\Zd \setminus \Lambda_L}$, using the inequality~\eqref{eq:07170111}, that the averaged field $\hat \eta_{\Lambda_L}$ is Gaussian and that its variance is of order $L^{-2}$, we obtain, for any $\delta > 0$,
\begin{equation} \label{eq:155625}
    \P \left(\sup_{\tau_1, \tau_2 \in \S^{\Zd \setminus \Lambda_L}} \frac{1}{|\Lambda_L|} \sum_{v \in \Lambda_L} \left( \left\langle f_v \left( \sigma \right) \right\rangle_{\Lambda_L}^{\tau_1}  -  \left\langle f_v \left( \sigma \right) \right\rangle_{\Lambda_L}^{\tau_2} \right) < \delta \right) \geq c_\delta,
\end{equation}
where the constant $c_\delta$ depends only on $\delta$ and satisfies $c_\delta \geq e^{-C/\delta^4}$.

The lower bound stated in~\eqref{eq:155625} is weaker than the statement of Theorem~\ref{thm1prop2.31708main}. The strategy to upgrade the inequality~\eqref{eq:155625} into the quantitative estimate~\eqref{eq:11191006} relies on the observation that the inequality~\eqref{eq:155625} is scale invariant in two dimensions (the constant $c_\delta$ does not depend on $L$). One can thus implement a hierarchical decomposition of the space and leverage on this property to improve the result. Specifically, we implement a Mandelbrot percolation argument which allows to cover (almost completely) the box $\Lambda_L$ by a collection $\Q$ of disjoint boxes (see Figure~\ref{partitionempty}), all of which satisfy the property
\begin{equation} \label{eq:160925}
    \sup_{\tau_1, \tau_2 \in \S^{\Zd \setminus \Lambda}} \frac{1}{|\Lambda_L|} \sum_{v \in \Lambda} \left( \left\langle f_v \left( \sigma \right) \right\rangle_{\Lambda}^{\tau_1}  -  \left\langle f_v \left( \sigma \right) \right\rangle_{\Lambda}^{\tau_2} \right) \leq \frac{C}{\sqrt[4]{\ln \ln L}}.
\end{equation}
Once this is achieved, an application of the domain subadditivity property for the left-hand side of~\eqref{eq:160925} (see Section~\ref{sec.domainMarkov}) yields the bound
\begin{equation} \label{eq:08.15}
     \P \left( \sup_{\tau_1, \tau_2 \in \S^{\Zd \setminus \Lambda_L}} \frac{1}{|\Lambda_L|} \sum_{v \in \Lambda_L} \left( \left\langle f_v \left( \sigma \right) \right\rangle_{\Lambda_L}^{\tau_1}  -  \left\langle f_v \left( \sigma \right) \right\rangle_{\Lambda_L}^{\tau_2} \right) \geq \frac{C}{\sqrt[4]{\ln \ln L}}  \right) \leq  \exp \left( - c \sqrt{\ln L} \right).
\end{equation}
The inequality~\eqref{eq:08.15} is slighlty weaker than Theorem~\ref{thm1prop2.31708main}, since the stochastic integrability (the right-hand side of~\eqref{eq:08.15}) can be improved. This is achieved by using a concentration argument implemented in Section~\ref{sectionconcentrationatg}.

We complete this outline with a remark regarding the rate of convergence obtained: it is related to the lower bound $c_\delta$ through the formula
\begin{equation*}
   \inf \left\{ \delta \in (0 , 1) \, : \, c_{\delta} := \frac{1}{\ln L} \right\}.
\end{equation*}
The result of Lemma~\ref{L.lemma2.5} gives the value $c_{\delta} := \exp \left( - \frac{C}{\delta^4} \right)$, which yields the rate $1/\sqrt[4]{\ln \ln L}$.

\medskip

\subsection{Outline of the proofs of Theorem~\ref{prop:subcritical} and Theorem~\ref{thm.thm2}.}
To simplify the presentation of the argument, we make the additional assumption that $h = 0$. The proofs of Theorem~\ref{prop:subcritical} and Theorem~\ref{thm.thm2} rely on a Mermin-Wagner type argument (see~\cite{MW}) which allows to use the continuous symmetry of the model to upgrade the inequality~\eqref{eq:07170111}: the upper bound we obtain is stated in Proposition~\ref{propMermW} and (a simplified version of it) reads, for any $i \in \{ 1 , \ldots, n\}$
\begin{equation} \label{eq:09490111}
    \E \left[ \FE^{\tau, 0}_{\Lambda_L}(\eta) -\FE^{\tau, 0}_{\Lambda_L}(-\eta)  ~ \Big\vert ~  \eta_{ \Lambda_{L/2}, i }\right] \leq \frac{C}{L^2},
\end{equation}
where the free energy is defined in~\eqref{eq:17053112} below, the expectation in the left-hand side of~\eqref{eq:09490111} is the conditional expectation with respect to the values of the $i$-th coordinate of the random field $\eta$ inside the box $\Lambda_{L/2}$ (see Section~\ref{sec.condexpect}).

The main features of the inequality~\eqref{eq:09490111} are the following:
\begin{itemize}
    \item The right-hand side of~\eqref{eq:09490111} decays like $C/L^2$, in contrast to the upper bound $C/L$ of~\eqref{eq:07170111} in the case of general spin systems.
    \item The derivative with respect to the averaged field $\hat \eta_{\Lambda_{L/2},i} := \left|  \Lambda_{L/2}\right|^{-1} \sum_{v \in \Lambda_{L/2}} \eta_{v,i} $ of the left-hand side of~\eqref{eq:09490111} is explicit and satisfies the equality
    \begin{equation} \label{eq:10310111}
         \frac{\partial}{\partial \hat \eta_{L/2,i}} \E \left[ \FE^{\tau, 0}_{\Lambda_L}(\eta) - \FE^{\tau, 0}_{\Lambda_L}(- \eta) ~ \Big\vert ~  \eta_{\Lambda_{L/2,i} }\right]  =  -\frac{1}{\left| \Lambda_L \right|} \sum_{v \in \Lambda_{L/2}} \E \left[  \left\langle \sigma_{v,i} \right\rangle_{\Lambda_L}^{\tau, 0}(\eta) + \left\langle \sigma_{v,i} \right\rangle_{\Lambda_L}^{\tau, 0}(- \eta) ~ \Big\vert ~  \eta_{\Lambda_{L/2},i }\right].
    \end{equation}
    In particular, using the $\eta \to -\eta$ invariance of the law of the random field, we see that the expectation of the right-hand side of~\eqref{eq:10310111} is equal to $-2\E \left[ \left| \Lambda_L \right|^{-1} \sum_{v \in \Lambda_{L/2}} \left\langle \sigma_{v,i} \right\rangle_{\Lambda_L}^{\tau, 0} \right]$. The expectation of the spatially and thermally averaged magnetization can thus be characterized as the expectation of the derivative with respect to the variable $\hat \eta_{\Lambda_{L/2},i}$ of the left-hand side of~\eqref{eq:09490111}.
\end{itemize}
In dimensions $d \leq 3$, the fluctuations of the averaged field $\hat \eta_{\Lambda_{L/2},i}$ are of order $L^{-d/2}$, and are thus larger than the right-hand side of~\eqref{eq:09490111}. Combining this observation with a variational principle (see Section~\ref{sec.varlemmaCSS}) yields the algebraic rate of convergence stated in Theorem~\ref{prop:subcritical}. In dimension $4$, the fluctuations of the averaged field $\hat \eta_{\Lambda_{L/2},i}$ are of the same order of magnitude as the right-hand side of~\eqref{eq:09490111}, and we implement a Mandelbrot percolation argument similar to the one presented in Section~\ref{subsecoutkinemain} to obtain the result.

We point out that there is a distinct difference between the two constructions. In the presence of a continuous symmetry, we do not rely on the criterion~\eqref{eq:160925} to define which box should belong to the partition $\mathcal{Q}$ but rather on the inequality, for a box $\Lambda \subseteq \Lambda_L$,
\begin{equation*}
    \left| \E \left[ \frac{1}{\left| \Lambda \right|} \sum_{v \in \Lambda} \left\langle  \sigma_v \right\rangle_{\Lambda_L}^{\tau, 0} ~ \Big\vert ~ \hat \eta_{\Lambda} \right] \right| \leq \frac{C}{\sqrt{\ln \ln L}},
\end{equation*}
where the conditional expectation is taken with respect to the averaged field $\hat \eta_\Lambda := \left|  \Lambda\right|^{-1} \sum_{v \in \Lambda} \eta_v$. This difference accounts for the slightly better rate of convergence obtained in Theorem~\ref{thm.thm2}: we obtain the rate $1/\sqrt{\ln \ln L}$ instead of the rate $1/\sqrt[4]{\ln \ln L}$ in Theorem~\ref{thm1prop2.31708main}.

\section{Notations, assumptions and preliminaries} \label{section2not}

\subsection{General notation}

For each real number $a \in \R$, we use the notation $\lfloor a \rfloor$ to refer to the integer part of $a$. A box is a subset of the form $v + \Lambda_L$ for $v \in \Zd$ and $L \in \N$. We call the vertex $v$ the center of the box $\Lambda := v + \Lambda_L$ and the integer $(2L+1)$ its side length. Given a box $\Lambda \subseteq \Zd$ of side length $L$ and a real number $\alpha > 0$, we denote by $\alpha \Lambda$ the box with the same center as $\Lambda$ and side length $\lfloor \alpha L \rfloor$, and define $\Lambda_{\alpha L} := \alpha \Lambda_L$. Given a set $\Lambda \subseteq \Zd$ and a vertex $v \in \Zd$, we denote by $\dist (v , \Lambda) := \min_{w \in \Lambda} |v - w|.$ For each set $\Lambda \subseteq \Zd$, we denote by $\indc_\Lambda$ the indicator function of $\Lambda$.

Given a function $g$ defined on either $\R$, a subset $\Lambda \subseteq \Zd$ or the set of configurations and valued in $\R^m$, we denote by $g_1 , \ldots, g_m$ its components. In particular, we denote by $f_{v , 1}, \ldots, f_{v , m}$ the components of the observable $f_v$, and, in the case of continuous spin systems studied in Section~\ref{SectionCSS}, we denote by $\sigma_{v, 1}, \ldots,\sigma_{v, n}$ the components of the spin $\sigma_v$.

For each bounded set $\Lambda \subseteq \Zd$ and each fixed boundary condition $\tau_0 \in \S^{\Zd \setminus \Lambda}$, we denote by
\begin{equation} \label{NotaFluc}
    \fluc_{\Lambda}(\eta) := \sup_{\tau_1 , \tau_2 \in \S^{\Zd \setminus \Lambda}}  \left| \frac{1}{|\Lambda|} \sum_{v \in \Lambda} \left( \left\langle f_v \left( \sigma \right) \right\rangle_{\Lambda}^{\tau_1}  -  \left\langle f_v \left( \sigma \right) \right\rangle_{\Lambda}^{\tau_2} \right) \right|,
\end{equation}
and
\begin{equation} \label{NotaFlucper}
    \fluc_{\Lambda}^{\tau_0}(\eta) := \sup_{\tau \in \S^{\Zd \setminus \Lambda}}  \left| \frac{1}{|\Lambda|} \sum_{v \in \Lambda} \left( \left\langle f_v \left( \sigma \right) \right\rangle_{\Lambda}^{\tau}  -  \left\langle f_v \left( \sigma \right) \right\rangle_{\Lambda}^{\tau_0} \right) \right|.
\end{equation}
Similarly, for each $i \in \{ 1 , \ldots, m\}$, we define
\begin{equation} \label{NotaFluci}
    \fluc_{\Lambda,i}(\eta) := \sup_{\tau_1 , \tau_2 \in \S^{\Zd \setminus \Lambda}} \left| \frac{1}{|\Lambda|}  \sum_{v \in \Lambda} \left( \left\langle f_{v,i} \left( \sigma \right) \right\rangle_{\Lambda}^{\tau_1}  -  \left\langle f_{v,i} \left( \sigma \right) \right\rangle_{\Lambda}^{\tau_2} \right) \right|,
\end{equation}
and
\begin{equation} \label{NotaFlucperi}
    \fluc_{\Lambda,i}^{\tau_0}(\eta) := \sup_{\tau \in \S^{\Zd \setminus \Lambda}}
    \left| \frac{1}{|\Lambda|}  \sum_{v \in \Lambda} \left( \left\langle f_{v,i} \left( \sigma \right) \right\rangle_{\Lambda}^{\tau}  -  \left\langle f_{v,i} \left( \sigma \right) \right\rangle_{\Lambda}^{\tau_0} \right) \right|.
\end{equation}
Let us note that these quantities only depend on the value of the field inside the set $\Lambda$, and that we have the inequalities
\begin{equation} \label{eq:11250211}
    \fluc_{\Lambda}(\eta) \leq  \sum_{i = 1}^m \fluc_{\Lambda,i}(\eta), \hspace{3mm} \fluc_{\Lambda}(\eta) \leq 2 \fluc_{\Lambda}^{\tau_0}(\eta) \hspace{3mm} \mbox{and} \hspace{3mm} \fluc_{\Lambda,i}(\eta) \leq 2 \fluc_{\Lambda,i}^{\tau_0}(\eta).
\end{equation}
Additionally, by the pointwise bound $\left| f_v(\sigma) \right| \leq 1$, the four quantities~\eqref{NotaFluc},~\eqref{NotaFlucper},~\eqref{NotaFluci} and~\eqref{NotaFlucperi} are bounded by $2$ for any realization of the random field $\eta$.

\subsection{Structure of the random field} \label{secstructRF}
Given an integer $i \in \{1 , \ldots, m \}$ and a vertex $v \in \Zd$, we denote by $\eta_i$ and $\eta_{v , i}$ the $i$-th component of $\eta$ and $\eta_v$ respectively.

\smallskip

Given a bounded set $\Lambda \subseteq \Zd$, we denote by $\eta_\Lambda = (\eta_{\Lambda , 1} , \ldots , \eta_{\Lambda , m})$ the restriction of the field $\eta$ to the set $\Lambda$. We use the decomposition $\eta_\Lambda := (\hat \eta_\Lambda , \eta^\perp_\Lambda)$ with $$\hat \eta_\Lambda := \left| \Lambda \right|^{-1} \sum_{v \in \Lambda} \eta_v \hspace{3mm} \mbox{and} \hspace{3mm}\eta^\perp_\Lambda := \eta_{\Lambda} - \hat \eta_\Lambda.$$

\smallskip

More generally, given a map $w : \Lambda \to \R^m$, we use the decomposition $\eta_\Lambda := (\hat \eta_{w,\Lambda} , \eta^\perp_{w,\Lambda})$ with, for each $i \in \{ 1 , \ldots , m\}$, $$\hat \eta_{w ,\Lambda, i} := \frac{1}{\sum_{v \in \Lambda} w_i(v)^2} \sum_{v \in \Lambda} w_i(v) \eta_{v,i} \hspace{3mm} \mbox{and} \hspace{3mm} \eta^\perp_{w,\Lambda,i} := \eta_{\Lambda,i} - \hat \eta_{w ,\Lambda, i} w_i.$$
We set $\hat \eta_{w ,\Lambda, i}=0$ if $w_i = 0$.

\smallskip

Since the field $\eta$ is assumed to be a standard $m$-dimensional Gaussian vector, all the random fields listed above are Gaussian. For each $i \in \{ 1 , \ldots, m \}$, the variances of $\hat \eta_{\Lambda, i}$ and $\hat \eta_{w, \Lambda , i}$ are equal $\left| \Lambda \right|^{-1}$ and $1/\left( \sum_{v \in \Lambda} w_i(v)^2\right)$ respectively. Moreover the fields $\hat \eta_\Lambda$ and $\eta^\perp_\Lambda$ are independent, and the fields $\hat \eta_{w,\Lambda}$ and $\eta^\perp_{w,\Lambda}$ are independent.

\smallskip

Given a function $F$ depending on the realization of $\eta$ in the set $\Lambda$ and an integer $i \in \{ 1 , \ldots, m \}$, we may abuse notation and write $F(\hat \eta_\Lambda , \eta^\perp_\Lambda)$, $F(\hat \eta_{\Lambda , i} , \left( \hat \eta_{\Lambda , j} \right)_{j \neq i} , \eta^\perp_\Lambda)$ or $F(\hat \eta_{w,\Lambda} , \eta^\perp_{w,\Lambda})$ instead of $F(\eta)$ when we want to emphasize the dependence of the map $F$ on a specific variable.

\smallskip

For $i \in \{1 , \ldots, m \}$, we denote by $\frac{\partial F}{\partial \hat \eta_{\Lambda,i}}(\eta)$ and $\frac{\partial F}{\partial \hat \eta_{w,\Lambda,i}}(\eta)$ the partial derivatives of the map $F$ with respect to the variables $\hat \eta_{\Lambda,i}$ and $\hat \eta_{w, \Lambda,i}$ respectively.

\subsection{Notation for conditional expectation} \label{sec.condexpect} Assume that the random field $\eta$ is defined on a probability space $(\Omega, \mathcal{F}, \mathcal{P})$, and that we are given two random variables $F$ and $X$ depending on the random field $\eta$, such that $\E \left[ |F| \right] < \infty$,  and a $\sigma$-algebra $\mathcal{F}_1 \subseteq \mathcal{F}$. We denote by $\E \left[ F \, | \, X\right]$ (resp. $\E \left[ F \, | \, \mathcal{F}_1\right]$) the conditional expectation of $F$ with respect to $X$ (resp. with respect to $\mathcal{F}_1$). Given an event $A$, we denote by $\P\left( A \, | \, X\right) := \E \left[ \indc_{A} \, | \, X \right]$ (resp. $\P\left( A \, | \, \mathcal{F}_1\right) := \E \left[ \indc_{A} \, | \, \mathcal{F}_1 \right]$) the conditional probability. The random variable $X$ will frequently be the fields $\hat \eta_\Lambda$, $\eta^\perp_\Lambda$, $\hat \eta_{w,\Lambda}, \eta^\perp_{w,\Lambda}$, the restriction of the field $\eta$ to some set $\Lambda$, or a combination of these options.

Since the random variable $\E \left[ F \,|\, X\right]$ depends only on the realization of $X$, we may write $\E \left[ F \, | \, X\right] (X)$ instead of $\E \left[ F \,| \,X\right] (\eta)$ when we wish to make the dependence on the random field explicit. Similarly, we may write $\P \left( A \,|\, X\right)(X)$ instead of $\P \left( A \,|\, X\right)(\eta)$.

Finally, in the conditional expectations, we may omit to display the dependency in the random field to simplify the notation (e.g., we will typically write $ \E [ \FE_{\Lambda}^{\tau} ~\vert ~ {\hat\eta_{\Lambda}}]$ instead of $ \E [ \FE_{\Lambda}^{\tau}(\eta) ~ \vert ~ {\hat\eta_{\Lambda}} ]$).

\subsection{The domain subadditivity property} \label{sec.domainMarkov}
In this section, we state a domain subadditivity property satisfied by the quantity $\fluc_{\Lambda}(\eta)$. The result is a direct consequence of the consistency property~\eqref{eq:14361411} and is used frequently in the proofs of Sections~\ref{sectionDSS} and~\ref{SectionCSS}.

\begin{proposition}[Domain subadditivity property] \label{propdommarkov}
Let $\beta > 0$, $\eta : \Zd \to \R^m$. Let $\Lambda_1' , \ldots, \Lambda_N'$ of be a collection of disjoint bounded subsets of $\Zd$ and define $\Lambda' := \cup_{j = 1}^N  \Lambda_j'$. Then we have
\begin{equation*}
     \fluc_{\Lambda'}(\eta) \leq \sum_{j = 1}^N \frac{\left| \Lambda_j'\right|}{\left| \Lambda' \right|} \fluc_{\Lambda_j'}(\eta),
\end{equation*}
as well as, for any integer $i \in \{ 1 , \ldots, m\}$,
\begin{equation*}
    \sup_{\tau \in \S^{\Zd \setminus \Lambda'}} \sum_{v \in \Lambda'} \left\langle f_{v,i} \left( \sigma \right) \right\rangle_{\Lambda'}^{\tau} \leq \sum_{j = 1}^N \sup_{\tau \in \S^{\Zd \setminus \Lambda_j'}} \sum_{v \in \Lambda_j'} \left\langle f_{v,i} \left( \sigma \right) \right\rangle_{\Lambda_j'}^{\tau}
\end{equation*}
and
\begin{equation*}
    \inf_{\tau \in \S^{\Zd \setminus \Lambda'}} \sum_{v \in \Lambda'} \left\langle f_{v,i} \left( \sigma \right) \right\rangle_{\Lambda'}^{\tau} \geq \sum_{j = 1}^N \inf_{\tau \in \S^{\Zd \setminus \Lambda_j'}} \sum_{v \in \Lambda_j'} \left\langle f_{v,i} \left( \sigma \right) \right\rangle_{\Lambda_j'}^{\tau}.
\end{equation*}
\end{proposition}

\subsection{Free energy: definition and basic properties} \label{Secfreeenergy}

The finite-volume free energy of the random system is defined below (and corresponds to the one introduced in~\eqref{eq:16110501}).

\begin{definition}[Finite-volume free energy] \label{def.freeenergy}
    Given a bounded domain $\Lambda \subseteq \Zd$, an inverse temperature $\beta > 0$, a boundary condition $\tau \in \S^{\Zd \setminus \Lambda}$, a realization of the random field $\eta$, we define the free energy by the formula
    \begin{equation*}
        \FE_\Lambda^{\tau} \left( \eta \right) := -\frac{1}{\beta \left| \Lambda \right|} \ln Z_{\beta , \Lambda, \tau}^{\eta , \lambda} = -\frac{1}{\beta \left| \Lambda \right|} \ln \int_{\S^{\Lambda}}  \exp \left( - \beta  H^{\eta , \lambda}_{\Lambda}(\sigma) \right) \prod_{v \in \Lambda} \kappa(d \sigma_v),
    \end{equation*}
where the integral is computed over the set of configurations satisfying $\sigma = \tau$ in $\Zd \setminus \Lambda$.
\end{definition}

We next collect without proof some basic properties of the finite-volume free energy.

\begin{proposition}[Properties of the free energy] \label{prop.basicpropfreeen}
For any bounded domain $\Lambda \subseteq \Zd$, any inverse temperature $\beta >0,$ and any boundary condition $\tau \in \S^{\Zd \setminus \Lambda}$, the map $\FE_\Lambda^{\tau}$ satisfies the properties:
\begin{itemize}
    \item \textnormal{Concavity and regularity:} the map $\eta \mapsto \FE_\Lambda^{\tau}(\eta )$ is concave and satisfies for any pair of fields $\eta , \eta'$,
    \begin{equation*}
        \left| \FE_{\Lambda}^{\tau} \left( \eta \right) - \FE_{\Lambda}^{\tau} \left( \eta' \right) \right| \leq \frac{\lambda}{\left| \Lambda \right|}\sum_{v \in \Lambda} \left| \eta_v - \eta'_v \right|.
    \end{equation*}
    \item \textnormal{Derivative:} the mapping $\eta \mapsto \FE^{\tau}_\Lambda \left( \eta  \right)$ is differentiable and satisfies, for any $i \in \{ 1 , \ldots , m\}$,
    \begin{equation} \label{eq:02111210}
        \frac{\partial \FE^{\tau}_{\Lambda}}{\partial \hat \eta_{\Lambda,i}} (\eta) =  -\frac{\lambda}{\left| \Lambda \right|} \sum_{v \in \Lambda} \left\langle f_{v,i} \left( \sigma \right) \right\rangle_{\Lambda}^{\tau}.
    \end{equation}
    More generally, for any map $w : \Lambda \mapsto \R^m$ which is not identically $0$,
    \begin{equation} \label{eq:17510211}
        \frac{\partial\FE^{\tau}_{\Lambda} }{\partial \hat \eta_{w , \Lambda , i}} (\eta) =  -\frac{\lambda}{\left| \Lambda \right|} \sum_{v \in \Lambda} w_i(v) \left\langle f_{v,i} \left( \sigma \right) \right\rangle_{\Lambda}^{\tau}.
    \end{equation}
    \item \textnormal{Finite energy:} for any boundary condition $\tau_1 \in \S^{\Zd \setminus \Lambda}$, and any realization of random field $\eta$,
    \begin{equation} \label{eq:12030211}
    \left| \FE^{\tau}_\Lambda \left( \eta \right) - \FE^{\tau_1}_\Lambda \left( \eta \right) \right| \leq C\frac{ \left| \partial \Lambda \right|}{\left| \Lambda \right|} + \frac{C \lambda}{\left| \Lambda \right|} \sum_{\substack{v \in \Lambda \\ \dist(v , \partial \Lambda) \leq R}} \left| \eta_v \right|.
    \end{equation}
\end{itemize}
\end{proposition}

The proof of these results is a direct consequence of the formula for the free energy stated in Definition~\ref{def.freeenergy}, the assumption $|f_v(\sigma)| \leq 1$ on the noised observable, the fact that the map $f_v$ has range $R$ and the inequality~\eqref{def.cteCH}.

In the case of systems equipped with a continuous symmetry, we make the dependence in the parameter $h$ explicit, and refer to the free energy using the notation
\begin{equation} \label{eq:17053112}
    \FE_\Lambda^{\tau, h} \left( \eta \right) := -\frac{1}{\beta \left| \Lambda \right|} \ln Z_{\beta , \Lambda, \tau}^{\eta , \lambda,h} = -\frac{1}{\beta \left| \Lambda \right|} \ln \int_{\S^{\Lambda}}  \exp \left( - \beta  H^{\eta , \lambda,h}_{\Lambda}(\sigma) \right) \prod_{v \in \Lambda} \kappa(d \sigma_v).
\end{equation}
For any realization of the random field $\eta$, and any boundary condition $\tau \in \S^{\partial \Lambda},$ the mapping $h \mapsto \FE^{\tau,h}_\Lambda (\eta)$ is concave, $1$-Lipschitz, differentiable, and satisfies, for any $i \in \{1 , \ldots ,n \}$,
\begin{equation} \label{eq:14210811}
        \frac{\partial \FE^{\tau , h}_{\Lambda}(\eta)}{\partial h_{i}}  =  -\frac{1}{\left| \Lambda \right|} \sum_{v \in \Lambda}  \left\langle \sigma_{v,i} \right\rangle_{\Lambda}^{\tau, h} .
\end{equation}
As it will be used in the proofs below, we record that, since the identity~\eqref{eq:14210811} is valid for any deterministic field $\eta$ and any boundary condition $\tau \in \S^{\partial \Lambda}$, it implies the following result: for any random (measurable)  boundary condition $\eta \mapsto \tau (\eta) \in \S^{\partial \Lambda}$,
\begin{equation*}
    \frac{\partial }{\partial h_{i}} \E \left[ \FE^{\tau(\eta), h}_{\Lambda}(\eta)\right] = -\E \left[ \frac{1}{\left| \Lambda \right|} \sum_{v \in \Lambda} \left\langle \sigma_{v,i} \right\rangle_{\Lambda}^{\tau(\eta), h} \right].
\end{equation*}
Moreover, the continuous symmetry of the systems manifests in the following ways: using the $\eta \to -\eta$ invariance of the law of the random field, we see that, for any random boundary condition $\eta \mapsto \tau(\eta) \in \S^{\partial \Lambda}$,
\begin{equation} \label{eq:16560511}
    \forall v \in \Lambda, \hspace{5mm} \E \left[ \left\langle \sigma_v \right\rangle_{\Lambda}^{\tau(\eta) , h} \right] = -  \E \left[ \left\langle \sigma_v \right\rangle_{\Lambda}^{\tilde \tau(\eta) , -h} \right],
\end{equation}
where we used the notation $\tilde \tau : \eta \mapsto - \tau(-\eta) \in \S^{\partial \Lambda}$.
In the case of the periodic boundary condition, the expected value of the magnetization does not depend on the vertex $v$, i.e.,
\begin{equation} \label{eq:10240401}
    \forall v \in \Lambda, \hspace{5mm} \E \left[ \left\langle \sigma_v \right\rangle_{\Lambda}^{\per , h} \right] =  \E \left[ \left\langle \sigma_0 \right\rangle_{\Lambda}^{\per , h} \right],
\end{equation}
and when the magnetic field $h$ is equal to $0$, its value is equal to $0$,
\begin{equation} \label{eq:16570511}
    \forall v \in \Lambda, \hspace{5mm} \E \left[ \left\langle \sigma_v \right\rangle_{\Lambda}^{\per , 0} \right] = 0.
\end{equation}

\medskip

\subsection{Convention for constants} Throughout this article, the symbols $c$ and $C$ denote positive constants which may vary from line to line, with $C$ increasing larger than $1$ and $c$ decreasing smaller than $1$. These constants may depend only on the strength of the random field $\lambda$, the parameter $m$, the constant $C_H$ and the radius $R$. In the setup of spin system with continuous symmetry, they may depend on the strength of the random field $\lambda$, the dimension of the sphere $n$ and the map $\Psi$.

\section{Proofs for two-dimensional disordered spin systems} \label{sectionDSS}

In this section, we study the general spin systems presented in Section~\ref{sec:2d systems} and prove Theorem~\ref{thm1prop2.31708main}, Theorem~\ref{prop:1109}, Corollary~\ref{prop:1109cor} and Theorem~\ref{theoremuniqueness}.

In Section~\ref{section3.1}, we introduce the notion of $\delta$-stability for $\lambda$-Lipschitz convex function and quantify the Lebesgue and Gaussian measures of the $\delta$-stability set (see Proposition~\ref{Propdeltastab} and Corollary~\ref{coroGaussmea} below).

The next four sections are devoted to the proof of Theorem~\ref{thm1prop2.31708main} following the outline of Section~\ref{Secstratproof}: Section~\ref{section3.2LB} contains the proof of the estimate~\eqref{eq:155625}, in Section~\ref{sectionfirstmandelbrotperc}, we implement the Mandelbrot percolation argument and prove the inequality~\eqref{eq:08.15}, and in Section~\ref{sectionconcentrationatg} we complete the proof of Theorem~\ref{thm1prop2.31708main} (in the general case) by upgrading the stochastic integrability of~\eqref{eq:08.15}. Finally Section~\ref{eq:fluctuations around limiting value} is devoted to the proof of the inequality~\eqref{eq:fluctuations around limiting value} of Theorem~\ref{thm1prop2.31708main} (in the translation-invariant setup).

Three remaining sections (Sections~\ref{proofTHm2},~\ref{proofTHM3} and~\ref{proofofTheorem3}) are devoted to the proofs of Theorem~\ref{prop:1109}, Corollary~\ref{prop:1109cor} and Theorem~\ref{theoremuniqueness} respectively.

\subsection{A notion of $\delta$-stability for real-valued convex functions} \label{section3.1}

The following statement is a general result about real-valued convex functions; it asserts that if two $\lambda$-Lipschitz continuous, convex and differentiable functions are close (in the $L^{\infty}$-norm), then their derivatives cannot be too distant from each other on a set of large Lebesgue measure. We first introduce the following set.

\begin{definition}
Fix $\lambda > 0$. For each $\lambda$-Lipschitz, convex function $g : \R \to \R$, and each parameter $r > 0$, we define the set
\begin{multline*}
    N_{\lambda, r}(g) := \left\{ g_1 : \R \to \R \, : \, g_1~\mbox{is convex,}~\lambda\mbox{-Lipschitz continuous, differentiable} \right. \\ \left. \mbox{and satisfies}~ \sup_{t \in \R}\left| g(t) - g_1(t) \right| \leq r \right\}.
\end{multline*}
\end{definition}

We then define the $\delta$-stability set of the function $g$ as follows.

\begin{definition}[$\delta$-stability set]
For each triplet of parameters $ \lambda, \delta, r > 0$, and each function $g : \R \to \R$ convex, $\lambda$-Lipschitz and differentiable, we define the set
\begin{equation*}
    \mathrm{Stab}(\lambda, \delta, r , g) := \left\{ t \in \R \, : \, \exists g_1 \in N_{\lambda, r}(g) ~\mbox{such that}~ \left| g_1'(t) - g'(t) \right| > \delta \right\}.
\end{equation*}
\end{definition}

The next proposition estimates the Lebesgue measure of the set $\mathrm{Stab}(\lambda, \delta, r , g)$.

\begin{proposition}[$\delta$-stability for $\lambda$-Lipschitz convex functions] \label{Propdeltastab}

There exists a constant $C \geq 1$ such that, for each $\lambda > 0$, each function $g : \R \to \R$ convex, $\lambda$-Lipschitz and differentiable, and each pair of parameters $r , \delta > 0$,
\begin{equation} \label{eq:1147}
    \Leb \left( \mathrm{Stab}(\lambda, \delta, r , g) \right) \leq  \frac{C \lambda r}{\delta^2}.
\end{equation}
\end{proposition}

\begin{center}
\begin{figure}
\includegraphics[scale=0.8]{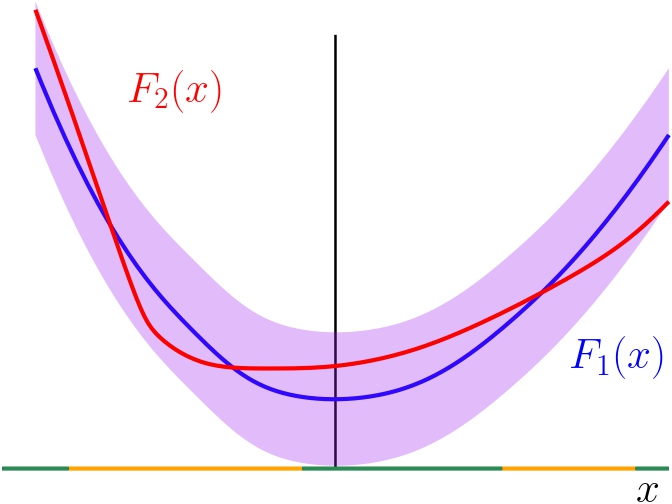}
\caption{The figure represents a $\lambda$-Lipshitz convex function $F_1$. The area in purple represents the surface where functions in the set $N_{\lambda,r}(F_1)$ must lie. An example of a function $F_2 \in N_{\lambda,r}(F_1)$ is drawn in red, and the set $\mathrm{Stab}(\lambda, \delta, r , g)$ is drawn in orange.}
\end{figure}
\end{center}

\begin{proof}
We fix three parameters $\lambda, \delta , r> 0$, a $\lambda$-Lipschitz, convex and differentiable function $g : \R \to \R$, and observe that if a point $t \in \R$ belongs to the set $\mathrm{Stab}(\lambda, \delta, r , g)$ then there exists a function $g_1 \in N_{\lambda, r}(g)$ (which may depend on the value of $t$), such that either:
\begin{enumerate}
    \item The inequality $g'_1(t) -  g'(t) > \delta$ holds;
    \item Or the inequality $ g'(t) -  g'_1(t) > \delta$ holds.
\end{enumerate}
Let us first assume that the inequality (1) is satisfied; we claim that it implies the estimate
\begin{equation} \label{eq:16331}
    g'\left(t + \frac{4r}{\delta} \right) \geq g'(t) + \frac{\delta}2.
\end{equation}
To prove~\eqref{eq:16331}, note that the assumption $\sup_{t \in \R} \left|g (t)- g_1(t)\right|  \leq r$ implies, for any $s\in \R$,
\begin{equation} \label{eq:17071}
g(s) - r \leq g_1(s) \leq g(s) + r.
\end{equation}
Using the inequality $ g_1'(t) -  g'(t) > \delta$ and the convexity of the map $g$, we see that, for any $s > t$,
\begin{equation}  \label{eq:17081}
g_1(s) \geq g_1(t) + g_1'(t) (s -t) >  g_1(t) + \left(g'(t) + \delta\right)  (s-t) >  g(t) - r + \left(g'(t) + \delta\right) (s-t).
\end{equation}
A combination of the estimates~\eqref{eq:17071} and~\eqref{eq:17081} yields
\begin{equation*}
\frac{g(s) - g(t)}{s-t} > g'(t) + \delta - \frac{2r}{s-t}.
\end{equation*}
Choosing the value $s = t + 4r/\delta$ in the previous inequality and using the convexity of $g$ shows
\begin{equation*}
    g'\left(t + \frac{4r}{\delta}\right) \geq \frac{g\left(t + \frac{4r}{\delta}\right) - g(t)}{4r/\delta} > g'(t) + \delta - \frac{\delta}{2} \geq g'(t) + \frac{\delta}{2}.
\end{equation*}
The proof of the claim~\eqref{eq:16331} is complete. In the case when the inequality (2) is satisfied, a similar argument yields the estimate
    \begin{equation} \label{eq:163312}
    g' \left(t - \frac{4r}{\delta}\right) \leq g'(t) - \frac{\delta}2.
\end{equation}
A combination of~\eqref{eq:16331} and~\eqref{eq:163312}, and the assumption that $g$ is convex (which implies that its derivative is increasing) shows that, for any point $t \in \mathrm{
Stab}(\lambda, \delta, r , g)$,
\begin{equation} \label{eq:10055}
    g' \left(t + \frac{4r}{\delta} \right) \geq g' \left(t - \frac{4r}{\delta}\right) + \frac{\delta}{2}.
\end{equation}
Using that the map $g$ is convex and $\lambda$-Lipschitz, we see that, for any triplet of real numbers $t_- , t , t_+ \in \R$ satisfying $t_- < t < t_+$,
\begin{equation} \label{eq:10056}
    -\lambda \leq g'(t_-) \leq  g'(t)  \leq g'(t_+) \leq \lambda.
\end{equation}
The estimates~\eqref{eq:10055} and~\eqref{eq:10056} imply that there cannot exist a family $t_1 , \ldots, t_{\lfloor \frac{4\lambda}{\delta} \rfloor +1}$ of $\left(\lfloor \frac{4 \lambda}{\delta} \rfloor + 1 \right)$-points satisfying the following properties:
\begin{enumerate}
\item For any pair of distinct integers $i,j \in \left\{1 , \ldots, \lfloor 4\lambda /\delta \rfloor +1 \right\}$, one has $|t_i - t_j| > \frac{8r}{\delta}$;
\item For any integer $i \in \left\{1 , \ldots, \lfloor 4\lambda/\delta  \rfloor +1 \right\}$, the point $t_i$ belongs to the set $\mathrm{Stab}(\lambda, \delta, r, g)$.
\end{enumerate}
This property implies that the set $\mathrm{Stab}(\lambda, \delta, r, g)$ is included in the union of (at most) $\lfloor \frac{4\lambda}{\delta} \rfloor$ intervals of length $16r/\delta$ which implies the upper bound
\begin{equation*}
    \Leb \left( \mathrm{Stab}(\lambda,\delta, r, g) \right) \leq  \frac{C \lambda r}{\delta^2}.
\end{equation*}
This is~\eqref{eq:1147}. The proof of Proposition~\ref{Propdeltastab} is complete.
\end{proof}

Proposition~\ref{Propdeltastab} implies a lower bound on the Gaussian measure of the set $\mathrm{Stab}(\lambda,  \delta, r, g)$.

\begin{corollary} \label{coroGaussmea}
There exists a constant $C \geq 1$ such that for each $\lambda > 0$, each $\lambda$-Lipschitz, convex and differentiable function $g : \R \to \R$, each pair of parameters $\delta ,r > 0$, and each variance $\sigma^2> 0$ such that $r \geq \sigma \delta^2/\lambda$,
\begin{equation*}
    \frac{1}{\sqrt{2 \pi \sigma^2}}\int_{\mathrm{Stab}(\lambda, \delta, r, g)} e^{-\frac{t^2}{2\sigma^2}} \, dt \leq  1 - e^{- C \frac{\lambda^2 r^2}{\sigma^2 \delta^4}}.
\end{equation*}
\end{corollary}

\begin{proof}
The proof is a consequence of the following inequality: there exists a constant~$C>0$ such that, for any $ \sigma^2 \in (0,1)$, any real number $\alpha \geq \sigma$,
   \begin{equation} \label{eq:1715}
       \sup_{A \subseteq \R,\, \Leb \left( A \right)  \leq \alpha}  \frac{1}{\sqrt{2 \pi \sigma^2}} \int_{A} e^{-\frac{t^2}{2 \sigma^2}} \, dt = \frac{1}{\sqrt{2 \pi \sigma^2}} \int_{-\frac \alpha 2}^{\frac \alpha2} e^{-\frac{t^2}{2 \sigma^2}} \, dt \leq  1 - e^{-  C\frac{ \alpha^2}{\sigma^2}}.
   \end{equation}
Corollary~\ref{coroGaussmea} is then obtained by combining the inequality~\eqref{eq:1715} with the value $\alpha = C \lambda r / \delta^2$ and the estimate~\eqref{eq:1147}.
\end{proof}

\subsection{A lower bound for the effect for the boundary condition on the averaged magnetization} \label{section3.2LB}

In this section, we apply the result of Proposition~\ref{Propdeltastab} to the free energy associated with a discrete spin system and obtain the inequality~\eqref{eq:155625}. As it will be useful in the rest of the proof of Theorem~\ref{thm1prop2.31708main}, we establish the result both in the case of boxes (see~\eqref{eq:1707bis}) and annuli (see~\eqref{eq:1707ter}). The lower bound is stated in the following lemma, and we recall the notations $\fluc_{\Lambda}(\eta)$, $\fluc_{\Lambda}^{ \tau_0}(\eta)$, $\fluc_{\Lambda,i}(\eta)$ and $\fluc_{\Lambda,i}^{\tau_0}(\eta)$ introduced in Section~\ref{NotaFluc}, as well as the notation~$R$ for the range of the observables $(f_v)$.

\begin{lemma} \label{L.lemma2.5}
Fix $d = 2$, $\beta >0$, $\lambda >0$, and $L \geq 2$. There exists a positive constant $C \geq 1$ depending on the parameters $\lambda$, $C_H$, $m$ and $R$ such that, for any $\delta >0$,
\begin{equation} \label{eq:1707bis}
        \P \left( \fluc_{\Lambda_L}(\eta) < \delta + \frac{C}{L} \right) > \exp \left( - \frac{C}{\delta^4} \right).
    \end{equation}
Additionally, for each vertex $v \in \Lambda_L$ and each integer $L' \leq \frac{L}2$ such that $v + \Lambda_{L'} \subseteq \Lambda_L$,
\begin{equation} \label{eq:1707ter}
        \P \left( \fluc_{\Lambda_L \setminus (v + \Lambda_{L'})}(\eta) < \delta + \frac{C}{L} \right) > \exp \left( - \frac{C}{\delta^4} \right).
    \end{equation}
\end{lemma}

\begin{remark}\label{Remark uniform beta}
Before giving the proof, we mention that, under the additional assumption that the Hamiltonians $(H_\Lambda)$ satisfy the upper bounds $\left| H_\Lambda \right| \leq C \left| \Lambda \right|$ (as is the case for the models presented in Section~\ref{sec:examples}), the free energy is a locally Lipschitz continuous function of the inverse temperature $\beta$, and thus, a small perturbation of the parameter $\beta$ only slightly modifies the free energy. Consequently, the argument given below (which relies on the notion of $\delta$-stability introduced in Section~\ref{section3.1}) can be extended so that the result of Lemma~\ref{L.lemma2.5} holds uniformly over the parameter $\beta$ when this quantity belongs to a small open set.
\end{remark}

\begin{proof}
 The argument relies on Proposition~\ref{Propdeltastab} and Corollary~\ref{coroGaussmea}. Using the upper bound $\fluc_{\Lambda_L}(\eta) \leq 2$, we may assume that $\delta \leq 2$ without loss of generality. First, by the inequalities~\eqref{eq:11250211}, we see that, to prove the inequality~\eqref{eq:1707bis}, it is sufficient to show, for any $i \in \{ 1 , \ldots, m \}$ and any fixed boundary condition $\tau_0 \in \S^{\Z^2 \setminus \Lambda_L}$,
 \begin{equation} \label{eq:11160211}
     \P \left( \fluc_{\Lambda_L,i}^{\tau_0}(\eta) < \delta + \frac{C}{L} \right) > \exp \left( - \frac{C}{\delta^4} \right).
 \end{equation}
 We now fix an integer $i \in \{1 , \ldots,m \}$ and prove~\eqref{eq:11160211}. We assume, without loss of generality, that $L \geq 2 R$.

We introduce the notation
\begin{equation}
    \bar \fluc_{\Lambda_L,i}^{\tau_0}(\eta) := \sup_{\tau \in \S^{\Z^2 \setminus \Lambda_L}} \left| \frac{1}{|\Lambda_L|}  \sum_{v \in \Lambda_{\left(L - R \right)} } \left\langle f_{v,i} \left( \sigma \right) \right\rangle_{\Lambda_L}^{\tau}  -  \left\langle f_{v,i} \left( \sigma \right) \right\rangle_{\Lambda_L}^{\tau_0} \right|,
\end{equation}
and make a few observations pertaining to this quantity. First, we have the volume estimates
\begin{equation} \label{eq:15141011}
\left| \Lambda_{\left(L - R \right)} \right| \geq c \left| \Lambda_{L} \right| \hspace{5mm} \mbox{and} \hspace{5mm} \left| \Lambda_L \setminus \Lambda_{\left(L - R \right)} \right| \leq C \left| \Lambda_L\right| / L.
\end{equation}
Additionally, if we decompose the field $\eta_{\Lambda_L}$ according to $\eta_{\Lambda_L} := \left( \hat \eta_{w,\Lambda_L}, \eta_{w , \Lambda_L }^\perp \right)$, where $w$ is the weight function given by $w := \indc_{\Lambda_{(L - R)}}$, then we have
\begin{equation} \label{eq:15261011}
     \frac{\partial\FE^{\tau_0}_{\Lambda_L} }{\partial \hat \eta_{w , \Lambda_L , i}} (\eta) =  -\frac{\lambda}{|\Lambda_L|}  \sum_{v \in \Lambda_{\left(L - R \right)} } \left\langle f_{v,i} \left( \sigma \right) \right\rangle_{\Lambda_L}^{\tau_0},
\end{equation}
and the random variable $\hat \eta_{w,\Lambda_L}$ is independent of the realization of the random field $\eta$ in the boundary layer $\Lambda_{L} \setminus \Lambda_{\left(L- R\right)}$. Moreover, using that the maps $(f_v)$ are bounded by $1$, we have
\begin{equation} \label{eq:15191011}
    \left| \fluc_{\Lambda_L,i}^{\tau_0}(\eta) - \bar \fluc_{\Lambda_L,i}^{\tau_0}(\eta) \right| \leq \frac{C}{L}.
\end{equation}
We then claim that the inequality~\eqref{eq:1707bis} is implied by the conditional inequality
    \begin{equation} \label{eq:17072}
        \P \left( \bar \fluc_{\Lambda_L,i}^{\tau_0}  < \delta ~ \Big\vert ~ \left(\hat \eta_{w ,\Lambda_L, j} \right)_{j \neq i}, \eta^\perp_{w ,\Lambda_L} \right) > \exp \left( - \frac{C}{\delta^4} \left(  1 + \frac{L}{\left| \Lambda_L\right|} \sum_{v \in \Lambda_{L} \setminus \Lambda_{\left(L - R \right)}} \left|\eta_v \right| \right)^2\right) \hspace{3mm} \P-\mbox{a.s.}
    \end{equation}
Indeed, taking the expectation in~\eqref{eq:17072} shows
\begin{align} \label{eq:15181011}
    \P \left( \bar \fluc_{\Lambda_L,i}^{\tau_0}(\eta)  < \delta \right) & = \E \left[\P \left( \bar \fluc_{\Lambda_L,i}^{\tau_0}  < \delta ~ \Big\vert ~ \left(\hat \eta_{w,\Lambda_L, j} \right)_{j \neq i}, \eta^\perp_{w,\Lambda_L} \right) \right] \\
    & >  \E \left[ \exp \left( - \frac{C}{\delta^4} \left(  1 + \frac{L}{\left| \Lambda_L\right|} \sum_{v \in \Lambda_{L} \setminus \Lambda_{\left(L - R \right)}} \left|\eta_v \right| \right)^2\right) \right] \notag \\
    & > \exp \left( - \frac{C}{\delta^4} \right). \notag
\end{align}
where we used that the random variables $\left(\eta_v \right)_{v \in \Lambda_{L} \setminus \Lambda_{\left(L - R \right)}}$ are Gaussian, independent, and the second volume estimate stated in~\eqref{eq:15141011}. Combining~\eqref{eq:15181011} with the pointwise bound~\eqref{eq:15191011} completes the proof of~\eqref{eq:1707bis}.

We now focus on the proof of~\eqref{eq:17072}.
To this end, let us fix a realization of the averaged fields $\left(\hat \eta_{w,\Lambda_L, j} \right)_{j \neq i}$ and of the orthogonal field $\eta^\perp_{w,\Lambda_L}$. We define
\begin{equation} \label{eq:15581011}
    g : \hat \eta_{w,\Lambda_L, i} \mapsto -\FE^{ \tau_0}_{\Lambda_L} \left( \hat \eta_{w,\Lambda_L,i}, \left(\hat \eta_{w,\Lambda_L, j} \right)_{j \neq i} , \eta^{\perp}_{w,\Lambda_L} \right) \hspace{3mm} \mbox{and} \hspace{3mm} r : = \frac{C \left| \partial \Lambda_L \right|}{\left| \Lambda_L \right|} + \frac{C \lambda}{\left| \Lambda_L\right|} \sum_{v \in \Lambda_{L} \setminus \Lambda_{\left(L - R \right)}} \left| \eta_v \right|,
\end{equation}
where $C$ is the constant appearing the right-hand side of~\eqref{eq:12030211}. We note that, by Proposition~\ref{prop.basicpropfreeen}, the function $g$ is convex, $\lambda$-Lipschitz, differentiable and its derivative is given by the formula~\eqref{eq:15261011}.

\smallskip

The estimate~\eqref{eq:12030211} can be rewritten as follows: for any boundary condition $\tau \in \S^{\Z^2 \setminus \Lambda_L}$, the map $\hat \eta_{w,\Lambda_L, i} \mapsto -\FE^{\tau}_{\Lambda_L} \left( \hat \eta_{w,\Lambda_L, i}, \left(\hat \eta_{w,\Lambda_L, j} \right)_{j \neq i} , \eta^{\perp}_{w,\Lambda_L} \right)$ belongs to the space $N_{\lambda, r}(g).$

\smallskip

The previous observation combined with the identity~\eqref{eq:17510211} yields the inclusion of sets
\begin{equation*}
    \left\{ \hat \eta_{\Lambda_L, w,i} \in \R \, : \, \bar \fluc_{\Lambda_L,i}^{\tau_0}( \hat \eta_{w,\Lambda_L, i}, \left(\hat \eta_{w,\Lambda_L, j} \right)_{j \neq i} , \eta^{\perp}_{w , \Lambda_L}  ) > \delta \right\} \subseteq \mathrm{Stab}(\lambda, \lambda \delta, r , g).
\end{equation*}
Applying Proposition~\ref{Propdeltastab} and the formula for the parameter $r$ stated in~\eqref{eq:15581011}, we deduce
\begin{align} \label{eq:1139}
    \lefteqn{\Leb \left(\left\{ \hat \eta_{w,\Lambda_L, i} \in \R \, : \, \bar \fluc_{\Lambda_L,i}^{\tau_0}( \hat \eta_{w,\Lambda_L, i}, \left(\hat \eta_{w,\Lambda_L, j} \right)_{j \neq i} , \eta^{\perp}_{w,\Lambda_L}  ) \geq \delta \right\} \right)} \qquad \qquad\qquad\qquad & \\
    & \leq \frac{C \left| \partial \Lambda_L \right|}{ \delta^2 \left| \Lambda_L \right|} + \frac{C}{ \delta^2 \left| \Lambda_L\right|} \sum_{v \in \Lambda_{L} \setminus \Lambda_{\left(L - R \right)}} \left| \eta_v \right| \notag \\
    & \leq  \frac{C }{ \delta^2  L} + \frac{C}{ \delta^2 \left| \Lambda_L\right|} \sum_{v \in \Lambda_{L} \setminus \Lambda_{\left(L - R \right)}} \left| \eta_v \right|, \notag
\end{align}
where we used the upper bound $\left| \partial \Lambda_L \right|/\left| \Lambda_L \right| \leq C/L $ in the second inequality.
Using that the random variable ${\hat\eta_{w,\Lambda_L,i}}$ is Gaussian of variance $\left| \Lambda_{\left(L - R\right) } \right|^{-1} \geq c L^{-2}$, that the random variables ${\hat\eta_{w,\Lambda_L,i}}$, $\left( \hat\eta_{w,\Lambda_L,j} \right)_{j \neq i}$ and $\eta^{\perp}_{w,\Lambda_L}$ are independent, and Corollary~\ref{coroGaussmea}, we obtain that the inequality~\eqref{eq:1139} implies the estimate~\eqref{eq:17072}.

The proof of the inequality~\eqref{eq:1707ter} only requires a notational modification of the previous argument, we thus omit the details.
\end{proof}

\subsection{Mandelbrot percolation argument} \label{sectionfirstmandelbrotperc}
In this section, we combine the result obtained in Lemma~\ref{L.lemma2.5} with a Mandelbrot percolation argument to obtain the quantitative estimate~\eqref{eq:08.15}. The result is stated in the following lemma.

\begin{lemma} \label{prop2.31708}
Fix $d = 2$, $\beta > 0$, $\lambda > 0$ and $L \geq 3$. There exist two positive constants $c \in (0,1)$ and $C\in (1 , \infty)$ such that
\begin{equation} \label{eq:18031}
     \P \left( \fluc_{\Lambda_L}(\eta) < \frac{C}{\sqrt[4]{\ln \ln L}}  \right) \geq  1 - \exp \left( - c \sqrt{\ln L} \right).
    \end{equation}
\end{lemma}

\begin{proof}
We set $\delta := (C_0 / \sqrt[4]{\ln \ln L}) \wedge 4 $, for some large constant $C_0 \geq 8$ whose value will be selected later in the proof. Using the estimate~\eqref{eq:1707ter} and the inequality $\sqrt{\ln \ln L} \leq \sqrt{L}$, for any $L \geq 3$, we assume that the constant $C_0$ is large enough so that: for any box $\Lambda \subseteq \Zd$ of side length $\ell \geq \sqrt{L}$, and any box $\Lambda' \subseteq \Lambda$ whose side length is smaller than $\ell /2$,
\begin{equation} \label{eq:16231011}
        \P \left( \fluc_{\Lambda \setminus \Lambda'}(\eta) < \frac{\delta}{2} \right) > \exp \left( - \frac{C}{\delta^4} \right).
\end{equation}
The strategy is to implement a Mandelbrot percolation argument. To this end, we define the following notion of good box: a box $\Lambda' \subseteq \Lambda_L$ is good if and only if
\begin{equation} \label{eq:10041}
    \fluc_{\Lambda'}(\eta) \leq \delta.
\end{equation}
We say that a box is bad if it is not good. Let us recall that the event~\eqref{eq:10041} only depends on the realization of the random field $\eta$ inside the box $\Lambda'$ (since the random variable $\fluc_{\Lambda'}$ only depends on the value of the field $\eta$ inside the box $\Lambda'$).
\begin{figure}
       \centering
       \includegraphics[scale = 0.1]{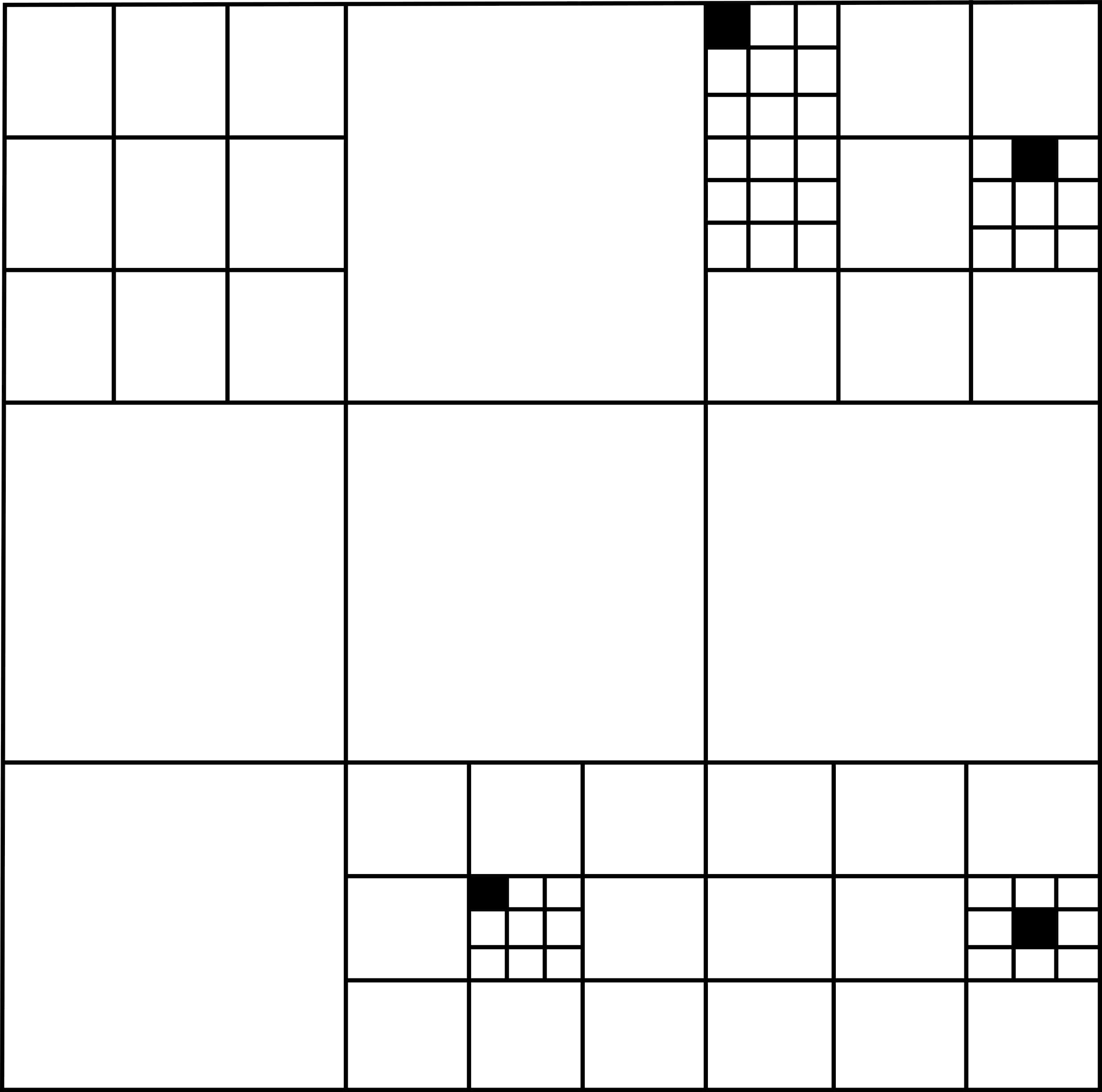}
       \caption{A realization of the Mandelbrot percolation with the values $k = 3$ and $l_{\max} = 3$. The bad cubes are drawn in black.}
        \label{partitionempty}
\end{figure}

Let us first introduce a few additional notations. Given an odd integer $k \geq 3$, we denote by $l_{\max}$ the largest integer which satisfies $k^{l_{\max}} \leq \sqrt{L}$, i.e., $l_{\max} = \lfloor \ln L / (2\ln k) \rfloor$. For each integer $l \in
\left\{ 0 , \ldots , l_{\max} \right\}$, we introduce the set of boxes
\begin{equation} \label{def.TL}
   \mathcal{T}_l := \left\{ \left( z  + \left[ - \frac{L}{k^{l}}  ,  \frac{L}{k^{l}} \right)^2 \right) \cap \Lambda_L \, : \, z \in \frac{2 L}{k^{l}} \Z^2 \cap [-L , L]^2 \right\}.
\end{equation}
We note that that, for each integer $l \in \{0 , \ldots, l_{\max} \}$, the collection of boxes $\mathcal{T}_l$ forms a partition of the box $\Lambda_L$. Additionally, two boxes of the collection $\bigcup_{l=0}^{l_{\max}} \mathcal{T}_l$ are either disjoint or included in one another. For each vertex $v \in \Lambda_L$, and each integer $l \in \{0 , \ldots, l_{\max} \},$ we denote by $\Lambda_l(v)$ the unique box of the set $\mathcal{T}_l$ containing the point $v$.

We select the integer $k$ to be the smallest odd integer larger or equal to $3$ such that the following properties are satisfied for any $L \geq 9$:
\begin{equation*}
    l_{\max} \geq 1 \hspace{3mm} \mbox{and} \hspace{3mm} \forall v \in \Lambda_L, \forall l \in \{ 0 , \ldots , l_{\max}-1 \}, \hspace{3mm} \left| \Lambda_{l+1} (v)\right| \leq \frac{\delta}{4} \left| \Lambda_l (v)\right|.
\end{equation*}
We remark that this integer always exists if the constant $C_0$ is chosen large enough, and that there exist two constants $C,c$ such that $c \delta^{-\frac 12} \leq k \leq C \delta^{-\frac 12} $.

\smallskip

We then construct recursively a (random) sequence of collection of good boxes $\Q_l \subseteq \mathcal{T}_l$, for $l \in \{ 0 , \ldots, l_{\max} \}$, according to the following algorithm:
\begin{itemize}
    \item Initiation: we set $\mathcal{Q}_0 = \left\{ \Lambda_L \right\}$ if $\Lambda_L$ is a good box and $\mathcal{Q}_0 = \emptyset$ otherwise;
    \item Induction step: we assume that the sets $\mathcal{Q}_0 , \ldots, \Q_{l-1}$ have been constructed, and wish to construct the collection $\Q_l$. We consider the set of boxes $\mathcal{T}_l$ to which we remove all the boxes which are included in a box of the collection $\bigcup_{i=0}^{l-1} \mathcal{Q}_i$, that is, we define the set
    \begin{equation*}
        \mathcal{T}'_l := \left\{ \Lambda' \in \mathcal{T}_l \, : \, \forall \Lambda'' \in \bigcup_{i=0}^{l-1} \mathcal{Q}_i , \, \Lambda' \not\subseteq \Lambda'' \right\}.
    \end{equation*}
    We define the set $\Q_l$ to be the set of boxes which are good and belong to $\mathcal{T}'_l$, i.e.,
    \begin{equation*}
        \mathcal{Q}_l := \left\{ \Lambda' \in \mathcal{T}_l' \, : \, \Lambda' ~\mbox{is good}  \right\}.
    \end{equation*}
\end{itemize}
We then define $\Q := \cup_{l = 0}^{l_{\max}} \Q_l$. Let us note that two boxes in the set $\Q$ are either equal or disjoint. The collection of boxes $\mathcal{Q}$ is not in general a partition of the box $\Lambda_L$, and there is a non-empty set of uncovered points which can be characterized by the following criterion:
\begin{equation} \label{eq:1449}
v \in \Lambda_L ~\mbox{is uncovered} \iff \forall l \in \left\{0 , \ldots, l_{\max} \right\}, \, \Lambda_l(v) ~\mbox{is a bad box}.
\end{equation}
We next show that the set of uncovered points is small. To this end, we show the following upper bound on the probability of a vertex $v \in \Lambda_L$ to be uncovered: there exists a constant $c \in (0,1)$ such that
\begin{equation} \label{1333}
    \P \left( v ~\mbox{is uncovered}\right) \leq \exp \left( - c \sqrt{\ln L}\right).
\end{equation}
To prove the inequality~\eqref{1333}, we fix a vertex $v \in \Lambda_L$, and rewrite the equivalence~\eqref{eq:1449} as follows:
\begin{equation} \label{eq:1350}
    \left\{  v ~\mbox{is uncovered} \right\} = \bigcap_{l=0}^{l_{\max}} \left\{ \fluc_{\Lambda_l(v)} (\eta) >  \delta \right\}.
\end{equation}
The strategy is then to prove that the $(l_{\max}+1)$ events in the right side of~\eqref{eq:1350} are well-approximated by independent events, and to use the independence in order to estimate the probability of their intersection. By the domain subadditivity property stated in Proposition~\ref{propdommarkov}, and the pointwise bound $\fluc_{\Lambda}  \leq 2$, we have, for any $v \in \Lambda_L$ and any $l \in \{ 0 , \ldots , l_{\max} -1\}$,
\begin{align} \label{eq:14051}
     \fluc_{\Lambda_l(v)}(\eta) & \leq \frac{\left| \Lambda_l(v) \setminus \Lambda_{l+1}(v) \right|}{\left| \Lambda_l(v)\right|} \fluc_{\Lambda_l(v) \setminus \Lambda_{l+1}(v)}(\eta)  + \frac{\left| \Lambda_{l+1}(v) \right|}{\left| \Lambda_l(v)\right|} \fluc_{ \Lambda_{l+1}(v)}(\eta) \\
     & \leq \fluc_{\Lambda_l(v) \setminus \Lambda_{l+1}(v)}(\eta) + \frac{2|\Lambda_{l+1}(v)|}{|\Lambda_l(v)|} \notag \\
     & \leq \fluc_{\Lambda_l(v) \setminus \Lambda_{l+1}(v)}(\eta) + \frac{\delta}{2 }, \notag
\end{align}
where we used in the last inequality that, by the definition of the integer $k$, the ratio of the volumes of the boxes $\Lambda_{l+1}(v)$ and $\Lambda_{l}(v)$ is smaller than $4/\delta$. The estimate~\eqref{eq:14051} implies the inclusion of events
\begin{equation} \label{eq:1413}
    \left\{  v ~\mbox{is uncovered} \right\} \subseteq \bigcap_{l=0}^{l_{\max}-1} \left\{ \fluc_{\Lambda_l(v) \setminus \Lambda_{l+1}(v)}(\eta) \geq \frac \delta{2} \right\}.
\end{equation}
Using that the annuli $\left(\Lambda_l(v) \setminus \Lambda_{l+1}(v) \right)_{l \in \{ 0 , \ldots, l_{\max}-1\}}$ are disjoint and that the random variables $\fluc_{\Lambda_l(v) \setminus \Lambda_{l+1}(v)}(\eta)$ depend only on the restriction of the random field to the annulus $\Lambda_l(v) \setminus \Lambda_{l+1}(v)$, we obtain that the events in the right side of~\eqref{eq:1413} are independent. We deduce that
\begin{equation*}
    \P \left(  v ~\mbox{is uncovered} \right) \leq \prod_{l = 0}^{l_{\max}-1} \P \left( \fluc_{\Lambda_l(v) \setminus \Lambda_{l+1}(v)}(\eta) \geq \frac{\delta}{2} \right).
\end{equation*}
We recall the definition of the parameter $\delta$ and of the integer $l_{\max}$ stated at the beginning of the proof. Using~\eqref{eq:16231011}, we obtain the following dichotomy: 
\begin{itemize}
\item If $C_0 / \sqrt[4]{\ln \ln L} < 4$ (i.e., if $L$ is sufficiently large  so that $\delta < 4$), then
\begin{equation*}
    \P \left(  v ~\mbox{is uncovered} \right)  \leq \left( 1 - \exp \left( - \frac{C}{\delta^4} \right) \right)^{l_{\max}}  \leq \left( 1 - \frac{1}{\left( \ln L \right)^{C / C_0^4}} \right)^{c \frac{\ln L}{\ln \ln \ln L}}.
\end{equation*}
\item If $C_0 / \sqrt[4]{\ln \ln L} \geq 4$, then $\delta = 4$ and thus all the boxes are good (as $\fluc_{\Lambda'}$ is smaller than $2$ almost-surely for any box $\Lambda'$ by assumption, the condition~\eqref{eq:10041} is always satisfied). This gives
\begin{equation*}
    \P \left(  v ~\mbox{is uncovered} \right)  = 0.
\end{equation*}
\end{itemize}
Choosing the constant $C_0$ large enough, e.g., larger than $\sqrt[4]{4 C}$, we obtain (in both cases)
\begin{equation*}
    \P \left(  v ~\mbox{is uncovered} \right) \leq \exp \left( - c \sqrt{\ln L} \right).
\end{equation*}
The proof of~\eqref{1333} is complete. We now use the inequality~\eqref{1333} to complete the proof of the estimate~\eqref{eq:18031}. By the domain subadditivity property for the quantity $\fluc_{\Lambda}$ and the pointwise bound $\fluc_{\Lambda}  \leq 2$, we have
\begin{align} \label{eq:17073}
    \fluc_{\Lambda_L}(\eta) & \leq \sum_{\Lambda \in \mathcal{Q}} \frac{\left| \Lambda \right|}{\left| \Lambda_L \right|} \fluc_{\Lambda}(\eta)+ 2\frac{\left| \Lambda_L \setminus \bigcup_{\Lambda \in \Q} \Lambda \right|}{|\Lambda_L|}  \\
    & \leq \sum_{\Lambda \in \mathcal{Q}} \frac{\left| \Lambda \right|}{\left| \Lambda_L \right|} \delta +2 \frac{\left| \Lambda_L \setminus \bigcup_{\Lambda \in \Q} \Lambda \right|}{|\Lambda_L|} \notag \\
    & \leq \delta  + 2 \frac{\left| \Lambda_L \setminus \bigcup_{\Lambda \in \Q} \Lambda \right|}{|\Lambda_L|}. \notag
\end{align}
By using~\eqref{1333} and Markov's inequality, we have
\begin{align} \label{eq:1304}
    \P \left( \frac{\left| \Lambda_L \setminus \bigcup_{\Lambda \in \Q} \Lambda \right|}{|\Lambda_L|} \geq \delta \right)  \leq \frac{\E \left[ \frac{1}{|\Lambda_L|} \sum_{v \in \Lambda_L} \indc_{\{ v \,\mathrm{is \, uncovered} \}} \right]}{\delta} & \leq \frac{\sum_{v \in \Lambda_L}\P \left( v ~\mbox{is uncovered} \right)}{|\Lambda_L| \delta} \\ & \leq \frac{\exp \left( - c \sqrt{\ln L} \right)}{\delta} \notag \\
    & \leq C \exp \left( - c \sqrt{\ln L} \right), \notag
\end{align}
by reducing the value of the constant $c$ in the last inequality. We have thus obtained
\begin{equation*}
    \P \left( \fluc_{\Lambda_L}(\eta) > 3 \delta \right) \leq C \exp \left( - c \sqrt{\ln L} \right).
\end{equation*}
The inequality~\eqref{eq:18031} can then be obtained by adjusting the values of the constants $C$ and $c$. The proof of Lemma~\ref{prop2.31708} is complete.
\end{proof}

\subsection{Upgrading the stochastic integrability} \label{sectionconcentrationatg}

This section is the final step of the proof of Theorem~\ref{thm1prop2.31708main} (in the general case). We use a concentration argument combined with the domain subadditivity property applied to the quantity $\fluc_\Lambda$ to upgrade the stochastic integrability obtained in Lemma~\ref{prop2.31708}.

\medskip

\begin{proof}[Proof of~\eqref{eq:11191006} of Theorem~\ref{thm1prop2.31708main}]
We split the box $\Lambda_L$ into (approximately) $N \simeq  L^{\frac 32}$ boxes of side length of order $L^{\frac14}$. We denote these boxes by $\tilde \Lambda_1 , \ldots, \tilde \Lambda_N$. By the domain subadditivity property for the quantity $\fluc_\Lambda$, we have the inequality
\begin{equation} \label{eq:1455}
    \fluc_{\Lambda_L}(\eta) \leq \frac{1}{N} \sum_{i = 1}^N  \fluc_{\tilde \Lambda_i}(\eta).
\end{equation}
Since the boxes $\tilde \Lambda_1 , \ldots, \tilde \Lambda_N$ are disjoint, the sum in the right side is a sum of independent random variables, to which we can apply a concentration argument. To implement this strategy, we fix an integer $i \in \{ 1 , \ldots  , N\}$, apply the inequality~\eqref{eq:18031} to the box $\tilde \Lambda_i$ (the result was proved for the boxes of the form $\Lambda_L$ for $L\geq 3$ but can be extended to any box of $\Zd$ by translation invariance of the random field), together with the bound $ \fluc_{\tilde \Lambda_i} \leq 2$, and the fact that that the side length of the box $\tilde \Lambda_i$ is larger than $cL^{\frac 14}$. We obtain, for any integer $i \in \{ 1 , \ldots , N\}$,
\begin{equation} \label{eq:16290211}
    \E \left[  \fluc_{\tilde \Lambda_i} (\eta) \right] \leq \frac{C}{\sqrt[4]{\ln \ln L^{1/4}}} \leq \frac{C}{\sqrt[4]{\ln \ln L}}.
\end{equation}
We use the inequalities~\eqref{eq:1455},~\eqref{eq:16290211} to obtain, for any $K \geq 0$,
\begin{equation} \label{eq:1607}
    \P \left[  \fluc_{\Lambda_L}(\eta) \geq \frac{C}{\sqrt[4]{\ln \ln L}}  + K \right] \leq \P \left[  \frac{1}{N} \sum_{i = 1}^N \left( \fluc_{\tilde \Lambda_i}(\eta)  - \E \left[ \fluc_{\tilde \Lambda_i}(\eta) \right]  \right)\geq  K \right] \notag.
\end{equation}
Using that the random variables $\fluc_{\tilde \Lambda_1}(\eta) , \ldots, \fluc_{\tilde \Lambda_N}(\eta)$ are i.i.d., non-negative, bounded by $2$ almost surely together with Hoeffding's inequality, we obtain, for any $K \geq 0$,
\begin{equation*}
    \P \left[  \fluc_{\Lambda_L}(\eta) \geq \frac{C}{\sqrt[4]{\ln \ln L}}  + K \right] \leq C e^{-c N K^2}.
\end{equation*}
Recalling that $N$ is comparable to $L^{\frac 32}$, and choosing $K = \frac{C}{\sqrt[4]{\ln \ln L}}$, we have obtained
\begin{equation} \label{eq:1629}
    \P \left[  \fluc_{\Lambda_L}(\eta) \geq \frac{ 2 C}{\sqrt[4]{\ln \ln L}} \right] \leq C \exp \left(-\frac{ c L^{3/2}}{\sqrt{\ln \ln L}} \right)  \leq C \exp \left(-c L \right). \notag
\end{equation}
The estimate~\eqref{eq:11191006} is then obtained by adjusting the values of the constants $C$ and $c$.
\end{proof}

\subsection{The translation-invariant setup} In this section, we prove the estimate~\eqref{eq:fluctuations around limiting value} pertaining to the translation-invariant setup and thus complete the proof of Theorem~\ref{thm1prop2.31708main}

\begin{proof}[Proof of~\eqref{eq:fluctuations around limiting value} of Theorem~\ref{thm1prop2.31708main}]

We assume in this proof that the model is translation invariant and satisfies the corresponding additional assumptions stated in Section~\ref{sec:general setup}. We fix an integer $i \in \{ 1 , \ldots , m \}$.

First, using the upper bound $|f_v| \leq 1$ and the estimate~\eqref{eq:11191006} of Theorem~\ref{thm1prop2.31708main}, we obtain, for any integer $L \geq 3$,
\begin{align} \label{14273112}
    \lefteqn{\E \left[ \sup_{\tau \in \S^{\Z^2 \setminus \Lambda_L}} \frac{1}{\left| \Lambda_L\right|}\sum_{v \in \Lambda_L} \left\langle f_{\mathbf 0 , i} \left( \mathcal{T}_v \sigma \right) \right\rangle_{\Lambda_L}^{\tau}  \right] - \E \left[ \inf_{\tau \in \S^{\Z^2 \setminus \Lambda_L}}  \frac{1}{\left| \Lambda_L\right|} \sum_{v \in \Lambda_L} \left\langle f_{\mathbf 0 , i} \left( \mathcal{T}_v \sigma \right) \right\rangle_{\Lambda_L}^{\tau} \right]} \qquad \qquad \qquad & \\  & \leq
    \E \left[ \sup_{\tau_1 , \tau_2 \in \S^{\Z^2 \setminus \Lambda_L}} \left|  \frac{1}{\left| \Lambda_L\right|} \sum_{v \in \Lambda_L} \left( \left\langle f_{\mathbf 0} \left( \mathcal{T}_v \sigma \right) \right\rangle_{\Lambda_L}^{\tau_1} -  \left\langle f_{\mathbf 0} \left( \mathcal{T}_v \sigma \right) \right\rangle_{\Lambda_L}^{\tau_2} \right) \right| \right] \notag
    \\ & \leq \frac{C}{\sqrt[4]{\ln \ln L}}. \notag
\end{align}
In the translation-invariant setup, we may combine the inequalities stated in Proposition~\ref{propdommarkov} with a subadditivity argument to obtain the following convergences
\begin{equation} \label{14503112}
    \begin{aligned}
     \E \left[ \sup_{\tau \in \S^{\Z^2 \setminus \Lambda_L}}  \frac{1}{\left| \Lambda_L\right|} \sum_{v \in \Lambda_L} \left\langle f_{\mathbf 0 , i}\left( \mathcal{T}_v \sigma \right) \right\rangle_{\Lambda_L}^{\tau} \right] \underset{L \to \infty}{\longrightarrow} \inf_{L \in \N} \, \E \left[  \sup_{\tau \in \S^{\Z^2 \setminus \Lambda_L}}  \frac{1}{\left| \Lambda_L \right|} \sum_{v \in \Lambda_L} \left\langle f_{\mathbf 0 , i} \left( \mathcal{T}_v \sigma \right) \right\rangle_{\Lambda_L}^{\tau} \right],\\ \vspace{3mm}
    \E \left[ \inf_{\tau \in \S^{\Z^2 \setminus \Lambda_L}}  \frac{1}{\left| \Lambda_L\right|} \sum_{v \in \Lambda_L} \left\langle f_{\mathbf 0 , i} \left( \mathcal{T}_v \sigma \right) \right\rangle_{\Lambda_L}^{\tau} \right] \underset{L \to \infty}{\longrightarrow} \sup_{L \in \N} \, \E \left[  \inf_{\tau \in \S^{\Z^2 \setminus \Lambda_L}} \frac{1}{\left| \Lambda_L \right|} \sum_{v \in \Lambda_L} \left\langle f_{\mathbf 0 , i} \left( \mathcal{T}_v \sigma \right) \right\rangle_{\Lambda_L}^{\tau} \right].
    \end{aligned}
\end{equation}
Using the inequality~\eqref{14273112}, we see that the two limits in the right-hand sides of~\eqref{14503112} are equal. In particular, we may define
\begin{equation*}
    \alpha_i := \inf_{L \in \N} \, \E \left[   \frac{1}{\left| \Lambda_L\right|} \sup_{\tau \in \S^{\Z^2 \setminus \Lambda_L}} \sum_{v \in \Lambda_L} \left\langle f_{\mathbf 0 , i} \left( \mathcal{T}_v \sigma \right) \right\rangle_{\Lambda_L}^{\tau} \right] =  \sup_{L \in \N} \, \E \left[  \inf_{\tau \in \S^{\Z^2 \setminus \Lambda_L}} \frac{1}{\left| \Lambda_L\right|} \sum_{v \in \Lambda_L} \left\langle f_{\mathbf 0 , i} \left( \mathcal{T}_v \sigma \right) \right\rangle_{\Lambda_L}^{\tau} \right].
\end{equation*}
Additionally, we see that, for any integer $L \in \N$,
\begin{equation} \label{eq:145531212}
    \E \left[  \inf_{\tau \in \S^{\Z^2 \setminus \Lambda_L}}  \frac{1}{\left| \Lambda_L\right|} \sum_{v \in \Lambda_L} \left\langle f_{\mathbf 0 , i} \left( \mathcal{T}_v \sigma \right) \right\rangle_{\Lambda_L}^{\tau} \right] \leq \alpha_i \leq  \E \left[  \sup_{\tau \in \S^{\Z^2 \setminus \Lambda_L}}  \frac{1}{\left| \Lambda_L\right|} \sum_{v \in \Lambda_L} \left\langle f_{\mathbf 0 , i} \left( \mathcal{T}_v \sigma \right) \right\rangle_{\Lambda_L}^{\tau} \right].
\end{equation}
A combination of~\eqref{14273112} and~\eqref{14503112} thus yields
\begin{equation*}
\begin{aligned}
     \left| \E \left[ \sup_{\tau \in \S^{\Z^2 \setminus \Lambda_L}} \frac{1}{\left| \Lambda_L \right|} \sum_{v \in \Lambda_L} \left\langle f_{\mathbf 0 , i}\left( \mathcal{T}_v \sigma \right) \right\rangle_{\Lambda_L}^{\tau} \right] - \alpha_i \right| \leq \frac{C}{\sqrt[4]{\ln \ln L}},\\ \vspace{3mm}
    \left| \E \left[ \inf_{\tau \in \S^{\Z^2 \setminus \Lambda_L}} \frac{1}{\left| \Lambda_L \right|} \sum_{v \in \Lambda_L} \left\langle f_{\mathbf 0 , i} \left( \mathcal{T}_v \sigma \right) \right\rangle_{\Lambda_L}^{\tau} \right]  - \alpha_i \right| \leq  \frac{C}{\sqrt[4]{\ln \ln L}}.
    \end{aligned}
\end{equation*}
Consequently, using the inequalities stated in Proposition~\ref{propdommarkov} together with the same concentration argument as the one developed in the proof of~\eqref{eq:11191006} of Theorem~\ref{thm1prop2.31708main}, we obtain the inequalities
\begin{equation} \label{eq:15573112}
    \begin{aligned}
     \P \left[ \sup_{\tau \in \S^{\Z^2 \setminus \Lambda_L}} \frac{1}{\left| \Lambda_L \right|} \sum_{v \in \Lambda_L} \left\langle f_{\mathbf 0 , i}\left( \mathcal{T}_v \sigma \right) \right\rangle_{\Lambda_L}^{\tau} - \alpha_i \geq \frac{C}{\sqrt[4]{\ln \ln L}}  \right] \leq \exp \left( - c L \right),\\ \vspace{3mm}
    \P \left[ \inf_{\tau \in \S^{\Z^2 \setminus \Lambda_L}} \frac{1}{\left| \Lambda_L \right|} \sum_{v \in \Lambda_L} \left\langle f_{\mathbf 0 , i}\left(\mathcal{T}_v \sigma \right) \right\rangle_{\Lambda_L}^{\tau} - \alpha_i \leq - \frac{C}{\sqrt[4]{\ln \ln L}}  \right] \leq \exp \left( - c L \right).
    \end{aligned}
\end{equation}
We then note that the following inclusion of events holds
\begin{multline} \label{eq:15563112}
    \left\{ \sup_{\tau \in \S^{\Z^2 \setminus \Lambda_L}} \left| \frac{1}{\left| \Lambda_L \right|}\sum_{v \in \Lambda_L} \left\langle f_{\mathbf 0 , i}\left( \mathcal{T}_v \sigma \right) \right\rangle_{\Lambda_L}^{\tau}  - \alpha_i  \right| \geq \frac{C}{\sqrt[4]{\ln \ln L}}  \right\} \\ \subseteq \left\{ \sup_{\tau \in \S^{\Z^2 \setminus \Lambda_L}} \frac{1}{\left| \Lambda_L \right|} \sum_{v \in \Lambda_L} \left\langle f_{\mathbf 0 , i}\left( \mathcal{T}_v \sigma \right) \right\rangle_{\Lambda_L}^{\tau} - \alpha_i \geq \frac{C}{\sqrt[4]{\ln \ln L}} \right\} \\ \bigcup  \left\{ \inf_{\tau \in \S^{\Z^2 \setminus \Lambda_L}} \frac{1}{\left| \Lambda_L \right|} \sum_{v \in \Lambda_L} \left\langle f_{\mathbf 0 , i}\left( \mathcal{T}_v \sigma \right) \right\rangle_{\Lambda_L}^{\tau} - \alpha_i \leq -\frac{C}{\sqrt[4]{\ln \ln L}} \right\}.
\end{multline}
A combination of~\eqref{eq:15573112},~\eqref{eq:15563112} and a union bound then yield
\begin{equation} \label{eq:16053112}
    \P \left( \sup_{\tau\in \S^{\Z^2 \setminus \Lambda_L}} \left| \alpha_i - \frac{1}{|\Lambda_L|} \sum_{v \in \Lambda_L} \left\langle f_{\mathbf 0 , i} \left( \mathcal{T}_v\sigma \right) \right\rangle_{\Lambda_L}^{\tau} \right|  > \frac{C}{\sqrt[4]{\ln \ln L}}  \right) \leq \exp \left( - c L \right).
\end{equation}
Since the inequality~\eqref{eq:16053112} is valid for any integer $i \in \{ 1 , \ldots, m\}$, we obtain the inequality~\eqref{eq:fluctuations around limiting value} of Theorem~\ref{thm1prop2.31708main} with the value $\alpha := \left( \alpha_1 , \ldots, \alpha_m \right)$.

\end{proof}

 \subsection{Proof of Theorem 2} \label{proofTHm2}

The objective of this section is to generalize Theorem~\ref{thm1prop2.31708main}. We prove that for any $L \geq 3$, any box $\Lambda \subseteq \Zd$ of side length $L$, and any deterministic weight function $w : \Lambda \to [-1,1]^m$, the quantity
\begin{equation*}
     \fluc_{w ,\Lambda}(\eta) := \sup_{\tau_1 , \tau_2 \in \S^{\Z^2 \setminus \Lambda}} \frac{1}{|\Lambda|} \left| \sum_{v \in \Lambda} w(v) \cdot \left( \left\langle f_v \left( \sigma \right) \right\rangle_{\Lambda}^{\tau_1}  -  \left\langle f_v \left( \sigma \right) \right\rangle_{\Lambda}^{\tau_2} \right) \right|
\end{equation*}
is smaller than $\sqrt[8]{\ln \ln L}$ with high probability. Let us note that the quantity  $\fluc_{w ,\Lambda}(\eta)$ depends only on the realization of the field $\eta$ inside the box $\Lambda$, satisfies the same domain subadditivity property as the one stated in Proposition~\ref{propdommarkov}, and that, by the assumption $\left| w(v)\right| \leq 2^{m/2}$ (which follows from $w(v) \in [-1,1]^m$), one has the bound $0 \leq\fluc_{w ,\Lambda}(\eta) \leq 2^{m/2 +1 }$ for any realization of the random field $\eta$.

The argument is similar to the one developed in the proof of Theorem~\ref{thm1prop2.31708main}, the main difference is that we rely on the identity~\eqref{eq:17510211} to obtain information on the observable $\fluc_{w ,\Lambda_L }(\eta)$ instead of~\eqref{eq:02111210}. The fact that the map $w$ can take small values must be taken into account in the argument and causes a slight deterioration of the rate of convergence: we obtain the quantitative rate $\sqrt[8]{\ln \ln L}$ instead of the rate $\sqrt[4]{\ln \ln L}$ obtained in Theorem~\ref{thm1prop2.31708main}.

\begin{proof}[Proof of Theorem~\ref{prop:1109}]
The strategy of the argument is similar to the proof of Theorem~\ref{thm1prop2.31708main}. We only present a detailed sketch of the argument pointing out the main differences with the proof of Theorem~\ref{thm1prop2.31708main}. The first step is to prove the estimate: for any box $\Lambda \subseteq \Lambda_L$ of side length larger than $\sqrt{L}$ such that $ \sum_{v \in \Lambda} \left|w(v)\right|^2 \geq \left| \Lambda \right|/\sqrt{\ln \ln L}$ and any $\delta >$,
    \begin{equation} \label{eq:1233}
        \P \left( \fluc_{w ,\Lambda}(\eta) < \delta \right) \geq  \exp \left( - \frac{C\sqrt{\ln \ln L}}{\delta^4} \right).
    \end{equation}
The proof is similar to the proof of Lemma~\ref{L.lemma2.5}; the main differences are that we need to decompose the random field $\eta$ according to the formula, for each $i \in \{ 1 , \ldots, m \}$,  $\eta := \left( \hat \eta_{w , \Lambda , i} , \left( \hat\eta_{w , \Lambda , j} \right)_{j \neq i} , \hat\eta_{w , \Lambda}^\perp \right)$ (following the notation introduced in Section~\ref{secstructRF}), use the identity~\eqref{eq:17510211} (instead of~\eqref{eq:02111210}), and use Corollary~\ref{coroGaussmea} with the variance $\sigma^2 := 1 / \sum_{v \in \Lambda} w_i(v)^2$ instead of  $\sigma^2 = 1/\left|\Lambda \right|$.

\smallskip

The second step of the argument corresponds to Section~\ref{sectionfirstmandelbrotperc}. We combine the inequality~\eqref{eq:1233} with a Mandelbrot percolation argument and prove the estimate
    \begin{equation} \label{eq:18033}
     \P \left(  \fluc_{w ,\Lambda_L}(\eta)  < \frac{C}{\sqrt[8]{\ln \ln L}}  \right) \geq  1 - \exp \left( - c \sqrt{\ln L} \right).
    \end{equation}
To this end, we set $\delta := (C_0 / \sqrt[8]{\ln \ln L}) \wedge 2^{m/2+1}$, for some large constant $C_0 \geq 8$, and define the following notion of good box: a box $\Lambda \subseteq \Lambda_L$ is said to be good if and only if
\begin{equation} \label{eq:10042}
    \fluc_{w ,\Lambda}(\eta) \leq \delta \hspace{3mm} \mbox{or} \hspace{3mm} \frac{1}{\left| \Lambda \right|}\sum_{v \in \Lambda} |w(v)|^2  \leq \frac{1}{4\sqrt{\ln \ln L}}.
\end{equation}
Let us note that, for any box $\Lambda \subseteq \Lambda_L$, the assumption $\left| f(\sigma) \right| \leq 1$ and the Cauchy-Schwarz inequality yield
\begin{align*}
\fluc_{w,\Lambda}(\eta)  & \leq \left( \frac{1}{\left| \Lambda \right|} \sum_{v \in \Lambda} |w(v)|^2 \right)^{\frac 12} \left(\sup_{\tau_1, \tau_2 \in \S^{\Z^2 \setminus \Lambda}}\frac{1}{|\Lambda|}\sum_{v \in \Lambda}  \left| \left\langle f_v \left(  \sigma \right) \right\rangle_{\Lambda}^{\tau_1} - \left\langle f_v\left( \sigma \right) \right\rangle_{\Lambda}^{\tau_2} \right|^2 \right)^\frac 12\\
& \leq 2\sqrt{ \frac{1}{\left| \Lambda \right|} \sum_{v \in \Lambda} |w(v)|^2 } \\
& \leq \frac{1}{\sqrt[4]{\ln \ln L}}.
\end{align*}
Using that $\fluc_{w,\Lambda}(\eta) \leq 2^{m/2+1}$, we obtain that
\begin{equation*}
\fluc_{w,\Lambda}(\eta) \leq \frac{1}{\sqrt[4]{\ln \ln L}} \wedge 2^{m/2+1} \leq \delta.
\end{equation*}
This implies that, if a box $\Lambda$ is good, then $\fluc_{w,\Lambda}(\eta) \leq \delta$ (as both options in~\eqref{eq:10042} lead to this inequality).

Following the argument presented in the proof of Lemma~\ref{prop2.31708}, we may construct a collection $\mathcal{Q}$ of good boxes such that the set of uncovered points is small. The only difference is that we need to select the integer $k$ so as to satisfy the (more restrictive) properties:
\begin{equation*}
    l_{\max} \geq 1 \hspace{3mm} \mbox{and} \hspace{3mm} \forall v \in \Lambda_L, \forall l \in \{ 0 , \ldots , l_{\max}-1 \}, \hspace{3mm} \left| \Lambda_{l+1} (v)\right| \leq \left( \frac{C_0}{8\sqrt{\ln \ln L}} \wedge 1 \right) \left| \Lambda_l (v)\right|.
\end{equation*} 
Proceeding this way, we obtain the inequality
\begin{equation} \label{eq:0957}
    \P \left( v ~\mbox{is uncovered}\right) \leq \exp \left( - c \sqrt{\ln L}\right).
\end{equation}
We now use the estimate~\eqref{eq:0957} to prove the inequality~\eqref{eq:18033}. Using the domain subadditivity property for the quantity $\fluc_{w ,\Lambda}$ and the upper bound $\fluc_{w ,\Lambda} \leq 2$, we have
\begin{align} \label{eq:17087}
    \fluc_{w ,\Lambda_L }(\eta) & \leq \sum_{\Lambda \in \mathcal{Q}} \frac{\left| \Lambda \right|}{\left| \Lambda_L \right|}  \fluc_{w ,\Lambda}(\eta) +  2\frac{\left| \Lambda_L \setminus \bigcup_{\Lambda \in \Q} \Lambda \right|}{|\Lambda_L|} \\
    & \leq  \delta + 2\frac{\left| \Lambda_L \setminus \bigcup_{\Lambda \in \Q} \Lambda \right|}{|\Lambda_L|}. \notag
\end{align}
We then estimate the second term in the right side of~\eqref{eq:17087} by combining Markov's inequality with the estimate~\eqref{eq:0957} as was done in the computation~\eqref{eq:1304}.
The concentration argument is essentially identical to the one presented in the proof of Theorem~\ref{thm1prop2.31708main}, we thus omit the details.
\end{proof}

We conclude this section by recording a stronger version of Theorem~\ref{prop:1109} where the weight are allowed to be (partially) random. The result is used in the proof of Theorem~\ref{theoremuniqueness}. Before stating the result, we introduce the following definition. We fix $\ep \in (0 , 1)$ and, for $v \in \Z^2$, let $\eta_{v}^\ep, \eta_{v}^{1-\ep}$ be $m$-dimensional Gaussian random vectors of mean zero and covariance matrix $\ep I_m$ and $(1-\ep) I_m$, respectively. We assume that the Gaussian vectors are independent, and set $\eta := \eta^\ep + \eta^{1-\ep}$. In this setup, $\eta$ satisfies the assumptions required by the random disorder. We also let $\mathcal{F}_{1-\ep}$ be a $\sigma$-algebra such that $\eta^{1-\ep}$ is $\mathcal{F}_{1-\ep}$-measurable and $\eta^{\ep}$ is $\mathcal{F}_{1-\ep}$-independent. Finally, we define the event
\begin{equation} \label{eq:17580308}
\mathcal{W}_{1 - \ep,L,R} = \left\{\sum_{v \in (x + \Lambda_{L_0}) \setminus (x + \Lambda_{(L_0-R)})} |\eta^{1-\ep}_v| \leq 16 R L_0 \, : \, \forall x \in \Z^2, L_0 \in [\sqrt{L} , L] ~\mbox{such that}~ x + \Lambda_{L_0} \subseteq \Lambda_{L} \right\}.
\end{equation}
This events allows to control the $L^1$-norm of the field $\eta^{1-\ep}$ in the boundary layer of every cube $\Lambda$ whose sidelength is at least $\sqrt{L}$ and is contained in the cube $\Lambda_L$. Whenever $\eta^{1-\ep} \in \mathcal{W}_{1 - \ep,L,R}$, the maximal contribution of this part of the field to the Hamiltonian restricted to a box of size $L_0$ cannot exceed a constant multiple of $L_0$, mimicking the boundary dependence of the deterministic portion of the Hamiltonian.

We additionally note that, since the number of boxes contained in the box $\Lambda_L$ with side length larger than $\sqrt{L}$ grows as a power of $L$, a union bound combined with a large deviation estimate implies the following lower bound on the probability of the event $\mathcal{W}_{1 - \ep,L,R}$
\begin{equation} \label{eq:17590308}
    \P \left( \mathcal{W}_{1 - \ep,L,R} \right) \geq 1 - C \exp (- c \sqrt{L}).
\end{equation}

The next proposition extends the result of Theorem~\ref{prop:1109} by allowing the weight function to be partially random.

\begin{proposition} \label{prop:7.8august}
    Fix $d = 2$, $\beta > 0$, $\lambda > 0$ and $L \geq 3$. Fix $\ep \in (0 , 1)$ and consider the two random disorders $\eta^{1 - \ep}$ and $\eta^{\ep}$ as defined above. There exist constants $C,c>0$ depending only on $\lambda$, $C_H$, $m$ and $R$ such that, for any integer $L \geq 3$, any realization $\eta^{1- \ep}$ in  $\mathcal{W}_{1 - \ep,L,R}$, and any $\mathcal{F}_{1-\ep}$-measurable  weight function $w:\Lambda_L\to[-1,1]^m$,
    \begin{multline*}
\P \left( \sup_{\tau_1, \tau_2 \in \S^{\Z^2 \setminus \Lambda_L}}  \left|\frac{1}{|\Lambda_L|}\sum_{v \in \Lambda_L} w(v) \cdot \left(\left\langle f_v \left( \sigma \right) \right\rangle_{\Lambda_L}^{\tau_1} - \left\langle f_v \left( \sigma \right) \right\rangle_{\Lambda_L}^{\tau_2}  \right) \right|  \leq \frac{C}{\sqrt{\ep}\sqrt[8]{\ln \ln L}} ~\bigg|~ \mathcal{F}_{1-\ep} \right) \\ \geq 1 -  \exp \left( - c L \right).
\end{multline*}
\end{proposition}

\begin{proof}
Fix a finite set $\Lambda \subseteq \Zd$. Using the decomposition $\eta := \eta^{1-\ep} + \eta^\ep$, we may write the disordered Hamiltonian $H^\eta_{\Lambda}$ as follows
\begin{equation*}
    H^\eta_{\Lambda}(\sigma) = H_{\Lambda}(\sigma) - \lambda \sum_{v \in \Lambda} \eta^{1 - \ep} \cdot f_v(\sigma) - \lambda \sum_{v \in \Lambda} \eta^{\ep} \cdot f_v(\sigma).
\end{equation*}
For any realization of the random field $\eta^{1 - \ep}$, let us define the Hamiltonian $\tilde{H}_\Lambda(\sigma) = H_\Lambda(\sigma) - \lambda \sum_{v \in \Lambda} \eta^{1 - \ep} \cdot f_v(\sigma)$. Setting $\tilde{\eta} = \eta^{\ep}/\sqrt{\ep}$, we thus have
\begin{equation} \label{eq:reptildeH}
H_\Lambda^{\eta}(\sigma) = \tilde{H}_\Lambda(\sigma) - \lambda \sum_{v \in \Lambda} \tilde{\eta}_v \cdot (\sqrt{\ep} f_v(\sigma)).
\end{equation}
Fix a sidelength $L \geq 3$. The definition of the event $\mathcal{W}_{1 - \ep,L,R}$ ensures that the following property holds: for each $\eta^{1-\ep} \in \mathcal{W}_{1 - \ep,L,R}$ and each box $\Lambda' \subseteq \Lambda_L$ of sidelength at least $\sqrt{L}$, we have that
\begin{equation} \label{eq:17480208}
|\tilde{H}_{\Lambda'}(\sigma) - \tilde{H}_{\Lambda'}(\sigma')| \leq \left(C_H + \lambda 16 R \right) |\partial \Lambda'|, \quad \text{for}~ \sigma,\sigma': \Z^2 \mapsto \mathcal{S} \text{ satisfying } \sigma_{\Lambda'} = \sigma'_{\Lambda'}.
\end{equation}

The representation~\eqref{eq:reptildeH} is the sum of a $\mathcal{F}_{1 - \ep}$-measurable Hamiltonian 
and a $\mathcal{F}_{1 - \ep}$-independent Gaussian random vector of mean zero and covariance matrix $I_m$. For any fixed realization $\eta^{1-\ep} \in \mathcal{W}_{1 - \ep,L,R}$, the Hamiltonian $\tilde H$ satisfies all the assumptions listed in Section~\ref{sec:2d systems} except one: the bound~\eqref{def.cteCH} does not hold for all finite subsets $\Lambda \subseteq \Zd$, but, by the inequality~\eqref{eq:17480208}, only holds when the subset $\Lambda$ is a box of sidelength at least $\sqrt{L}$ and contained in $\Lambda_L$. An inspection of the proof of Theorem~\ref{thm1prop2.31708main} (in particular, the arguments developed in Section~\ref{sectionfirstmandelbrotperc}) shows that this property is sufficient for the conclusion of Theorem~\ref{thm1prop2.31708main} to hold. A similar argument shows that Theorem~\ref{prop:1109} holds under the assumption~\eqref{eq:17480208}. We thus conclude that, for any $\eta^{1- \ep} \in \mathcal{W}_{1 - \ep,L,R}$,
\begin{multline*}
\P \left( \sup_{\tau_1, \tau_2 \in \S^{\Z^2\setminus \Lambda_L}}  \left|\frac{1}{|\Lambda_L|}\sum_{v \in \Lambda_L} w(v) \cdot \left(\left\langle  \sqrt{\ep} f_v \left( \sigma \right) \right\rangle_{\Lambda_L}^{\tau_1} - \left\langle \sqrt{\ep} f_v \left( \sigma \right) \right\rangle_{\Lambda_L}^{\tau_2}  \right) \right|  \leq \frac{C}{\sqrt[8]{\ln \ln L}} ~\bigg|~ \mathcal{F}_{1-\ep} \right) \\ \geq 1 -  \exp \left( - c L \right).
\end{multline*}
Dividing through by $\sqrt{\ep}$ in the probability completes the proof.
\end{proof}

\subsection{Proof of Corollary~\ref{prop:1109cor}} \label{proofTHM3}
As a corollary of Theorem~\ref{prop:1109}, we show that the absolute value of the expectation (with respect to the random field) of the thermal expectations $\left\langle f_v \left(  \sigma \right) \right\rangle_{\Lambda_L}^{\tau_0}$ and $\left\langle f_v \left(  \sigma \right) \right\rangle_{\Lambda_L}^{\tau_1}$ is quantitatively small for any pair of random boundary conditions $\eta \mapsto \tau_0(\eta), \tau_1(\eta)$.

 \begin{proof}[Proof of Corollary~\ref{prop:1109cor}]
We select a pair of random (measurable) boundary conditions $\eta \mapsto \tau_0 \left( \eta \right), \tau_1\left( \eta \right) \in \S^{\Z^2 \setminus \Lambda_L}$, and define the deterministic weight function $w$ according to the formula
\begin{equation} \label{eq:08540111}
    \forall v \in \Lambda_L, \forall i \in \{1 , \ldots, m \}, ~ w_{i}(v) := \mathrm{sign} \left(\E \left[ \left\langle f_{v,i} \left(  \sigma \right) \right\rangle_{\Lambda_L}^{\tau_0 \left( \eta \right)} - \left\langle f_{v,i} \left(  \sigma \right) \right\rangle_{\Lambda_L}^{\tau_1 \left( \eta \right)} \right] \right),
\end{equation}
where $\mathrm{sign}(x) = 1$ for $x \geq 0$ and $-1$ otherwise. Applying Theorem~\ref{prop:1109} yields
\begin{equation} \label{eq:09241606}
\E \left[  \frac{1}{\left| \Lambda_L \right|} \sum_{v \in \Lambda_L} w(v) \cdot \left(  \left\langle f_v \left( \sigma \right) \right\rangle_{\Lambda_L}^{\tau_0(\eta)} - \left\langle f_v \left( \sigma \right) \right\rangle_{\Lambda_L}^{\tau_1(\eta)} \right) \right]  \leq \E \left[  \fluc_{w ,\Lambda_L }(\eta) \right] \leq \frac{C}{\sqrt[8]{\ln \ln L}}.
\end{equation}
We can then estimate the left-hand side of~\eqref{eq:09241606}. We obtain
\begin{align*}
    \E \left[  \frac{1}{\left| \Lambda_L \right|} \sum_{v \in \Lambda_L} w(v) \cdot \left(   \left\langle f_v \left(  \sigma \right) \right\rangle_{\Lambda_L}^{\tau_0(\eta)} - \left\langle f_v \left(  \sigma \right) \right\rangle_{\Lambda_L}^{\tau_1(\eta)} \right) \right] & =  \frac{1}{\left| \Lambda_L \right|} \sum_{v \in \Lambda_L} w(v) \cdot \E \left[  \left\langle f_v \left(  \sigma \right) \right\rangle_{\Lambda_L}^{\tau_0(\eta)} - \left\langle f_v \left( \sigma \right) \right\rangle_{\Lambda_L}^{\tau_1(\eta)} \right] \\
    & \geq \frac{1}{\left| \Lambda_L \right|} \sum_{v \in \Lambda_L}\left|\E \left[ \left\langle f_v \left(  \sigma \right) \right\rangle_{\Lambda_L}^{\tau_0(\eta)} - \left\langle f_v \left( \sigma \right) \right\rangle_{\Lambda_L}^{\tau_1(\eta)}  \right] \right|.
\end{align*}
A combination of the two previous displays completes the proof of Corollary~\ref{prop:1109cor}.
\end{proof}

\subsection{Proof of Theorem~\ref{theoremuniqueness}} \label{proofofTheorem3}

This section contains the proof of Theorem~\ref{theoremuniqueness}. The argument relies on an application of Proposition~\ref{prop:7.8august} combined with Levy's zero-one law.

\begin{proof}[Proof of Theorem~\ref{theoremuniqueness}]
For every $v \in \Zd$ and $k \geq 1$, let $G_{v,k}$ be $m$-dimensional Gaussian random vectors of mean zero and covariance matrix $2^{-k} I_m$. The random vectors are assumed to be independent. The fields $\eta$ and $\eta^{(\ell)}$ are defined by 
\[\eta_v = \sum_{k=1}^\infty G_{v,k}, \quad \text{and} \quad \eta_v^{(\ell)} = \sum_{k=1}^\ell G_{v,k}.
\]
We note that the covariance matrix of $\eta$ is $I_m$, that $\eta - \eta^{(\ell)}$ is a normal random vector of mean zero and covariance matrix $2^{-\ell} I_m$ which is independent of $\eta^{(\ell)}$, and that $\eta$ is measurable with respect to the $\sigma$-algebra generated by $\{\eta^{(\ell)}\}_{\ell \geq 1}$. For $\ell \in \N \cup \{ \infty \}$, we denote by $\mathcal{F}_\ell$ the $\sigma$-algebra generated by the collection of random variables $(\eta^{(k)})_{k \leq \ell}. $

Let $\mu_1$ and $\mu_2$ be two translation-covariant Gibbs measures, and define the weight function $w_{\eta^{(\ell)}}: \Zd \to \R^m$ according to the formula
\begin{equation} \label{eq:15352507}
\forall v \in \Zd, \forall i \in \{1 , \ldots, m \}, ~ w_{\eta^{(\ell)}, i} (v) = \mathrm{sign} \left(\E \left[ \left\langle f_{v,i} \left(  \sigma \right) \right\rangle_{\mu_1^\eta} - \left\langle f_{v,i} \left(  \sigma \right) \right\rangle_{\mu_2^\eta} ~\big|~ \mathcal{F}_\ell \right] \, \right),
\end{equation}
where $\mathrm{sign}(x) = 1$ for $x \geq 0$ and $-1$ otherwise. 

Since we assume that $\mu_1$ and $\mu_2$ are translation-covariant, it is straightforward to see that $w_{\eta^{(\ell)}}$ is also translation-covariant --- that is,
\begin{equation} \label{eq:18012507}
   \mathcal{T}_v w_{\eta^{(\ell)}} =  w_{\mathcal{T}_v\eta^{(\ell)}}.
\end{equation}
For notational convenience, we introduce the event 
\begin{equation*}
    E_\ell := \left\{ \sup_{\tau_1, \tau_2 \in \S^{\Z^2 \setminus \Lambda_L}}  \left|\frac{1}{|\Lambda_L|}\sum_{v \in \Lambda_L} w_{\eta^{(\ell)}}(v) \cdot \left(\left\langle f_v \left( \sigma \right) \right\rangle_{\Lambda_L}^{\tau_1} - \left\langle f_v \left( \sigma \right) \right\rangle_{\Lambda_L}^{\tau_2} \right) \right|  \leq \frac{C}{2^{-\frac \ell 2}\sqrt[8]{\ln \ln L}} \right\}.
\end{equation*}
Applying Proposition~\ref{prop:7.8august} with the value $\ep := 2^{-\ell}$, recalling the definition of the event ~\eqref{eq:17580308} and the lower bound~\eqref{eq:17590308}, we deduce that
\begin{align} \label{eq:16552507}
    \P(E_\ell) & = \E[ \P ( E_\ell \mid \mathcal{F}_\ell ) ] \\
    & \geq \E[ \indc_{\mathcal{W}_{1 - \ep,L,R}} \P ( E_\ell \mid \mathcal{F}_\ell ) ] \notag \\
    & \geq \P(\mathcal{W}_{1 - \ep,L,R})  (1 -  \exp \left( - c L \right))
    \geq 1 - C \exp \left( - c \sqrt{L} \right).\notag
\end{align}
Using that the weight function $w_{\eta^{(\ell)}}$ and the observable $f_v$ are bounded by $1$, we obtain the upper bound
\begin{align} \label{eq:17162507}
    \E \left[ \sup_{\tau_1, \tau_2 \in \S^{\Z^2 \setminus \Lambda_L}}  \left|\frac{1}{|\Lambda_L|}\sum_{v \in \Lambda_L} w_{\eta^{(\ell)}}(v) \cdot \left(\left\langle f_v \left( \sigma \right) \right\rangle_{\Lambda_L}^{\tau_1} - \left\langle f_v \left( \sigma \right) \right\rangle_{\Lambda_L}^{\tau_2} \right) \right| \right] & \leq \frac{C}{2^{-\ell/2}\sqrt[8]{\ln \ln L}} \P \left[ E_\ell \right] + C (1 - \P \left[ E_\ell \right]) \notag \\
    & \leq \frac{C}{2^{-\ell/2} \sqrt[8]{\ln \ln L}} + C \exp (-c\sqrt{L}) \\
    & \leq \frac{C}{2^{-\ell/2}\sqrt[8]{\ln \ln L}}. \notag
\end{align}
The consistency relations~\eqref{eq:14042507} for the infinite-volume Gibbs measures imply, for any sidelength $L \geq 3$ and any realization of the random field, 
\begin{multline} \label{eq:17172507}
    \frac{1}{|\Lambda_L|}\sum_{v \in \Lambda_L} w_{\eta^{(\ell)}}(v) \cdot \left(\left\langle f_v \left( \sigma \right) \right\rangle_{\mu_1^\eta} - \left\langle f_v \left( \sigma \right) \right\rangle_{\mu_2^\eta} \right)  \\ \leq \sup_{\tau_1, \tau_2 \in \S^{\Z^2 \setminus \Lambda_L}}  \left|\frac{1}{|\Lambda_L|}\sum_{v \in \Lambda_L} w_{\eta^{(\ell)}}(v) \cdot \left(\left\langle f_v \left( \sigma \right) \right\rangle_{\Lambda_L}^{\tau_1} - \left\langle f_v \left( \sigma \right) \right\rangle_{\Lambda_L}^{\tau_2} \right) \right|,
\end{multline}
and therefore 
\begin{equation} \label{eq:10332607}
    \E \left[ \frac{1}{|\Lambda_L|}\sum_{v \in \Lambda_L} w_{\eta^{(\ell)}}(v) \cdot \left(\left\langle f_v \left( \sigma \right) \right\rangle_{\mu_1^\eta} - \left\langle f_v \left( \sigma \right) \right\rangle_{\mu_2^\eta} \right)  \right] \leq \frac{C}{2^{-\ell/2} \sqrt[8]{\ln \ln L}}.
\end{equation}
By the translation-covariance of $w_{\eta^{(\ell)}}$ and of the measures $\mu_1$ and $\mu_2$, we see that, $\P$-almost-surely,
\[
w_{\eta^{(\ell)}}(v) \cdot \left(\left\langle f_v \left( \sigma \right) \right\rangle_{\mu_1^\eta} - \left\langle f_v \left( \sigma \right) \right\rangle_{\mu_2^\eta} \right) = w_{\mathcal{T}_{-v} \eta^{(\ell)}}(0) \cdot  \left(\left\langle f_\mathbf{0} \left(\sigma \right) \right\rangle_{\mu_1^{\mathcal{T}_{-v} \eta}} - \left\langle f_\mathbf{0} \left( \sigma \right) \right\rangle_{\mu_2^{\mathcal{T}_{-v}\eta}} \right). 
\]
Since the field $\eta$ is made up of i.i.d. random variables, the expectation of the righthand side over $\eta$ is independent of $v$. Therefore, for any $L \geq 3$,
\begin{align} \label{eq:10322607}
\E \left[ w_{\eta^{(\ell)}}(\mathbf{0}) \cdot \left(\left\langle f_\mathbf{0} \left(\sigma \right) \right\rangle_{\mu_1^{\eta}} - \left\langle f_\mathbf{0} \left( \sigma \right) \right\rangle_{\mu_2^{\eta}} \right) \right] & =     \E \left[\frac{1}{|\Lambda_L|}\sum_{v \in \Lambda_L} w_{\eta^{(\ell)}}(v) \cdot \left(\left\langle f_v \left( \sigma \right) \right\rangle_{\mu_1^\eta} - \left\langle f_v \left( \sigma \right) \right\rangle_{\mu_2^\eta} \right) \right]  \\ & \leq \frac{C}{2^{-\ell/2} \sqrt[8]{\ln \ln L}}  . \notag
\end{align}
Taking the limit as $L$ goes to infinity, we deduce that, for any $\ell \geq 1$, 
\begin{equation*}
     \E \left[ w_{ \eta^{(\ell)}}(\mathbf{0}) \cdot \left(\left\langle f_\mathbf{0} \left(\sigma \right) \right\rangle_{\mu_1^{\eta}} - \left\langle f_\mathbf{0} \left( \sigma \right) \right\rangle_{\mu_2^{\eta}} \right) \right] = 0.
\end{equation*}
Since $\eta$ is measurable with respect to the $\sigma$-algebra $\mathcal{F}_\infty$, Levy's zero-one law implies that, almost surely on the set $\{ \left\langle f_{\mathbf{0},i} \left(  \sigma \right) \right\rangle_{\mu_1^\eta} - \left\langle f_{\mathbf{0},i} \left(  \sigma \right) \right\rangle_{\mu_2^\eta} \neq 0 \}$ (so that the sign function is continuous),
\begin{align*}
\lim_{\ell \to \infty}\mathrm{sign} \left(\E \left[ \left\langle f_{\mathbf{0},i} \left(  \sigma \right) \right\rangle_{\mu_1^\eta} - \left\langle f_{\mathbf{0},i} \left(  \sigma \right) \right\rangle_{\mu_2^\eta} ~\big|~ \mathcal{F}_\ell \right] \, \right) & = \mathrm{sign} \left( \mathbb{E} \left[ \left\langle f_{\mathbf{0},i} \left(  \sigma \right) \right\rangle_{\mu_1^\eta} - \left\langle f_{\mathbf{0},i} \left(  \sigma \right) \right\rangle_{\mu_2^\eta} ~\big|~ \mathcal{F}_\infty \right] \, \right) \\ & = \mathrm{sign}\left(\left\langle f_{\mathbf{0},i} \left(  \sigma \right) \right\rangle_{\mu_1^\eta} - \left\langle f_{\mathbf{0},i} \left(  \sigma \right) \right\rangle_{\mu_2^\eta} \right),
\end{align*}
as the thermal expectations of the noised observables $f_{\mathbf{0},i}$ are $\eta$-measurable. Since both $w$ and $f_\mathbf{0}$ are bounded by 1, the dominated convergence theorem implies that
\begin{equation} \label{eq:26071050}
\E \left[ \left| \left\langle f_\mathbf{0} \left(\sigma \right) \right\rangle_{\mu_1^{\eta}} - \left\langle f_\mathbf{0} \left( \sigma \right) \right\rangle_{\mu_2^{\eta}} \right| \right] = \lim_{\ell \to \infty }  \E \left[ w_{ \eta^{(\ell)}}(\mathbf{0}) \cdot \left(\left\langle f_\mathbf{0} \left(\sigma \right) \right\rangle_{\mu_1^{\eta}} - \left\langle f_\mathbf{0} \left( \sigma \right) \right\rangle_{\mu_2^{\eta}} \right) \right]  = 0 .
\end{equation}
This implies that the thermal expectation of the noised observable $f_\mathbf{0}(\sigma)$ is the same for any pair of translation-covariant infinite-volume Gibbs measures, and thus completes the proof. 
\end{proof}

\section{Proofs for spin systems with continuous symmetry} \label{SectionCSS}

In this section, we study the spin systems with continuous symmetry presented in Section~\ref{sec:continuous symmetry results} and prove Theorem~\ref{prop:subcritical} and Theorem~\ref{thm.thm2}. The section is organized as follows:
\begin{itemize}[leftmargin=*]
\item In Section~\ref{sec.varlemmaCSS}, we establish two variational lemmas on the set of bounded functions whose integral on any interval of the real line $\R$ is bounded.
\item In Section~\ref{secMW}, we implement a Mermin-Wagner type argument to prove an upper bound on the free energy (see Proposition~\ref{propMermW}).
\item Section~\ref{section4.3CSSsubcrit} is devoted to the proof of Theorem~\ref{prop:subcritical} in the subcritical dimensions $d = 1,2,3$. In Subsection~\ref{subsec4.3.11739}, we combine the Mermin-Wagner upper bound obtained in Proposition~\ref{propMermW} with Proposition~\ref{labelMLEmma} to prove the algebraic decay of the thermally and spatially averaged magnetization with fixed boundary condition stated in~\eqref{eq:TV0904Final}. In Subsection~\ref{subsec4.3.11739}, we build upon the results of Subsection~\ref{subsec4.3.11739} and prove the inequality~\eqref{eq:TV090404final}, thus completing the proof of Theorem~\ref{prop:subcritical}.
\item Section~\ref{criticdimesninosCSS} is devoted to the proof of Theorem~\ref{thm.thm2} following the outline of Section~\ref{Secstratproof} and is divided into three subsections. In Subsection~\ref{lemma.08538}, we combine the Mermin-Wagner upper bound of Proposition~\ref{propMermW} and Proposition~\ref{LemmavarprinCSS2} and establish that, given a box $\Lambda \subseteq \Lambda_L$, if the averaged field $\hat \eta_{\Lambda,i}$ is negative enough, then the thermally and spatially averaged magnetization must be small (see Lemma~\ref{lemma.08538}). In Subsection~\ref{subsec4.3.11739}, we combine the result of Proposition~\ref{labelMLEmma} with a Mandelbrot percolation argument and obtain the quantitative estimate stated in Lemma~\ref{lemmasecondmandelbrot} on the expectation of the spatially and thermally averaged magnetization with a fixed boundary condition. Finally in Subsection~\ref{subsecprooftheorem2}, we upgrade the result of Lemma~\ref{lemmasecondmandelbrot} to include a supremum over all the possible boundary conditions, and complete the proof of Theorem~\ref{thm.thm2}.
\end{itemize}

\subsection{Variational lemmas} \label{sec.varlemmaCSS}
In this section, we state and prove two variational lemmas on the set of measurable bounded functions defined on $\R$ and whose integral on every interval is bounded in absolute value by $1$, i.e.,
\begin{equation} \label{def.setF}
\mathcal{G} := \left\{ g : \R \to \R \, : \, g \mbox{ is measurable, bounded and, for any real interval} ~ I \subseteq \R,~ \left| \int_I g(t) dt \right| \leq 1 \right\}.
\end{equation}

The first result we asserts that the Gaussian expectation of any map $g \in \mathcal{G}$ is bounded by an explicit constant.

\begin{proposition} \label{labelMLEmma}
One has the inequality
\begin{equation} \label{eq:105720}
    \sup_{g \in \mathcal{G}} \left| \int_\R g(t) e^{-\frac{t^2}2} \, dt \right| \leq 2 \int_0^\infty t e^{-\frac{t^2}2} \, dt.
\end{equation}
\end{proposition}

\begin{proof}
We fix a function $g \in \mathcal{G}$, and let~$G(t):= \int_0^t g(s) ds$.
By the definition of the set $\mathcal{G}$, the map $G$ satisfies $\left| G(t) \right| \leq 1$ for any $t \in \R$, the identity $G(0) = 0$, and is Lipschitz continuous. By performing an integration by parts, we obtain
\begin{align*}
      \hspace{25mm} \left|\int_\R g(t) e^{-\frac{t^2}{2}} \, dt \right|  =  \left| \int_\R G(t) t e^{-\frac{t^2}{2}} \, dt \right|
    \leq  \int_\R | t| e^{-\frac{t^2}{2}} \, dt = 2 \int_0^\infty t e^{-\frac{t^2}{2}} \, dt. \hspace{25mm} \qedhere
\end{align*}
\end{proof}

The next proposition provides a lower bound on the Gaussian measure of the set $\left\{ g \leq \delta \right\}$, for any map $g \in \mathcal{G}$ satisfying $g \geq -1$, and any $\delta >0$.

\begin{proposition} \label{LemmavarprinCSS2}
There exists a constant $C > 0$ such that, for any $\delta \in (0,1]$,
\begin{equation} \label{eq:0951}
    \inf_{\substack{g \in \mathcal{G} \\ g \geq -1}} \int_\R \indc_{\left\{ g(t) \leq \delta\right\}} e^{-t^2/2} \, dt \geq  e^{-C/\delta^2}.
\end{equation}
\end{proposition}

\begin{proof}
We select a function $g \in \mathcal{G}$ and let $G_{0}, G: \R \to \R$ be the maps defined by the formulas
\begin{equation*}
    \forall t \in \R , \hspace{5mm} G_{0}(t):= \int_0^{t} \indc_{\left\{ g(s) \leq \delta\right\}} d s
\hspace{5mm} \mbox{and} \hspace{5mm}
G(t):= \int_{-t}^{t} \indc_{\left\{ g(s) \leq \delta\right\}} d s.
\end{equation*}
Note that the functions $G_{0}$ and $G$ are increasing, $1$ and $2$-Lipschitz continuous respectively, and satisfy the identity $G(t) = G_{0}(t) - G_{0}(-t)$. By performing an integration by parts and a change of variable, we see that
\begin{align} \label{eq:0946}
    \int_\R \indc_{\left\{ g(t) \leq \delta\right\}} e^{-\frac{t^2}{2}} \, dt = \int_\R G_{0}(t) t e^{-\frac{t^2}{2}} \, dt & =  \int_0^\infty \left( G_{0}(t) - G_{0}(-t) \right) t e^{-\frac{t^2}{2}} \, dt \\ & =  \int_0^\infty G(t) t e^{-\frac{t^2}{2}} \, dt. \notag
\end{align}
We next claim that the map $G$ satisfies the lower bound
\begin{equation} \label{keypropH}
    \forall t \geq 0, \, G \left( t\right) \geq \max \left( 0, \frac{2\delta t - 1}{1 + \delta} \right).
\end{equation}
To prove~\eqref{keypropH}, we use the assumption $g \geq -1$ and write
\begin{align*}
    g & \geq \delta \indc_{\{ g > \delta \}} - \indc_{\{ g \leq \delta \}}  \geq \delta \left( 1 - \indc_{\{ g \leq \delta \}} \right) - \indc_{\{ g \leq \delta \}} \geq \delta - (1 + \delta) \indc_{\{g \leq \delta \}}.
\end{align*}
Integrating this inequality over the interval $[-t ,t]$ and using the properties on the function $g$, we obtain
\begin{equation} \label{1918}
    1 \geq 2\delta t - (1 + \delta) G(t) \iff G(t) \geq \frac{2\delta t - 1}{1 + \delta}.
\end{equation}
We conclude the proof of~\eqref{keypropH} by using that $G$ is non-negative. A combination of~\eqref{eq:0946} and~\eqref{keypropH} implies the inequality
\begin{equation} \label{eq:0950}
    \inf_{\substack{g \in \mathcal{G} \\ g \geq -1}} \int_\R \indc_{\left\{ g(t) \leq \delta\right\}} e^{-\frac{t^2}{2}} \, dt \geq \int_{\frac{1}{2\delta}}^\infty  \frac{2\delta t - 1}{1 + \delta}  t e^{-\frac{t^2}{2}} \, dt \geq  e^{-C/\delta^2}.\qedhere
\end{equation}
\end{proof}

\subsection{A Mermin-Wagner upper bound for the free energy} \label{secMW}

In this section, we obtain an upper bound on the free energy of a spin system equipped with a continuous symmetry by implementing a Mermin-Wagner type argument. We recall that the spin space is assumed to be the sphere $\mathbb{S}^{n-1}$, for some $n \geq 2$, as well as the notation for conditional expectations and probabilities introduced in Section~\ref{sec.condexpect} (as they will be used frequently in the proofs below). Before stating the result, we introduce the following definition.

\begin{definition}[Free energy]
Let $\Lambda_0, \Lambda$ be two boxes of $\Zd$ such that $\Lambda \subseteq \Lambda_0.$ For any field $\eta : \Lambda_0 \to \R$, we denote by
\begin{equation} \label{def.tildeeta}
    \tilde \eta_{\Lambda_0 , \Lambda,v} := \left\{
    \begin{aligned}
    \eta_v &~\mbox{if}~ v \in \Lambda_0 \setminus \Lambda, \\
    -\eta_v &~\mbox{if}~ v \in \Lambda.
    \end{aligned}
    \right.
\end{equation}
We define the free energy, for any $\eta : \Lambda_0 \to \R$,
\begin{equation*}
    \tilde \FE_{\Lambda_0 ,\Lambda}^{\tau, h}(\eta)  := \FE_{\Lambda_0}^{\tau, h} \left(  \tilde \eta_{\Lambda_0, \Lambda}\right).
\end{equation*}
\end{definition}

The main result of this section is an upper bound on the difference of the free energies $\FE_{\Lambda_0}^{\tau,h}$ and $\tilde \FE_{\Lambda_0, \Lambda}^{\tau,h}$ conditionally on the values of the field in the box $\Lambda$ and outside the box $2\Lambda$.

\begin{proposition}[Mermin-Wagner upper bound for the energy] \label{propMermW}
Let $n\ge 2$, $d\in\{1,2,3,4\}$ and $i \in \{ 1 , \ldots , n\}$. Let $\beta > 0$ be the inverse temperature, $\lambda>0$ be the disorder strength and $h \in \R^n$ be the deterministic external field. Fix a box $\Lambda_0 \subseteq \Zd$ of side length $L$, and let $\tau \in \S^{\partial \Lambda_0}$ be a boundary condition. For any box $\Lambda$ of side length $\ell$ such that $2 \Lambda \subseteq \Lambda_0$, we have the estimate
\begin{equation} \label{eq:1602}
    \E \left[ \tilde \FE_{\Lambda_0 ,\Lambda}^{\tau,h}  - \FE_{\Lambda_0}^{\tau,h}  ~ \Big\vert ~ \eta_{\left(\Lambda_0 \setminus 2 \Lambda \right) \cup \Lambda , i}  \right] \leq C \frac{\ell^{d-2}}{L^d} + C \frac{\ell^d}{L^d} \left| h \right| \hspace{5mm} \P-\mbox{almost-surely},
\end{equation}
\end{proposition}

\begin{proof}
Let us fix two boxes $\Lambda_0, \Lambda$ satisfying $2\Lambda \subseteq \Lambda_0$, of side lengths $L$ and $\ell$ respectively, and a boundary condition $\tau \in \S^{\partial \Lambda_0}$. We recall the notations $e_1 , \ldots, e_n$ for the canonical basis of $\R^n$ and $\Lambda_0^+ := \Lambda_0 \cup \partial \Lambda_0$. All the configurations $\sigma \in \S^{\Lambda_0^+}$ in this proof are implicitly assumed to satisfy $\sigma_{\partial \Lambda_0} = \tau$.

Let us consider a smooth map $r : \R \to O(n)$, where $O(n)$ denotes the orthogonal group of $\mathbb{R}^{n}$, satisfying $r_0 = r_{2 \pi}= I_n$, for any pair $\theta_1 , \theta_2 \in \R$, $r_{\theta_1} \circ r_{\theta_2} = r_{\theta_1 + \theta_2}$, and such that $r_\pi(e_i) = - e_i$. For each vertex $v \in \Zd$, we denote by
\begin{equation} \label{def.thetaMW}
    \theta_v := \left\{
    \begin{aligned}
    0 \hspace{20mm}& \hspace{3mm} \mbox{if}~ v \in \Zd \setminus 2 \Lambda, \\
    \pi  \left(\frac{2 \dist \left( v , \partial \left(2\Lambda\right) \right)}{\ell} \wedge  1 \right) & \hspace{3mm}\mbox{if}~ v \in 2 \Lambda  .
    \end{aligned} \right.
\end{equation}
This definition implies that, for any $v \in \Lambda$, $\theta_v = \pi$.
We then define two rotations $R, \tilde R$ on the space of configurations by the formulas, for any $\sigma \in \S^{\Lambda_0^+}$ and any vertex $v \in \Zd$,
\begin{equation*}
    \left(R\sigma \right)_v = r_{\theta_v}\sigma_v \hspace{3mm} \mbox{and} \hspace{3mm}
    \left(\tilde R\sigma \right)_v = r_{-\theta_v}\sigma_v.
\end{equation*}
We extend the domain of the rotations $R, \tilde R$ to the set of fields, and write, for any realization of the random field $\eta$ and any vertex $v \in \Lambda_0$,
\begin{equation*}
    \left(R\eta \right)_v = r_{\theta_v}\eta_v \hspace{3mm} \mbox{and} \hspace{3mm}
    \left( \tilde R\eta \right)_v = r_{-\theta_v}\eta_v.
\end{equation*}
For any realization of the field $\eta$ and any configuration $\sigma \in \S^{\Lambda_0^+}$, we have the identities
\begin{equation} \label{eq:1534}
\sum_{v \in \Lambda_0} (R\eta)_v  \cdot (R\sigma)_v = \sum_{v \in \Lambda_0}  \eta_v  \cdot \sigma_v \hspace{3mm} \mbox{and} \hspace{3mm} \sum_{v \in \Lambda_0}  (\tilde R\eta)_v  \cdot (\tilde R\sigma)_v = \sum_{v \in \Lambda_0}  \eta_v  \cdot \sigma_v.
\end{equation}
Additionally, since the two rotations $R$ and $\tilde R$ are equal to the identity outside the box $2 \Lambda$, we have
\begin{equation} \label{eq:01200501}
    \left| \sum_{v \in \Lambda_0} h \cdot (R\sigma)_v - \sum_{v \in \Lambda_0} h \cdot \sigma_v\right| + \left| \sum_{v \in \Lambda_0} h \cdot (\tilde R\sigma)_v - \sum_{v \in \Lambda_0} h \cdot \sigma_v\right| \leq C \ell^d |h|.
\end{equation}
We next prove the inequality, for any configuration $\sigma \in \mathcal{S}^{\Lambda_0^+}$,
\begin{equation} \label{eq:1440}
    \sum_{\substack{ v , w \in \Lambda_0^+ \\ v\sim w}} \Psi \left( r_{\theta_v} \sigma_v , r_{\theta_w} \sigma_w \right) + \sum_{\substack{ v , w \in \Lambda_0^+ \\ v\sim w}}  \Psi \left( r_{-\theta_v} \sigma_v , r_{-\theta_w} \sigma_w \right) \leq 2 \sum_{\substack{ v , w \in \Lambda_0^+ \\ v\sim w}} \Psi \left( \sigma_{v} , \sigma_{w} \right) +C  \ell^{d-2},
\end{equation}
for some constant $C$ depending only on the map $\Psi$ and the rotation $r$. To prove the inequality~\eqref{eq:1440}, we first use that the map $\Psi$ is invariant under the rotations $r_{\theta_v}$ and $r_{-\theta_v}$, and the properties of the map $r$. We obtain
\begin{multline} \label{eq:155924}
     \sum_{\substack{ v , w \in \Lambda_0^+ \\ v\sim w}} \Psi \left( r_{\theta_v} \sigma_v , r_{\theta_w} \sigma_w \right) + \sum_{\substack{ v , w \in \Lambda_0^+ \\ v\sim w}}  \Psi \left( r_{-\theta_v} \sigma_v , r_{-\theta_w} \sigma_w \right)  \\ = \sum_{\substack{ v , w \in \Lambda_0^+ \\ v\sim w}} \Psi \left(  \sigma_v , r_{\theta_w - \theta_v} \sigma_w \right) +  \sum_{\substack{ v , w \in \Lambda_0^+ \\ v\sim w}} \Psi \left( \sigma_v , r_{ \theta_v -\theta_w } \sigma_w \right).
\end{multline}
Using that the map $\Psi$ is assumed to be twice continuously differentiable and bounded, that the map $r$ is smooth and that the state space $\mathbb{S}^{n-1}$ is compact, we can perform a Taylor expansion and obtain that there exists a constant $C$, depending on the maps $\Psi$ and $r$ such that, for each $\theta \in \R$, and each pair of spins $\sigma_1, \sigma_2 \in \mathbb{S}^{n-1}$,
\begin{equation} \label{eq:1537}
    \left| \Psi \left(  \sigma_1 , r_{\theta} \sigma_2 \right) +  \Psi \left( \sigma_1 , r_{ -\theta } \sigma_2 \right) - 2 \Psi \left(  \sigma_1 , \sigma_2 \right) \right| \leq C \theta^2.
\end{equation}
Applying the inequality~\eqref{eq:1537} with the values $\theta = \theta_v - \theta_w$, $\sigma_1 = \sigma_v$, $\sigma_2 = \sigma_w$,
and summing over all the pairs of neighboring vertices $v , w$ in $\Lambda_0^+$ yields
\begin{equation} \label{eq:154724}
    \sum_{\substack{ v , w \in \Lambda_0^+ \\ v\sim w}} \Psi \left(  \sigma_v , r_{\theta_w - \theta_v} \sigma_w \right) +  \sum_{\substack{ v , w \in \Lambda_0^+ \\ v\sim w}} \Psi \left( \sigma_v , r_{ \theta_v -\theta_w } \sigma_w \right) \leq 2 \sum_{\substack{ v , w \in \Lambda_0^+ \\ v\sim w}} \Psi \left(  \sigma_v , \sigma_w \right) + C \sum_{\substack{ v , w \in \Lambda_0^+ \\ v\sim w}} \left| \theta_v - \theta_w \right|^2.
\end{equation}
By the definition of the map $\theta$ stated in~\eqref{def.thetaMW}, we have, for any pair of neighboring vertices $v , w \in \Lambda_0^+$,
\begin{equation} \label{eq:154624}
    \left\{
    \begin{aligned}
    \left| \theta_v - \theta_w \right| &= 0 & \hspace{3mm} \mbox{if}~ v,w \in \Lambda \cup \left( \Lambda_0 \setminus 2 \Lambda \right),\\
    \left| \theta_v - \theta_w \right| &\leq \frac{C}{\ell} & \hspace{3mm} \mbox{if}~ \{ v, w \} \cap \left( 2\Lambda \setminus \Lambda \right) \neq \emptyset.
    \end{aligned}
    \right.
\end{equation}
Combining the estimates~\eqref{eq:154724},~\eqref{eq:154624}, and using that the volume of the annulus $ \left( 2\Lambda \setminus \Lambda \right)$ is of order $\ell^d$, we obtain
\begin{equation} \label{eq:155824}
    \sum_{\substack{ v , w \in \Lambda_0^+ \\ v\sim w}} \Psi \left(  \sigma_v , r_{\theta_w - \theta_v} \sigma_w \right) +  \sum_{\substack{ v , w \in \Lambda_0^+ \\ v\sim w}} \Psi \left( \sigma_v , r_{ \theta_v -\theta_w } \sigma_w \right) \leq 2 \sum_{\substack{ v , w \in \Lambda_0^+ \\ v\sim w}} \Psi \left(  \sigma_v , \sigma_w \right) + C\ell^{d-2}.
\end{equation}
Combining the identity~\eqref{eq:155924} and  the inequality~\eqref{eq:155824} completes the proof of~\eqref{eq:1440}. Combining the estimates~\eqref{eq:1534},~\eqref{eq:01200501} and~\eqref{eq:1440} with the definition of the noised Hamiltonian~\eqref{eq125924}, we have obtained the inequality: for any realization of the random field $\eta$ and any configuration $\sigma \in \S^{\Lambda_0^+}$,
\begin{equation} \label{eq:1557}
      H^{R\eta , h }_{\Lambda_0} \left( R\sigma \right) + H^{\tilde R\eta ,h }_{\Lambda_0} \left( \tilde R\sigma \right) \leq 2H^{\eta ,h}_{\Lambda_0} \left( \sigma \right) + C \ell^{d-2} + C \ell^d |h|.
\end{equation}
We now use the inequality~\eqref{eq:1557} to prove the estimate~\eqref{eq:1602}.
By the rotational invariance of the measure $\kappa$ and the Cauchy-Schwarz inequality, we obtain, for any realization of the field $\eta$,
\begin{align*}
    \lefteqn{\FE_{\Lambda_0}^{\tau,h} \left( R \eta  \right) +  \FE_{\Lambda_0}^{\tau,h} \left( \tilde R \eta \right)} \qquad & \\ & = -\frac{1}{\beta\left|\Lambda_0 \right|}\ln  \left[ \int_{\S^{\Lambda_0}} \exp \left( - \beta H^{R\eta , h }_{\Lambda_0} (\sigma) \right) \prod_{ v \in \Lambda_0} \kappa \left( d \sigma_v\right) \int_{\S^{\Lambda_0}} \exp \left( - \beta H^{\tilde R \eta,h }_{\Lambda_0}(\sigma) \right) \prod_{ v \in \Lambda_0} \kappa \left( d \sigma_v\right) \right] \\
    & = -\frac{1}{\beta\left|\Lambda_0 \right|}\ln  \left[ \int_{\S^{\Lambda_0}} \exp \left( - \beta H^{R\eta,h}_{\Lambda_0}(R \sigma) \right) \prod_{ v \in \Lambda_0} \kappa \left( d \sigma_v\right) \int_{\S^{\Lambda_0}} \exp \left( - \beta H^{\tilde R\eta,h}_{\Lambda_0}(\tilde R \sigma) \right) \prod_{ v \in \Lambda_0} \kappa \left( d \sigma_v\right) \right] \\
    & \leq -\frac{2}{\beta\left|\Lambda_0 \right|}\ln \int_{\S^{\Lambda_0}} \exp \left( - \beta \frac{ H^{ R\eta,h}_{\Lambda_0}(R\sigma) + H^{ \tilde R\eta,h}_{\Lambda_0}(\tilde R\sigma) }{2} \right) \prod_{ v \in \Lambda_0} \kappa \left( d \sigma_v\right).
\end{align*}
We then use the estimate~\eqref{eq:1557} and obtain
\begin{align} \label{eq:14120411}
   \frac{\FE_{\Lambda_0}^{\tau,h} \left( R \eta \right) +  \FE_{\Lambda_0}^{\tau,h} \left( \tilde R \eta  \right)}{2} & \leq -\frac{1}{\beta\left|\Lambda_0 \right|}\ln \int_{\S^{\Lambda_0}} \exp \left( - \beta H^{\eta, h}_{\Lambda_0}(\sigma) - \beta C \ell^{d-2} - \beta C \ell^d |h| \right) \prod_{ v \in \Lambda_0} \kappa \left( d \sigma\right) \notag \\
   & \leq \FE_{\Lambda_0}^{\tau, h} \left( \eta \right) + \frac{C \ell^{d-2}}{L^d} +  \frac{C\ell^d}{L^d} |h|. 
\end{align}
Moreover, by the definition of the rotations $R$ and $\tilde R$, we have $R\eta = \tilde R \eta$ in the box $\Lambda$, $(R \eta)_i = (\tilde R\eta)_i = -\eta_i$ along the $i$-th component of the field inside the box $\Lambda$, and $R \eta = \tilde R\eta = \eta$ inside the annulus $\left( \Lambda_0 \setminus 2 \Lambda \right)$. Using the rotational invariance of the law of the field $\eta$, we obtain the identities
\begin{equation} \label{eq:133920}
     \E \left[ \FE_{\Lambda_0}^{\tau,h} \left( R \, \cdot \right) ~ \Big\vert ~ \eta_{ \left(\Lambda_0 \setminus 2 \Lambda \right) \cup \Lambda,i }  \right] = \E \left[ \FE_{\Lambda_0}^{\tau,h} \left( \tilde R \, \cdot  \right) ~ \Big\vert ~ \eta_{\left( \Lambda_0 \setminus 2 \Lambda \right) \cup \Lambda,i }  \right] = \E \left[ \tilde \FE_{\Lambda_0, \Lambda}^{\tau,h} ~ \Big\vert ~ \eta_{\left( \Lambda_0 \setminus 2 \Lambda \right) \cup \Lambda,i }  \right].
\end{equation}
Taking the conditional expectation with respect to the field $\eta_{\left(\Lambda_0 \setminus 2 \Lambda \right) \cup \Lambda,i}$ in the inequality~\eqref{eq:14120411} and using the identity~\eqref{eq:133920} completes the proof of~\eqref{eq:1602}.
\end{proof}

\subsection{Proof of Theorem~\ref{prop:subcritical}} \label{section4.3CSSsubcrit}

In this section, we obtain an algebraic rate of convergence for the expectation of the spatially and thermally averaged magnetization in the subcritical dimensions $d \in \{ 1, 2 , 3\}$. We prove the following more refined version of Theorem~\ref{prop:subcritical}, which takes into account the dependence in the external magnetic field $h$. Theorem~\ref{prop:subcritical} is a direct consequence upon taking $|h| \leq L^{-2}$ and $|h| \leq L^{-1}$.

\begin{theorem} \label{prop:subcritical2}
Let $d \in \{ 1,2,3 \}$, $L \ge 2$ be an integer, $\lambda>0$ and $\beta > 0$ and $h \in \R^n$ be a magnetic field satisfying $|h|\leq 1$. Let $\tau \in \S^{\partial \Lambda_{2L}}$ be a boundary condition (which may be the free or periodic boundary condition) and set $\ell := |h|^{-\frac 12} \wedge L$. There exists a constant $C > 0$ depending on $\lambda$, $n$ and $\Psi$ such that,
\begin{equation} \label{eq:TV0904}
    \left| \E \left[\frac{1}{\left| \Lambda_\ell \right|}  \sum_{v \in  \Lambda_\ell } \left\langle \sigma_v \right\rangle^{\tau , h}_{\Lambda_{2L}} \right] \right|  \leq C \ell^{\frac d2 - 2}.
\end{equation}
Additionally, for any magnetic field $h \in \R^n$,
\begin{equation} \label{eq:TV090404}
    \E \left[ \sup_{\tau \in \S^{\partial \Lambda_L}} \left|\frac{1}{\left| \Lambda_L \right|}  \sum_{v \in  \Lambda_L} \left\langle \sigma_v \right\rangle^{\tau , h}_{\Lambda_L} \right| \right]  \leq C \left( |h| \vee L^{-1}  \right)^{\frac{4-d}{2(8-d)}},
\end{equation}
where the free and periodic boundary conditions are included in the supremum.
\end{theorem}

The proof of Theorem~\ref{prop:subcritical2} is decomposed into two subsections: in Subsection~\ref{subsec4.3.11739}, we establish the inequality~\eqref{eq:TV0904}, and in Subsection~\ref{subsec4.3.21739}, we prove the upper bound~\eqref{eq:TV090404}.

\subsubsection{Algebraic decay of the magnetization with fixed boundary condition} \label{subsec4.3.11739}

In this section, we combine Proposition~\ref{propMermW} with Proposition~\ref{labelMLEmma} to obtain the algebraic decay of the magnetization with a fixed boundary condition stated in~\eqref{eq:TV0904}.

\begin{proof}[Proof of Theorem~\ref{prop:subcritical2}: estimate~\eqref{eq:TV0904}]
Let us fix an integer $L \geq 2$, a boundary condition $\tau \in \S^{\partial \Lambda_{2L}}$, an external magnetic field $h \in \R^n$ such that $|h| \leq 1$ and an integer $i \in \{ 1 , \ldots, n\}$. We introduce the notation $\ell := |h|^{-\frac 12} \wedge L$. Applying Proposition~\ref{propMermW} with $\Lambda_0 = \Lambda_{2L}$ and $\Lambda = \Lambda_\ell$, we have the inequality
\begin{equation} \label{eq:1701}
    \E \left[ \tilde \FE_{\Lambda_{2L}, \Lambda_\ell}^{\tau,h}  - \FE_{\Lambda_{2L}}^{\tau,h}  ~ \Big\vert ~ {\hat\eta_{\Lambda_\ell,i}} \right] \leq  C \frac{\ell^{d-2}}{L^d}  \hspace{5mm} \mathbb{P}-\mbox{almost-surely}.
\end{equation}
Let us note that the conditional expectation depends only on the realization of the averaged field ${\hat\eta_{\Lambda}}$; it can thus be seen as a function defined on $\R$ and valued in $\R$ (see Section~\ref{sec.condexpect}). By the definition of the free energy $\tilde \FE_{\Lambda_{2L}, \Lambda_\ell}^{\tau,h} $ and the $\eta \to -\eta$ invariance of the law of the random field, we have the identity
\begin{equation} \label{eq:0851}
     \E \left[ \tilde \FE_{\Lambda_{2L}, \Lambda_\ell}^{\tau,h}  ~ \Big\vert ~ {\hat\eta_{\Lambda_\ell,i}} \right] \left(\hat\eta_{\Lambda_\ell,i}\right) =  \E \left[ \FE_{\Lambda_{2L}}^{\tau,h}  ~ \Big\vert ~ {\hat\eta_{\Lambda_\ell,i}} \right] \left(- \hat\eta_{\Lambda_\ell,i}\right) .
\end{equation}
To ease the notation, let us define the map $G : \R \to \R$ by the formula
\begin{equation*}
    G(\hat\eta_{\Lambda_\ell,i}) := \E \left[ \FE_{\Lambda_{2L}}^{\tau,h}  ~ \Big\vert ~ {\hat\eta_{\Lambda_\ell,i}} \right] \left(-\hat\eta_{\Lambda_\ell,i}\right) - \E \left[ \FE_{\Lambda_{2L}}^{\tau,h}  ~ \Big\vert ~ {\hat\eta_{\Lambda_\ell,i}} \right] \left(\hat\eta_{\Lambda_\ell,i}\right).
\end{equation*}
We note that, by Proposition~\ref{prop.basicpropfreeen} and the Gaussianity of the field, the derivative of the map $G$ is explicit and we have
\begin{equation} \label{eq:1907}
    G'(\hat\eta_{\Lambda_\ell,i}) = \lambda \E \left[ \frac{1}{\left| \Lambda_{2L}\right|} \sum_{v \in \Lambda_\ell} \left\langle \sigma_{v,i} \right\rangle^{\tau , h}_{\Lambda_{2L}}  ~ \Big\vert ~ {\hat\eta_{\Lambda_\ell,i}} \right]\left(-\hat\eta_{\Lambda_\ell,i}\right) + \lambda \E \left[ \frac{1}{\left| \Lambda_{2L} \right|} \sum_{v \in \Lambda_\ell} \left\langle  \sigma_{v,i} \right\rangle^{\tau , h}_{\Lambda_{2L}}  ~ \Big\vert ~ {\hat\eta_{\Lambda_\ell,i}} \right] \left(\hat\eta_{\Lambda_\ell,i}\right).
\end{equation}
The strategy is to apply Proposition~\ref{labelMLEmma} with the map $g: \R \to \R$ defined by the formula
\begin{equation} \label{eq:defogg}
     g(\hat\eta_{\Lambda_\ell,i}) :=\frac{L^d}{2C \ell^{d-2} \left| \Lambda_\ell \right|^{1/ 2}} G'\left(\frac{\hat\eta_{\Lambda_\ell,i}}{\left| \Lambda_\ell \right|^{1/ 2}}\right),
\end{equation}
where $C$ is the constant appearing in the right side of~\eqref{eq:1701}.
Let us first verify that the map $g$ belongs to the set $\mathcal{G}$ introduced in~\eqref{def.setF}. We fix an interval $I = [t_0 , t_1] \subseteq \R$. By the inequality~\eqref{eq:1701}, we have
\begin{equation*}
    \left| \int_I g(t) \, dt \right| = \frac{ L^d}{2C \ell^{d-2}} \left| G\left( \frac{t_1}{\left| \Lambda_\ell  \right|^{1/ 2}} \right) - G\left( \frac{t_0}{\left| \Lambda_\ell \right|^{1/ 2}} \right)  \right| \leq 1.
\end{equation*}
Consequently, the map $g$ belongs to the set $\mathcal{G}$. We can thus apply Proposition~\ref{labelMLEmma} and obtain
\begin{equation} \label{eq:1916}
    \left| \int_\R g(t) e^{-\frac{t^2}{2}} \, dt \right| \leq C,
\end{equation}
for some constant $C > 0$. Using the definition of $g$ stated in~\eqref{eq:defogg} and performing the change of variable $t \to \left| \Lambda_\ell  \right|^{\frac 12} t$, we obtain the inequality
\begin{equation*}
    \left| \int_\R G'(t) e^{-\frac{\left|\Lambda_\ell \right| t^2}{2}} \, dt \right| \leq C \frac{\ell^{d-2}}{L^d}.
\end{equation*}
Using that the random variable ${\hat\eta_{\Lambda_\ell, i}}$ is Gaussian of variance $\left| \Lambda_\ell \right|^{-1}$ and the identity~\eqref{eq:1907}, we obtain the equality
\begin{equation*}
    \sqrt{\frac{ \left| \Lambda_\ell \right|}{2\pi}} \int_\R G'(t) e^{- \frac{\left| \Lambda_\ell \right| t^2}{2}} \, dt = 2 \lambda \E \left[ \frac{1}{\left| \Lambda_{2L} \right|} \sum_{v \in \Lambda_{\ell}} \left\langle  \sigma_{v,i} \right\rangle^{\tau , h}_{\Lambda_{2L}} \right].
\end{equation*}
Combining the two previous displays shows
\begin{equation} \label{eq:16460411}
   \left| \E \left[ \frac{1}{\left| \Lambda_\ell \right|} \sum_{v \in \Lambda_\ell} \left\langle  \sigma_{v,i} \right\rangle^{\tau , h}_{\Lambda_{2L}}  \right] \right| \leq C \ell^{\frac d2 - 2}.
\end{equation}
Since the inequality~\eqref{eq:16460411} holds for any $i \in \{ 1 , \ldots, n \}$, it implies the inequality~\eqref{eq:TV0904}.
\end{proof}

\subsubsection{Algebraic decay of the magnetization uniform over the boundary conditions} \label{subsec4.3.21739}

In this section, we use the results established in Subsection~\ref{subsec4.3.11739} to obtain an algebraic decay for the magnetization which 
holds uniformly over the boundary condition.

\begin{proof}[Proof of Theorem~\ref{prop:subcritical2}: estimate~\eqref{eq:TV090404}]
Fix a side length $L \geq 2$. We consider the system with periodic boundary condition and note that, for any vertex $v \in \Lambda_{L}$ and any $h \in \R^n$,
\begin{equation} \label{eq:18100411}
    \E \left[ \left\langle \sigma_v \right\rangle^{\per , h}_{\Lambda_{L}} \right] = \E \left[ \left\langle \sigma_0 \right\rangle^{\per , h}_{\Lambda_{L}} \right].
\end{equation}
Let us now fix $h \in \R^n$ such that $\left| h \right| \leq 1$ and set $\ell := |h|^{-\frac 12} \wedge \frac{L}{2}$. Applying~\eqref{eq:TV0904} with the boxes $\Lambda_{L}$ and $\Lambda_{\ell}$, and using~\eqref{eq:18100411}, we obtain
\begin{equation*}
    \left| \E \left[ \left\langle \sigma_0 \right\rangle^{\per , h}_{\Lambda_{L}} \right] \right| = \left| \E \left[ \frac{1}{\left| \Lambda_{\ell} \right|}\sum_{v \in \Lambda_{\ell}}\left\langle \sigma_v \right\rangle^{\per , h}_{\Lambda_{L}} \right] \right| \leq C \ell^{\frac d2 - 2} \leq C (|h|\vee L^{-2})^{1 - \frac d4}.
\end{equation*}
Using the identity~\eqref{eq:14210811}, we obtain, for any $h \in \R^n$ such that $|h| \leq 1$,
\begin{align} \label{eq:1352199}
    \left| \E \left[ \FE^{\per , h}_{\Lambda_L}(\eta) - \FE^{\per , 0}_{\Lambda_L}(\eta)\right] \right| &  \leq \sum_{i=1}^n \int_0^1 \left| \E \left[ \frac{1}{\left| \Lambda_L \right|}\sum_{v \in \Lambda_L}\left\langle \sigma_{v,i} \right\rangle^{\per , t h}_{\Lambda_{L}} \right] \right| |h_i| \, dt  \\
    & \leq \sum_{i=1}^n \int_0^1 \left| \E \left[ \left\langle \sigma_{0,i} \right\rangle^{\per , t h}_{\Lambda_{L}} \right] \right| |h_i| \, dt \notag \\
    & \leq  C \left| h \right| \left( \left| h\right| \vee L^{- 2} \right)^{1 - \frac d4} \notag \\ \notag
    & \leq C  \left( \left| h\right| \vee L^{-2} \right)^{2- \frac d4},
\end{align}
where we used $ \left| h\right|\leq  \left| h\right| \vee L^{-2}$ in the last inequality.
Let us then fix an integer $i \in \{ 1 , \ldots, n \}$. For each realization of the random field $\eta$ and each $h \in \R^n$, we let $\tau_i(\eta , h) \in \S^{\partial \Lambda_L}$ be a boundary condition satisfying
\begin{equation} \label{eq:12010511}
    \frac{1}{\left| \Lambda_L \right|}  \sum_{v \in \Lambda_L} \left\langle  \sigma_{v,i} \right\rangle^{\tau_i(\eta , h) , h}_{\Lambda_L}  = \sup_{\tau \in \S^{\partial \Lambda_L}}  \frac{1}{\left| \Lambda_L \right|} \sum_{v \in \Lambda_L} \left\langle  \sigma_{v,i} \right\rangle^{\tau , h}_{\Lambda_L}.
\end{equation}
Note that, by the inequality~\eqref{eq:12030211} of Proposition~\ref{prop.basicpropfreeen} (with $R = 0$), we have, for any $h,h' \in \R^n$,
\begin{equation} \label{eq:1353199}
    \left| \E \left[ \FE^{\tau_i(\eta , h) , h'}_{\Lambda_L}(\eta) \right] - \E \left[ \FE^{\per , h'}_{\Lambda_L}(\eta) \right] \right| \leq \frac{C}{L}.
\end{equation}
A combination of the inequalities~\eqref{eq:1352199} and~\eqref{eq:1353199} yields, for any $h, h' \in \R^n$ satisfying $ |h'| \leq 1$,
\begin{equation} \label{eq:13561999}
    \left| \E \left[ \FE^{\tau_i(\eta , h) , h'}_{\Lambda_L}(\eta) \right] - \E \left[ \FE^{\tau_i(\eta, h) , 0}_{\Lambda_L}(\eta) \right] \right| \leq C \left( \left| h' \right| \vee L^{-2} \right)^{2 - \frac d4} + \frac{C}{L}. \\
\end{equation}
Let us fix $h = (h_1 , \ldots, h_n) \in \R^n$ such that $\left| h \right| \leq 1$, set $\alpha :=  2/(8-d)$ and denote by
\begin{equation*}
\tilde h :=
\left(h_1 , \ldots, h_{i-1}, h_i + \left( |h| \vee L^{-2} \right)^\alpha , h_{i+1} , \ldots, h_n\right).
\end{equation*}
We note that we have $|\tilde h| \leq 2 \left(\left| h \right| \vee L^{-2} \right)^\alpha \leq 2$. We next introduce the function $$G: h' \mapsto -\E \left[ \FE^{\tau_i(\eta , h) ,h'}_{\Lambda_L}(\eta) \right].$$ Observe that the map $G$ is convex and that its derivative with respect to the $i-$th variable satisfies
\begin{equation*}
    \frac{\partial G}{\partial h_i'}(h') = \E \left[ \frac{1}{\left| \Lambda_L \right|} \sum_{v \in \Lambda_L} \left\langle  \sigma_{v,i} \right\rangle^{\tau_i(\eta , h) , h'}_{\Lambda_L} \right],
\end{equation*}
and that, by~\eqref{eq:13561999} and the inequalities $\left| \tilde h \right| \leq 2 \left(\left| h \right| \vee L^{-2} \right)^\alpha$ and $L \geq \left( |h| \vee L^{-2} \right)^{-1/2}$,
\begin{align*}
    \left| G(\tilde h) - G(h) \right| & \leq C  \left( \left| \tilde h\right| \vee L^{-2} \right)^{2 - \frac d4} + C \left( \left| h\right| \vee L^{-2} \right)^{2 - \frac d4} + \frac{C}{L} \\
    & \leq C \left( \left| h\right| \vee L^{-2} \right)^{\alpha \left(2 - \frac d4\right) } +  C \left( \left| h\right| \vee L^{-2} \right)^{\frac 12}.
\end{align*}
Combining the previous observations with~\eqref{eq:12010511} and~\eqref{eq:13561999} and using the value $\alpha = 2/(8 - d)$, we obtain
\begin{align} \label{eq:16180511}
     \E \left[ \sup_{\tau \in \S^{\partial \Lambda_L}}  \frac{1}{\left| \Lambda_L \right|} \sum_{v \in \Lambda_L} \left\langle  \sigma_{v,i} \right\rangle^{\tau , h}_{\Lambda_L} \right]  = \E \left[ \frac{1}{\left| \Lambda_L \right|} \sum_{v \in \Lambda_L} \left\langle  \sigma_{v,i} \right\rangle^{\tau_i(\eta , h) , h}_{\Lambda_L} \right]  & = \frac{\partial G}{\partial h_i'}(h)  \\
     & \leq  \frac{G(\tilde h) - G(h)}{\left(|h|\vee L^{-2} \right)^\alpha} \notag \\
    & \leq C \left( \left| h\right| \vee L^{- 2} \right)^{\frac{4 - d}{2(8-d)}}. \notag
\end{align}
We next upgrade~\eqref{eq:16180511} by obtaining stronger concentration properties. To this end, we implement an argument similar to the one presented in Section~\ref{sectionconcentrationatg} and partition the box $\Lambda_L$ into $\tilde \Lambda_1, \ldots, \tilde \Lambda_N$ boxes of sidelength $\sqrt{L}$, with $N \simeq L^{d/2}$. Using the Hoeffding concentration inequality and the inequality~\eqref{eq:16180511} for boxes of side length $\sqrt{L}$, we have the upper bound
\begin{equation*}
    \mathbb{P} \left[ \frac{1}{N} \sum_{k = 1}^N \sup_{\tau \in \S^{\partial \tilde \Lambda_k}}  \frac{1}{\left| \tilde \Lambda_k \right|} \sum_{v \in \tilde  \Lambda_k} \left\langle  \sigma_{v,i} \right\rangle^{\tau , h}_{\tilde  \Lambda_k} \geq  C \left( \left| h\right| \vee L^{- 1} \right)^{\frac{4 - d}{2(8-d)}} + L^{-\frac{4 - d}{2(8-d)}} \right] \leq \exp \left( - c L^{\frac{d}{2} - \frac{4 - d}{(8-d)}} \right).
\end{equation*}
Using that the exponent $d/2 - (4 - d)/(8-d)$ is always strictly positive and the inequality
\begin{equation*}
    \sup_{\tau \in \S^{\partial \Lambda_L}}  \frac{1}{\left| \Lambda_L \right|} \sum_{v \in \Lambda_L} \left\langle  \sigma_{v,i} \right\rangle^{\tau , h}_{\Lambda_L} \leq \frac{1}{N} \sum_{k = 1}^N \sup_{\tau \in \S^{\partial \tilde \Lambda_k}}  \frac{1}{\left| \tilde \Lambda_k \right|} \sum_{v \in \tilde  \Lambda_k} \left\langle  \sigma_{v,i} \right\rangle^{\tau , h}_{\tilde  \Lambda_k},
\end{equation*}
we obtain the (weaker) bound
\begin{equation} \label{eq:1618051111}
    \mathbb{P} \left[  \sup_{\tau \in \S^{\partial \Lambda_L}}  \frac{1}{\left| \Lambda_L \right|} \sum_{v \in \Lambda_L} \left\langle  \sigma_{v,i} \right\rangle^{\tau , h}_{\Lambda_L}  \geq  C \left( \left| h\right| \vee L^{- 1} \right)^{\frac{4 - d}{2(8-d)}} \right] \leq C \left( \left| h\right| \vee L^{- 1} \right)^{\frac{4 - d}{2(8-d)}}.
\end{equation}
To complete the argument, let us consider the random boundary condition $\tau_{i,-}(\eta , h)$ defined so as to satisfy
\begin{equation} \label{def.tautildeii}
    \frac{1}{\left| \Lambda_L \right|}  \sum_{v \in \Lambda_L} \left\langle  \sigma_{v,i} \right\rangle^{\tau_{i,-}(\eta , h) , h}_{\Lambda_L}  = \inf_{\tau \in \S^{\partial \Lambda_L}}  \frac{1}{\left| \Lambda_L \right|} \sum_{v \in \Lambda_L} \left\langle  \sigma_{v,i} \right\rangle^{\tau , h}_{\Lambda_L},
\end{equation}
and define $\tilde \tau_{i,-}(\eta , h) := -  \tau_{i,-}(-\eta , h) $.
Using a similar computation as the one performed in~\eqref{eq:16180511}, we obtain
\begin{equation} \label{eq:16060511}
    \E \left[ \frac{1}{\left| \Lambda_L \right|} \sum_{v \in \Lambda_L} \left\langle  \sigma_{v,i} \right\rangle^{ \tilde \tau_{i,-}(\eta , h) , - h}_{\Lambda_L} \right]  \leq C \left( \left| h\right| \vee L^{- 2} \right)^{\frac{4- d}{2(8-d)}}.
\end{equation}
Combining~\eqref{def.tautildeii},~\eqref{eq:16060511} with the identity~\eqref{eq:16560511} yields
\begin{equation} \label{eq:16120511}
    \E \left[ \inf_{\tau \in \S^{\partial \Lambda_L}}  \frac{1}{\left| \Lambda_L \right|} \sum_{v \in \Lambda_L} \left\langle  \sigma_{v,i} \right\rangle^{\tau , h}_{\Lambda_L} \right] = - \E \left[ \frac{1}{\left| \Lambda_L \right|} \sum_{v \in \Lambda_L} \left\langle  \sigma_{v,i} \right\rangle^{ \tilde \tau_i(\eta , h) , - h}_{\Lambda_L} \right]  \geq -  C \left( \left| h\right| \vee L^{- 2} \right)^{\frac{4 - d}{2(8-d)}}.
\end{equation}
Using the same partition of the box $\Lambda_L$ as for the supremum above, we obtain the concentration inequality
\begin{equation} \label{eq:1612051111}
    \mathbb{P} \left[  \inf_{\tau \in \S^{\partial \Lambda_L}}  \frac{1}{\left| \Lambda_L \right|} \sum_{v \in \Lambda_L} \left\langle  \sigma_{v,i} \right\rangle^{\tau , h}_{\Lambda_L}  \leq  - C \left( \left| h\right| \vee L^{- 1} \right)^{\frac{4 - d}{2(8-d)}} \right] \leq C \left( \left| h\right| \vee L^{- 1} \right)^{\frac{4 - d}{2(8-d)}}.
\end{equation}
Combining~\eqref{eq:1618051111} and~\eqref{eq:1612051111} implies
\begin{equation} \label{eq:16200511}
    \E \left[ \sup_{\tau \in \S^{\partial \Lambda_L}}  \frac{1}{\left| \Lambda_L \right|} \left| \sum_{v \in \Lambda_L} \left\langle  \sigma_{v,i} \right\rangle^{\tau , h}_{\Lambda_L} \right| \right] \leq C \left( \left| h\right| \vee L^{- 1} \right)^{\frac{4 - d}{2(8-d)}}.
\end{equation}
Using that the inequality~\eqref{eq:16200511} holds for any integer $i \in \{ 1 , \ldots , n\}$ completes the proof of the estimate~\eqref{eq:TV090404}.
\end{proof}

\subsection{Proof of Theorem~\ref{thm.thm2}} \label{criticdimesninosCSS}

The objective of this section is to prove Theorem~\ref{thm.thm2} following the outline presented at the beginning of Section~\ref{SectionCSS}.

\subsubsection{A lower bound on the conditional expectation of the spatially-averaged magnetization}

In the first step of the proof, we show that, if the averaged value of the field $\eta$ in a box~$\Lambda$ is negative enough, then the thermally and spatially averaged magnetization of the continuous spin system with periodic boundary conditions in the box $\Lambda$ must be small. The argument relies on a combination of the variational lemma stated in Proposition~\ref{LemmavarprinCSS2} and of the Mermin-Wagner upper bound for the free energy (Proposition~\ref{propMermW}). Before stating the result, we introduce a notation for the quantile of the normal distribution which will be used in the statement and proof of Lemma~\ref{lemma.08538}.

\begin{definition}[Quantile of the normal distribution]
    For each $\delta > 0$, we define the $\exp \left( - \frac{1}{\delta^2} \right)$-quantile of the normal distribution by the formula
\begin{equation} \label{def.tdelta}
    t_\delta := \min \left\{ t \in \R \, : \, \frac{1}{\sqrt{2 \pi}}\int_t^\infty e^{-\frac{s^2}{2}} \, ds  \leq 1- \exp \left( - \frac{1}{\delta^2} \right) \right\},
\end{equation}
Let us note that there exist two constants $c , C \in (0, \infty)$ such that for any $\delta \in (0, 1/2]$,
$$-C \delta^{-1} \leq t_\delta \leq -c \delta^{-1}.$$
\end{definition}

\begin{lemma} \label{lemma.08538}
Let $d = 4$. Fix $\beta >0$, $\lambda >0$ and $i \in \{ 1 , \ldots , n\}$ and a box $\Lambda_0 \subseteq \Zd$. Let $\tau \in \S^{\partial \Lambda_0}$ be a boundary condition (which may be the free or periodic boundary conditions). For any box $\Lambda$ of side length $\ell$ such that $2 \Lambda \subseteq \Lambda_0$, any $\delta > 0$, $h \in \R^n$ satisfying $\left| h \right| \leq 1/\ell^2$, we have the estimate
\begin{equation*}
        \E \left[ \frac{1}{\left| \Lambda \right|} \sum_{v \in \Lambda}  \left\langle  \sigma_{v,i} \right\rangle^{\tau , h}_{\Lambda_0} ~ \Big\vert ~ {\hat\eta_{\Lambda,i}} ,  \eta_{\left( \Lambda_0 \setminus 2\Lambda\right),i} \right] \leq C \delta \hspace{3mm} \mbox{on the event}~\left\{ \hat \eta_{\Lambda , i} \leq \ell^{-2} t_\delta \right\}.
\end{equation*}
\end{lemma}

\begin{proof}
We denote by $L$ the side length of the box $\Lambda_0$. By Proposition~\ref{propMermW} and using the assumption $\left| h\right| \leq 1/\ell^2$, we have the inequality
\begin{equation} \label{eq:1316}
    \E \left[ \tilde \FE_{\Lambda_0 ,\Lambda}^{\tau,h} - \FE_{\Lambda_0}^{\tau,h} ~ \Big\vert ~ {\hat\eta_{\Lambda,i}} ,  \eta_{\left( \Lambda_0 \setminus 2\Lambda\right),i}  \right] \leq  \frac{C \ell^2}{L^4} \hspace{5mm} \mathbb{P}-\mbox{almost surely}.
\end{equation}
As in~\eqref{eq:0851}, we have the identity
\begin{equation*}
    \E \left[ \tilde \FE_{\Lambda_0 ,\Lambda}^{\tau,h} ~ \Big\vert ~ {\hat\eta_{\Lambda,i}} ,  \eta_{\left( \Lambda_0 \setminus 2\Lambda\right),i}  \right] ({\hat\eta_{\Lambda,i}} ,  \eta_{\left( \Lambda_0 \setminus 2\Lambda\right),i} )  = \E \left[  \FE_{\Lambda_0}^{\tau,h} ~ \Big\vert ~ {\hat\eta_{\Lambda,i}} ,  \eta_{\left( \Lambda_0 \setminus 2\Lambda\right),i}  \right] (- {\hat\eta_{\Lambda,i}} ,  \eta_{\left( \Lambda_0 \setminus 2\Lambda\right),i} ) .
\end{equation*}
We first claim that there exists a constant $C > 0$ such that, for every $\delta >0$,
\begin{equation} \label{eq:10471}
\P \left[ \E \left[ \frac{1}{\left| \Lambda \right|} \sum_{v \in \Lambda} \left\langle  \sigma_{v,i} \right\rangle^{\tau , h}_{\Lambda_0}  ~ \Big\vert ~ {\hat\eta_{\Lambda,i}} ,  \eta_{\left( \Lambda_0 \setminus 2\Lambda\right),i}  \right]  \leq C \delta  ~ \Big\vert ~  \eta_{\left( \Lambda_0 \setminus 2\Lambda\right),i}\right]
\geq \exp \left( - \frac{1}{\delta^2} \right) \hspace{5mm} \mathbb{P}-\mbox{almost surely}.
\end{equation}
To prove the inequality~\eqref{eq:10471}, we introduce the following map
\begin{multline*}
G_{\Lambda} : (\hat\eta_{ \Lambda,i} , \eta_{\left( \Lambda_0 \setminus 2\Lambda\right),i}) \mapsto \E \left[  \FE_{\Lambda_0}^{\tau,h} ~ \Big\vert ~ {\hat\eta_{\Lambda,i}} ,  \eta_{\left( \Lambda_0 \setminus 2\Lambda\right),i} \right] (-\hat\eta_{ \Lambda,i} ,\eta_{\left( \Lambda_0 \setminus 2\Lambda\right),i}) \\ -   \E \left[  \FE_{\Lambda_0}^{\tau,h} ~ \Big\vert ~ {\hat\eta_{\Lambda,i}} ,  \eta_{\left( \Lambda_0 \setminus 2\Lambda\right),i} \right](\hat\eta_{ \Lambda,i} , \eta_{\left( \Lambda_0 \setminus 2\Lambda\right),i}).
\end{multline*}
Let us note that the map $G_{\Lambda}$ satisfies the identity
\begin{align} \label{eq:1353}
    \frac{\partial G_{\Lambda}}{\partial \hat \eta_{\Lambda , i}}(\hat \eta_{\Lambda , i} , \eta_{\left( \Lambda_0 \setminus 2\Lambda\right),i}) & = \lambda \E \left[ \frac{1}{\left| \Lambda_0\right|} \sum_{v \in \Lambda} \left\langle  \sigma_{v,i} \right\rangle^{\tau , h}_{\Lambda_0}  ~ \Big\vert ~ {\hat\eta_{\Lambda,i}} ,  \eta_{\left( \Lambda_0 \setminus 2\Lambda\right),i} \right] (\hat \eta_{\Lambda, i} , \eta_{\left( \Lambda_0 \setminus 2\Lambda\right),i}) \\ & \quad + \lambda \E \left[ \frac{1}{\left| \Lambda_0\right|} \sum_{v \in \Lambda} \left\langle  \sigma_{v,i} \right\rangle^{\tau , h}_{\Lambda_0}  ~ \Big\vert ~ {\hat\eta_{\Lambda,i}} ,  \eta_{\left( \Lambda_0 \setminus 2\Lambda\right),i} \right] (-\hat \eta_{\Lambda , i} , \eta_{\left( \Lambda_0 \setminus 2\Lambda\right),i}). \notag
\end{align}
We next show the following inequality: for every $\delta >0$,
\begin{equation} \label{eq:10472}
\P\left( \frac{\partial G_{\Lambda}}{\partial \hat \eta_{\Lambda , i}} \leq \frac{C\ell^4}{L^4} \delta  ~ \Big\vert ~  \eta_{\left(\Lambda_0 \setminus 2\Lambda\right) , i} \right) \geq  \exp \left( - \frac{1}{\delta^2} \right) \hspace{3mm} \mathbb{P}-\mbox{almost-surely}.
\end{equation}
To prove the estimate~\eqref{eq:10472}, we fix a realization of the field $ \eta_{\left(\Lambda_0 \setminus 2\Lambda\right) , i}$, and apply Proposition~\ref{LemmavarprinCSS2} with the choice of function
\begin{equation*}
     g : \hat \eta_{\Lambda,i} \mapsto \frac{L^4}{2C(1 \vee \lambda)\ell^4} \frac{ \partial G_{\Lambda}}{\partial \hat \eta_{\Lambda,i}} \left( \frac{\hat \eta_{\Lambda,i}}{\ell^2}, \eta_{\left(\Lambda_0 \setminus 2\Lambda \right), i }\right),
\end{equation*}
where the constant $C$ is the one which appears in the right side of~\eqref{eq:1316}. We first verify that the map $g$ belongs to the set $\mathcal{G}$ (defined in~\eqref{def.setF}). The result is a consequence of the following computation: by~\eqref{eq:1316}, we have, for any interval $I = [t_0 , t_1] \subseteq \R$,
\begin{equation*}
    \left| \int_{I} g(t) \, dt  \right|= \frac{L^4}{2C(1 \vee \lambda)\ell^2} \left| G_\Lambda \left( \frac{t_1}{\ell^2} , \eta_{\left(\Lambda_0 \setminus 2\Lambda\right) , i} \right) - G_\Lambda \left( \frac{t_0}{\ell^2}, \eta_{\left(\Lambda_0 \setminus 2\Lambda\right) , i} \right) \right| \leq 1.
\end{equation*}
The fact that the map $g$ is larger than $-1$ is obtained as consequence of the assumption that the spin space is the sphere $\mathbb{S}^{n-1}$ (and thus the norm of a spin is always equal to $1$).

Applying Proposition~\ref{LemmavarprinCSS2} yields, for any $\delta \in (0 , 1]$,
\begin{equation*}
    \int_\R \indc_{\left\{ g(t) \leq \delta\right\}} e^{-t^2/2} \, dt \geq  e^{-\frac{C}{\delta^2}}.
\end{equation*}
Rescaling the previous inequality, using that the averaged field $\hat\eta_{\Lambda,i}$ is Gaussian of variance~$\ell^{-4}$, and that it is independent of the field $\eta_{\left( \Lambda_0 \setminus 2 \Lambda \right) , i}$ gives the estimate~\eqref{eq:10472}. We then reformulate the inequality~\eqref{eq:10472}: using the formula~\eqref{eq:1353} and a union bound, we obtain that there exists a constant $C > 0$ such that, $\P-$almost surely,
\begin{multline} \label{eq:1359}
    \P \left[ \E \left[ \frac{1}{\ell^4} \sum_{v \in \Lambda}  \left\langle  \sigma_{v,i} \right\rangle^{\tau , h}_{\Lambda_0} ~ \Big\vert ~ {\hat\eta_{\Lambda,i}} ,  \eta_{\left(\Lambda_0 \setminus 2\Lambda\right) , i}   \right]  \leq C\delta  ~ \Big\vert ~  \eta_{\left(\Lambda_0 \setminus 2\Lambda\right) , i}   \right] \\ + \P \left[ \E \left[ \frac{1}{\ell^4} \sum_{v \in \Lambda}  \left\langle  \sigma_{v,i} \right\rangle^{\tau , h}_{\Lambda_0} ~ \Big\vert ~ {\hat\eta_{\Lambda,i}} ,  \eta_{\left(\Lambda_0 \setminus 2\Lambda\right) , i}   \right](-\hat \eta_{\Lambda,i}, \eta_{\left(\Lambda_0 \setminus 2\Lambda\right) , i}) \leq C\delta  ~ \Big\vert ~ \eta_{\left(\Lambda_0 \setminus 2\Lambda\right) , i}  \right] \geq  \exp \left( - \frac{1}{\delta^2} \right).
\end{multline}
Since the law of the random variable $\hat \eta_{\Lambda,i}$ is invariant under the involution $\hat \eta_{\Lambda,i} \to -\hat \eta_{\Lambda,i}$, the two terms in the left side of~\eqref{eq:1359} are equal. We thus obtain, for any $\delta > 0$,
\begin{equation} \label{eq:15408}
    \P \left[ \E \left[ \frac{1}{\left| \Lambda \right|} \sum_{v \in \Lambda}  \left\langle  \sigma_{v,i} \right\rangle^{\tau , h}_{\Lambda_0} ~ \Big\vert ~ \hat\eta_{\Lambda,i} ,  \eta_{\left(\Lambda_0 \setminus 2\Lambda\right) , i}   \right]  \leq C\delta ~ \Big\vert ~  \eta_{\left(\Lambda_0 \setminus 2\Lambda\right) ,i}   \right] \geq \frac 12 \exp \left( -\frac{1}{\delta^2} \right) ~\P-\mbox{a.s.}
\end{equation}
Since the estimate~\eqref{eq:15408} is valid for any $\delta > 0$, it implies the inequality~\eqref{eq:10471} by increasing the value of the constant $C$ if necessary. We then observe that, for each fixed realization of the field $\eta_{\left(\Lambda_0 \setminus 2\Lambda\right),i}$, the map
$$\hat \eta_{\Lambda,i} \mapsto -\E \left[  \FE_{\Lambda_0}^{\tau,h} ~ \Big\vert ~ {\hat\eta_{\Lambda,i}} ,  \eta_{\left(\Lambda_0 \setminus 2\Lambda\right),i}  \right](\hat \eta_{\Lambda ,i}, \eta_{\left(\Lambda_0 \setminus 2\Lambda\right) , i})$$
is convex and that its derivative is the function $$\hat \eta_{\Lambda , i} \mapsto \E \left[ \frac{1}{\left| \Lambda_0\right|} \sum_{v \in \Lambda}  \left\langle  \sigma_{v,i} \right\rangle^{\tau , h}_{\Lambda_0} ~ \Big\vert ~ \hat\eta_{\Lambda,i} ,  \eta_{\left(\Lambda_0 \setminus 2\Lambda\right) , i}   \right] (\hat\eta_{\Lambda,i} , \eta_{\left(\Lambda_0 \setminus 2\Lambda\right),i}).$$ Since the derivative of a convex function is increasing, we obtain that the map
$$ \hat \eta_{\Lambda , i}\mapsto \E \left[ \frac{1}{\left| \Lambda  \right|} \sum_{v \in \Lambda}  \left\langle  \sigma_{v,i} \right\rangle^{\tau , h}_{\Lambda_0} ~ \Big\vert ~ {\hat\eta_{\Lambda,i}} ,  \eta_{\left(\Lambda_0 \setminus 2\Lambda \right), i}   \right](\hat\eta_{\Lambda,i} , \eta_{\left(\Lambda_0 \setminus 2\Lambda\right),i})$$ is increasing. Combining this observation with the inequality~\eqref{eq:10471}, the definition of the quantile $t_\delta$ stated in~\eqref{def.tdelta}, and the fact that the random variable $\hat \eta_{\Lambda,i}$ is Gaussian of variance $ \ell^{-4}$, we obtain, for any $\delta > 0$,
\begin{equation} \label{eq:1720}
    \E \left[ \frac{1}{\left| \Lambda \right|} \sum_{v \in \Lambda}  \left\langle  \sigma_{v,i} \right\rangle^{\tau , h}_{\Lambda_0} ~ \Big\vert ~ \hat\eta_{\Lambda,i} ,  \eta_{\left(\Lambda_0 \setminus 2\Lambda\right), i} \right]\left(\hat\eta_{\Lambda,i} , \eta_{\left(\Lambda_0 \setminus 2\Lambda\right),i}\right)  \leq C \delta \hspace{3mm} \mbox{on the event} ~\left\{ \hat\eta_{\Lambda , i} \leq \ell^{-2} t_{\delta} \right\}.
\end{equation}
The proof of Lemma~\ref{lemma.08538} is complete.
\end{proof}

\subsubsection{Mandelbrot percolation argument}

In this section, we combine the result of Lemma~\ref{lemma.08538} with a Mandelbrot percolation argument to obtain a quantitative rate of convergence on the expected value (with respect to the random field) of the spatially and thermally averaged magnetization with a fixed boundary condition.

\begin{lemma} \label{lemmasecondmandelbrot}
Let $d = 4$. Fix $\beta >0$, $\lambda >0$, a box $\Lambda_0 \subseteq \Zd$, an integer $L\geq 3$ such that $\Lambda_{2L} \subseteq \Lambda_0$, and an external magnetic field $h \in \R^n$ satisfying $|h| \leq L^{-2}$. Let $\tau \in \S^{\partial \Lambda_{0}}$ be a boundary condition (which may be the free and periodic boundary conditions). There exists a constant $C > 0$ depending only on $\lambda$, $n$ and~$\Psi$ such that
\begin{equation} \label{eq:17370611}
    \left| \E \left[ \frac{1}{\left|\Lambda_{L} \right|} \sum_{v \in \Lambda_{L}}\left\langle  \sigma_v \right\rangle^{\tau , h}_{\Lambda_0} \right] \right| \leq \frac{C}{\sqrt{\ln \ln L}}.
\end{equation}
\end{lemma}

\begin{proof}[Proof of Lemma~\ref{lemmasecondmandelbrot}]
First, let us note that, by the identity~\eqref{eq:16560511}, it is sufficient, in order to prove~\eqref{eq:17370611}, to prove, for any integer $i \in \{ 1 , \ldots , n \}$,
\begin{equation} \label{eq:17380611}
    \E \left[ \frac{1}{\left|\Lambda_{L} \right|} \sum_{v \in \Lambda_{L}}\left\langle  \sigma_{v,i} \right\rangle^{\tau , h}_{\Lambda_0} \right] \leq \frac{C}{\sqrt{\ln \ln L}}.
\end{equation}
Additionally, it is sufficient to prove the inequality~\eqref{eq:17380611} when $L$ is large enough.

We now fix an integer $i \in \{ 1 , \ldots , n \}$ and prove the inequality~\eqref{eq:17380611}. To this end, we set $\delta = (C_0 /(\ln \ln L)^\frac 12) \wedge (1/2)$ for some large constant $C_0$ whose value is decided at the end of the proof. The strategy is to implement a Mandelbrot percolation argument in the box $\Lambda_{L}$ with the following definition of good box:
\begin{center}
    a box $\Lambda \subseteq \Lambda_{L}$ of side length $\ell$ is good if $\hat \eta_{\Lambda,i} \leq \ell^{-2} t_{\delta}$.
\end{center}
We let $k := 2 \lfloor \exp \left( \sqrt[4]{\ln L} \right) \rfloor + 1$, assume that $L$ is large enough so that $k \geq 5$, and denote by $l_{\max}$ the largest integer which satisfies $k^{l_{\max}} \leq \sqrt{L}$ (and select $L$ sufficiently large so that $l_{\max} \geq 1$). We introduce the set of boxes
\begin{equation*}
   \mathcal{T}_l := \left\{ \left( z  + \left[ - \frac{L}{k^{l}}  ,  \frac{L}{k^{l}} \right)^4 \right) \cap \Lambda_L \, : \, z \in \frac{2 L}{k^{l}} \Z^4 \cap [-L , L]^4 \right\}.
\end{equation*}
and implement the Mandelbrot percolation argument developed in the second step of the proof of Lemma~\ref{prop2.31708}. We obtain a collection $\Q \subseteq \cup_{l=0}^{l_{\max}} \mathcal{T}_l$ of good boxes. We need to prove the following two properties pertaining to the collection $\mathcal{Q}$. First, we show that the set of uncovered points is typically small: we prove the inequality, for any vertex $v \in \Lambda_{L}$,
    \begin{equation} \label{eq:14522}
        \P \left( v ~\mbox{is not covered} \right) \leq  \exp \left( -c \sqrt{\ln L} \right).
    \end{equation}
Second, we prove that the expected value of the spatially and thermally averaged magnetization on a box of the collection $\Q$ is small: we show the estimate, for each box $\Lambda \in \cup_{l=0}^{l_{\max}} \mathcal{T}_l$,
    \begin{equation} \label{eq:1457}
        \E \left[  \indc_{\{ \Lambda \in \Q \}}   \frac{1}{\left| \Lambda \right|}\sum_{v \in \Lambda} \left\langle  \sigma_{v,i} \right\rangle^{\tau , h}_{\Lambda_0} \right] \leq C \delta  \E \left[  \indc_{\{ \Lambda \in \Q \}} \right] +  \frac{C}{(\ln L)^{2}}.
\end{equation}
We first focus on the proof of the inequality~\eqref{eq:14522}. To this end, we fix a vertex $v \in \Lambda_{L}$, let $\Lambda_0(v),  \ldots, \Lambda_{l_{\max}}(v)$ be the boxes of the collections $\mathcal{T}_0 , \ldots , \mathcal{T}_{l_{\max}}$ containing the vertex $v$, and denote their side length by $\ell_0 , \ldots , \ell_{\max}$ respectively. For any $l \in \{ 0 , \ldots, l_{\max} -1\}$, we denote by $k_l := \ell_l /\ell_{l+1}$ the ratio between the two side length $ \ell_l$ and $\ell_{l+1}$ and note that there exist constants $c , C$ such that $c k \leq k_l \leq C k$ as soon as $L$ is large enough. We denote by
\begin{equation*}
    \hat \upeta_{l}  :=  \frac{1}{\left|\Lambda_l(v)\right|} \sum_{u \in \Lambda_l(v) \setminus \Lambda_{l+1}(v)} \eta_{u,i}.
\end{equation*}
Note that the random variables $\hat\eta_{\Lambda_l(v),i}$ and $\hat \upeta_{l}$ are typically close to each other: the law of the random variable $\hat \upeta_{l} - \hat\eta_{\Lambda_l(v),i}$ is Gaussian and its variance is equal to $1/(k_l^4 \ell_{l}^4).$
We also note that the random variable $\hat \upeta_{l}$ is independent of the restriction field $\eta$ to the box $\Lambda_{l+1}(v)$.
\smallskip

We have the identity of events
\begin{equation} \label{eq:1606}
    \left\{ v~ \mbox{is not covered} \right\} = \bigcap_{l = 0}^{l_{\max}}   \left\{ \hat\eta_{\Lambda_l(v),i} > \ell_l^{-2} t_{\delta} \right\}.
\end{equation}
We then show that the $(l_{\max}+1)$ events in the right side of~\eqref{eq:1606} are well-approximated by independent events, and use the independence to estimate the probability of their intersection. To this end, we use the identity $\hat \upeta_{l} +  \frac{\hat\eta_{\Lambda_{l+1}(v),i}}{k_l^4} = {\hat\eta_{\Lambda_l(v),i}}$, and note that the following inclusion holds
\begin{multline} \label{eq:1813}
    \bigcap_{l = 0}^{l_{\max}}   \left\{ \hat\eta_{\Lambda_l(v),i} \geq \ell_l^{-2} t_{\delta} \right\} \\ \subseteq \left( \bigcap_{l = 0}^{l_{\max} -1}   \left\{ \hat \upeta_{l} \geq \ell_l^{-2} \left( t_{\delta} - \delta \right) \right\} \bigcap \left\{ \hat \eta_{\Lambda_{l_{\max}}(v),i} \geq \ell_{l_{\max}}^{-2} t_{\delta} \right\} \right) \bigcup \left( \bigcup_{l=0}^{l_{\max-1}} \left\{ \hat\eta_{\Lambda_{l+1}(v),i} \geq k_l^4 \delta \ell_l^{-2} \right\} \right).
\end{multline}
Using that the random variables $\left( \hat \upeta_{l} \right)_{1 \leq l \leq l_{\max}-1}$ are independent and a union bound, we obtain
\begin{equation} \label{eq:1822}
    \P \left( \bigcap_{l = 0}^{l_{\max}}   \left\{ {\hat\eta_{\Lambda_l(v),i}} \geq \ell_j^{-2} t_{\delta} \right\} \right) \leq \underbrace{\prod_{l = 0}^{l_{\max}-1}  \P \left(  \hat \upeta_{l} \geq \ell_l^{-2} \left(t_{\delta } - \delta \right) \right)}_{\eqref{eq:1822}-(i)} +  \underbrace{\sum_{l=0}^{l_{\max}-1} \P \left( {\hat\eta_{\Lambda_{l+1}(v),i}} \geq k_l^4 \delta \ell_l^{-2}\right)}_{\eqref{eq:1822}-(ii)}.
\end{equation}
We estimate the terms~\eqref{eq:1822}-(i) and ~\eqref{eq:1822}-(ii) separately. For the term~\eqref{eq:1822}-(i), we note that the quantile $t_\delta$ satisfies the inequality $-c/\delta \geq t_\delta \geq -C/\delta$. An explicit computation shows that there exist two constants $c , C \in (0 , \infty)$ such that
\begin{equation*}
\P \left(  \hat \upeta_{l} \geq \ell_l^{-2} \left(t_{\delta} - \delta \right) \right) \leq 1 - c \exp \left( - \frac{C}{\delta^2} \left( 1 + \delta^2 \right)^2 \right) \leq 1 - c \exp \left( - \frac{C}{\delta^2} \right) .
\end{equation*}
We recall that we have set $k = 2 \lfloor \exp \left( \sqrt[4]{\ln L} \right) \rfloor +1$, $\delta = (C_0 /(\ln \ln L)^\frac 12) \wedge (1/2)$, and $l_{\max} := \lfloor \ln L / (2\ln k)\rfloor \simeq \left(\ln L\right)^{3/4}$. Consequently, if the constant $C_0$ is chosen large enough,
\begin{equation} \label{eq:095324}
     \prod_{l = 0}^{l_{\max}-1}  \P \left(  \hat \upeta_{l} \geq \ell_l^{-2} \left( t_{\delta} - \delta\right) \right) \leq \left( 1 - c \exp \left( - \frac{C}{\delta^2} \right) \right)^{l_{\max}} \leq \exp \left( -c  \sqrt{\ln L} \right).
\end{equation}
We now estimate the term~\eqref{eq:1822}-(ii). The lower bound $k_l \geq c k$ and an explicit computation give, for each $l \in \{ 0 , \ldots , l_{\max} \}$,
    \begin{equation*}
    \P \left( {\hat\eta_{\Lambda_{l+1}(v),i}} \geq k_l^4 \ell_l^{-2} \delta \right) \leq \P \left( {\hat\eta_{\Lambda_{l+1}(v),i}} \geq k_l^2 \ell_{l+1}^{-2} \delta \right) \leq  \exp \left( - c k^4 \delta^2 \right),
\end{equation*}
and thus
\begin{equation} \label{eq:095424}
        \sum_{l=0}^{l_{\max}-1} \P \left( {\hat\eta_{\Lambda_{l+1}(v),i}} \geq k_l^4 \ell_l^{-2} \delta \right) \leq l_{\max} \exp \left( - c k_l^4 \delta^2 \right) \leq  \exp \left( - c \sqrt{\ln L} \right).
\end{equation}
A combination of~\eqref{eq:1822},~\eqref{eq:095324} and~\eqref{eq:095424} implies~\eqref{eq:14522}.

\medskip

We now focus on the proof of the inequality~\eqref{eq:1457}. We fix an integer $l \in \{ 1 , \ldots, l_{\max} \}$, consider a box $\Lambda' \in \mathcal{T}_l$ and denote its side length by $\ell_l$. We denote by $\Lambda_0' , \ldots, \Lambda_{l-1}'$ the family of boxes which contain the box $\Lambda'$ and belong to the sets $\mathcal{T}_0 , \ldots, \mathcal{T}_{l-1}$ respectively. We denote the side length of these boxes by $\ell_0 , \ldots, \ell_{l-1}$. By construction of the collection $\Q$, we have the identity
\begin{equation*}
    \left\{ \Lambda' \in \Q \right\}  = \left\{ \hat \eta_{\Lambda',i} \leq \ell_l^{-2} t_{\delta} \right\} \bigcap \bigcap_{j=0}^{l-1} \left\{ \hat \eta_{\Lambda_j',i} > \ell_j^{-2} t_{\delta} \right\}.
\end{equation*}
Our objective is to prove that this event is well-approximated by an event which belongs to the $\sigma$-algebra generated by the random variables $\hat \eta_{\Lambda',i}$ and $\eta_{\Lambda_0 \setminus 2\Lambda',i}$. To this end, let us define, for any integer $j \in \{ 0 , \ldots , l -1\}$,
\begin{equation*}
    \hat \upeta_{j}'  :=  \frac{1}{\left|\Lambda_j'\right|} \sum_{u \in \Lambda_j' \setminus 2\Lambda'} \eta_{u,i}.
\end{equation*}
Let us note that the random variable $\hat \upeta_{j}'$ depends only on the realization of the field outside the box $2 \Lambda'$, and that it satisfies the identity
\begin{equation*}
    \hat \eta_{\Lambda_j',i} = \hat \upeta_{j}' + \frac{|2\Lambda'|}{|\Lambda_j'|} \hat \eta_{2\Lambda',i}.
\end{equation*}
We additionally note that, by the definitions of the cube $\Lambda_j'$, the ratio between the side length of the boxes $\Lambda_j'$ and $2 \Lambda' $ is at least of order $k$: there exists a constant $c$ such that $\ell_j \geq c k \ell_l.$

As a consequence of the previous definitions and observations, we have the inclusion
\begin{equation} \label{eq:2051}
\left\{ \hat \eta_{\Lambda_j',i} >  \ell_j^{-2} t_{\delta} \right\} \Delta \left\{ \hat \upeta_{j}' > \ell_j^{-2} t_{\delta} \right\} \subseteq \underbrace{\left\{ \left| \hat \upeta_{j}' - \frac{t_{\delta}}{\ell_j^{2}} \right| \leq \frac{1}{\ell_j^{2}\left( \ln L \right)^3} \right\}}_{\eqref{eq:2051}-(i)}\cup  \underbrace{\left\{ \left| \hat \eta_{2\Lambda',i}\right| \geq \frac{c k^2}{\ell_l^{2}\left( \ln L \right)^3} \right\}}_{\eqref{eq:2051}-(ii)},
\end{equation}
where the symbol $\Delta$ denotes the symmetric difference between the events $\left\{ \hat \eta_{\Lambda_j',i} >  \ell_j^{-2} t_{\delta} \right\}$ and $\left\{ \hat \upeta_{j}' > \ell_j^{-2} t_{\delta} \right\}$. We estimate the probabilities of the two events in the right side of~\eqref{eq:2051}. For the event~\eqref{eq:2051}-(i), we note that the random variable $\hat \upeta_{j}'$ is Gaussian and that its variance satisfies
\begin{equation*}
    \var \hat \upeta_{j}' = \frac{\left| \Lambda_j' \setminus 2\Lambda' \right|}{\left| \Lambda_j' \right|^2} \geq \frac 1{2 \left| \Lambda_j' \right|} =  \frac{1}{2 \ell_j^4},
\end{equation*}
where we have used the inequality $\left| \Lambda_j' \setminus 2\Lambda' \right| \geq \left| \Lambda_j' \right|/2$, which is a consequence of the definition of the box $\Lambda_j'$ together with the assumption $k \geq 5$. We obtain
\begin{equation} \label{eq:1508}
    \P \left(  \left| \hat \upeta_{j}' -  \frac{t_{\delta}}{\ell_j^{2}} \right| \leq \frac{1 }{\ell_j^{2}\left( \ln L \right)^3} \right) \leq \P \left(  \left| \hat \upeta_{j}' \right| \leq \frac{1}{\ell_j^{2}\left( \ln L \right)^3} \right) \leq  \frac{C}{\left( \ln L \right)^3}.
\end{equation}
For the event~\eqref{eq:2051}-(ii), we use that the random variable $\hat \eta_{2\Lambda',i}$ is Gaussian and that its variance is comparable to $\ell_l^{-4}$ to write
\begin{equation} \label{eq:1509}
    \P \left( \left| \hat \eta_{2\Lambda',i}\right| \geq \frac{ck^{2}}{\ell_l^{2}\left( \ln L \right)^3} \right) \leq C \exp \left(  - \frac{ck^{4}}{\left( \ln L \right)^6} \right) \leq \frac{C}{\left( \ln L \right)^3}.
\end{equation}
This result implies
\begin{multline} \label{eq:12582}
    \left\{ \Lambda' \in \Q \right\} \Delta \left(\left\{ \hat \eta_{\Lambda',i} \leq \ell_l^{-2} t_{\delta} \right\} \bigcap \bigcap_{j=0}^{l-1} \left\{  \hat \upeta'_j > \ell_j^{-2} t_{\delta} \right\} \right) \\ \subseteq \left\{ \left| \hat \eta_{2\Lambda',i} \right|\geq  \frac{c k^{2}}{\ell_l^{2}\left( \ln L \right)^3} \right\} \bigcup \bigcup_{j=0}^{l-1} \left\{ \left| \hat \upeta'_j - \ell_j^{-2} t_{\delta} \right| \leq \frac{1}{\ell_j^{2}\left( \ln L \right)^3} \right\}.
\end{multline}
Let us introduce the notation
\begin{equation*}
    E_{\Lambda'} := \left\{ \hat \eta_{\Lambda',i} \leq \ell_l^{-2} t_{\delta} \right\} \bigcap \bigcap_{j=0}^{l-1} \left\{ \hat \upeta_j' > \ell_j^{-2} t_{\delta} \right\}.
\end{equation*}
A consequence of the inclusion~\eqref{eq:12582} is the inequality of indicator functions
\begin{equation} \label{eq:1507}
        \left|\indc_{\left\{\Lambda'\in \Q \right\}} - \indc_{E_{\Lambda'}} \right| \leq \indc_{\left\{ \left| \hat \eta_{2\Lambda',i} \right| \geq  \frac{c k^{2} }{\ell_l^{2}\left( \ln L \right)^3} \right\}} + \sum_{j= 0}^{l-1} \indc_{\left\{ \left| \hat \upeta_{j}' - \ell_j^{-2} t_{\delta} \right| \leq \frac{1}{\ell_j^{2}\left( \ln L \right)^3} \right\}}.
\end{equation}
We note that the event $E_{\Lambda'}$ is measurable with respect to the $\sigma$-algebra generated by the random variables $\hat \eta_{\Lambda',i}$ and $\eta_{\Lambda_0 \setminus 2\Lambda',i}$. Using Lemma~\ref{lemma.08538} and the fact that the event $E_{\Lambda'}$ is contained in the event $\left\{ \hat \eta_{\Lambda',i} \leq \ell_l^{-2} t_{\delta } \right\}$, we see that
\begin{align} \label{eq:1506}
    \E \left[\indc_{E_{\Lambda'}} \frac{1}{|\Lambda'|} \sum_{v \in \Lambda'} \left\langle  \sigma_{v,i} \right\rangle^{\tau , h}_{\Lambda_0} \right] & = \E \left[\E \left[\indc_{E_{\Lambda'}} \frac{1}{|\Lambda'|} \sum_{v \in \Lambda'} \left\langle  \sigma_{v,i} \right\rangle^{\tau , h}_{\Lambda_0} ~ \Big\vert ~ {\hat\eta_{\Lambda',i}} ,  \eta_{\Lambda_0 \setminus 2\Lambda',i}  \right]\right] \\ \notag & = \E \left[ \indc_{E_{\Lambda'}} \E \left[ \frac{1}{|\Lambda'|} \sum_{v \in \Lambda'} \left\langle  \sigma_{v,i} \right\rangle^{\tau , h}_{\Lambda_0}  ~ \Big\vert ~ {\hat\eta_{\Lambda',i}} ,  \eta_{\Lambda_0 \setminus 2\Lambda',i} \right] \right] \\ \notag & \leq C \delta \E \left[ \indc_{E_{\Lambda'}} \right].
\end{align}
We can now conclude the proof of the inequality~\eqref{eq:1457}. Applying the estimates~\eqref{eq:1508},~\eqref{eq:1509},~\eqref{eq:1507}, the computation~\eqref{eq:1506}, and the upper bound $l \leq C \ln L$, we obtain
\begin{align*}
    \lefteqn{\E \left[  \indc_{\{ \Lambda' \in \Q \}}   \frac{1}{\left| \Lambda' \right|}\sum_{v \in \Lambda'} \left\langle  \sigma_{v,i} \right\rangle^{\tau , h}_{\Lambda_0}  \right]} \qquad & \\ & \leq \E \left[  \indc_{E_{\Lambda'}}  \frac{1}{\left| \Lambda' \right|}\sum_{v \in \Lambda'} \left\langle  \sigma_{v,i} \right\rangle^{\tau , h}_{\Lambda_0}  \right] + \E \left[ \indc_{\left\{ \left|\hat \eta_{2\Lambda',i}\right| \geq \frac{ck^{2}}{\ell_l^{2} \left( \ln L \right)^3} \right\}} + \sum_{j= 0}^{l-1} \indc_{\left\{ \left| \hat \upeta_{j}' - \frac{t_{\delta}}{\ell_j^{2}} \right| \leq \frac{1}{\ell_j^{2}\left( \ln L \right)^3} \right\}} \right] \\
    & \leq C \delta \E \left[  \indc_{E_{\Lambda'}} \right] +  \E \left[ \indc_{\left\{ \left|\hat \eta_{2\Lambda',i}\right| \geq  \frac{c k^{2}}{ \ell_l^{2}\left( \ln L \right)^3} \right\}} + \sum_{j= 0}^{l-1} \indc_{\left\{ \left| \hat \upeta_{j}' - \frac{t_{\delta}}{\ell_j^{2}} \right| \leq \frac{1}{\ell_j^{2}\left( \ln L \right)^3} \right\}} \right] \\
    & \leq C \delta \E \left[  \indc_{\{ \Lambda' \in \Q \}} \right] + 2 \E \left[ \indc_{\left\{ \left|\hat\eta_{2\Lambda',i } \right| \geq  \frac{c k^2 }{\ell_l^{2}\left( \ln L \right)^3} \right\}} + \sum_{j= 0}^{l-1} \indc_{\left\{ \left| \hat \upeta_{j}' -  \frac{t_{\delta}}{\ell_j^{2}} \right| \leq \frac{1}{\ell_j^{2}\left( \ln L \right)^3} \right\}} \right] \\
    & \leq C \delta \E \left[  \indc_{\{ \Lambda' \in \Q \}} \right] + \frac{C l}{(\ln L)^3} \\
    & \leq  C \delta \E \left[  \indc_{\{ \Lambda' \in \Q \}} \right] + \frac{C}{(\ln L)^{2}}.
\end{align*}
The proof of~\eqref{eq:1457} is complete.

\medskip

We now use the two properties~\eqref{eq:14522} and~\eqref{eq:1457} of the collection $\Q$ of good boxes to complete the proof of Lemma~\ref{lemmasecondmandelbrot}. We write
$\sum_{\Lambda' \subseteq \Lambda_L}$ to refer to the sum $\sum_{l = 0}^{l_{\max}} \sum_{\Lambda' \in \mathcal{T}_l}$.
We decompose the expectation and write
\begin{align} \label{eq:1451}
    \lefteqn{\E \left[ \frac{1}{\left| \Lambda_L\right|} \sum_{v \in \Lambda_L} \left\langle  \sigma_{v,i} \right\rangle^{\tau , h}_{\Lambda_0} \right]} \qquad & \\ & = \E \left[ \sum_{\Lambda' \subseteq \Lambda_L} \frac{|\Lambda'|}{\left| \Lambda_L\right|}\indc_{\{ \Lambda' \in \Q \}}   \frac{1}{\left| \Lambda' \right|}\sum_{v \in \Lambda'} \left\langle  \sigma_{v,i} \right\rangle^{\tau , h}_{\Lambda_0}  \right] + \E \left[ \frac{1}{\left| \Lambda_L\right|} \sum_{v \in \Lambda_L} \indc_{\{ v \mathrm{ \, is\, uncovered}\}}  \left\langle  \sigma_{v,i} \right\rangle^{\tau , h}_{\Lambda_0} \right] \notag
    \\ & = \underbrace{\sum_{\Lambda' \subseteq \Lambda_L} \frac{|\Lambda'|}{\left| \Lambda_L\right|} \E \left[  \indc_{\{ \Lambda' \in \Q \}}   \frac{1}{\left| \Lambda' \right|}\sum_{v \in \Lambda'}\left\langle  \sigma_{v,i} \right\rangle^{\tau , h}_{\Lambda_0} \right]}_{\eqref{eq:1451}-(i)} +  \underbrace{\frac{1}{\left| \Lambda_L\right|} \sum_{v \in \Lambda_L}  \E \left[ \indc_{\{ v \mathrm{\, is\, uncovered}\}}  \left\langle  \sigma_{v,i} \right\rangle^{\tau , h}_{\Lambda_0} \right].}_{\eqref{eq:1451}-(ii)} \notag
\end{align}
We then estimate the two terms in the right side separately. We begin with the term~\eqref{eq:1451}-(i), use the inequality~\eqref{eq:1457} and the observations
\begin{equation*}
    \sum_{\Lambda' \subseteq \Lambda_L} |\Lambda'| \indc_{\{ \Lambda' \in \Q \}} \leq \left| \Lambda_L\right| ~ \mbox{and} ~ \sum_{\Lambda' \subseteq \Lambda_L} |\Lambda'| = \sum_{l=0}^{l_{\max}} \sum_{\Lambda' \in \mathcal{T}_l} |\Lambda'| = (l_{\max} +1 ) \left| \Lambda_L\right| \leq C \ln L \left| \Lambda_L\right|.
\end{equation*}
We obtain
\begin{align*}
    \sum_{\Lambda' \subseteq \Lambda_L} \frac{|\Lambda'|}{\left| \Lambda_L\right|} \E \left[  \indc_{\{ \Lambda' \in \Q \}}   \frac{1}{\left| \Lambda' \right|}\sum_{v \in \Lambda'} \left\langle  \sigma_{v,i} \right\rangle^{\tau , h}_{\Lambda_0}  \right] & \leq C \delta \sum_{\Lambda' \subseteq \Lambda_L} \frac{|\Lambda'|}{\left| \Lambda_L\right|} \E \left[  \indc_{\{ \Lambda' \in \Q\}} \right] + C \sum_{\Lambda' \subseteq \Lambda_L} \frac{|\Lambda'|}{\left| \Lambda_L\right|} \frac{1}{\left( \ln L \right)^2} \\
    & \leq \frac{C\delta}{\left| \Lambda_L\right|} \E \left[ \sum_{\Lambda' \subseteq \Lambda_L} |\Lambda'| \indc_{\{ \Lambda' \in \Q \}}   \right]  + \frac{C \ln L}{\left( \ln L \right)^2} \\
    & \leq C \delta.
\end{align*}
There only remains to treat the term~\eqref{eq:1451}-(ii). We use to the estimate~\eqref{eq:14522} and write
\begin{align*}
     \left| \frac{1}{\left| \Lambda_L\right|} \sum_{v \in \Lambda_L}  \E \left[ \indc_{\{ v \mathrm{\, is \, uncovered}\}} \left\langle  \sigma_{v,i} \right\rangle^{\tau , h}_{\Lambda_0}   \right] \right| & \leq \frac{1}{\left| \Lambda_L\right|} \sum_{v \in \Lambda_L} \P \left[ v ~\mbox{is uncovered} \right] \\ &
     \leq \exp \left( - c \sqrt{\ln L} \right) \\ &
     \leq \frac{C}{\sqrt{\ln \ln L}}.
\end{align*}
A combination of the two previous displays with the identity~\eqref{eq:1451} completes the proof of Lemma~\ref{lemmasecondmandelbrot}.
\end{proof}

\subsubsection{Proof of Theorem~\ref{thm.thm2}} \label{subsecprooftheorem2}

In this section, we combine the result of Lemma~\ref{lemmasecondmandelbrot} (applied with the periodic boundary condition) with an argument similar to the one developed in Subsection~\ref{subsec4.3.21739} to complete the proof of Theorem~\ref{thm.thm2}.

\begin{proof}[Proof of Theorem~\ref{thm.thm2}]
Fix a side length $L \geq 10$, $h \in \R^n$ such that $\left| h \right| \leq 1/10$, and set $\ell := (L/2) \wedge |h|^{-1/2}$. We consider the system with periodic boundary condition and apply Lemma~\ref{lemmasecondmandelbrot} with the boxes $\Lambda_0 := \Lambda_L$ and $\Lambda_L = \Lambda_\ell$. We obtain
\begin{equation*}
    \left| \E \left[ \left\langle \sigma_0 \right\rangle^{\per , h}_{\Lambda_{L}} \right] \right| = \left| \E \left[ \frac{1}{\left| \Lambda_\ell \right|}\sum_{v \in \Lambda_\ell}\left\langle \sigma_v \right\rangle^{\per , h}_{\Lambda_{L}} \right] \right| \leq C \frac1{\sqrt{\ln \ln \ell}} \leq  \frac C{\sqrt{\ln \ln  \left( L \wedge |h|^{-1} \right)}},
\end{equation*}
where we used the inequality $\ln \ell  \geq c \ln \left(L \wedge |h|^{-1}\right)$. Integrating over $h$ as it was done in~\eqref{eq:1352199}, we deduce that
\begin{equation}  \label{eq:13521999}
    \left| \E \left[ \FE^{\per , h}_{\Lambda_L}(\eta) - \FE^{\per , 0}_{\Lambda_L} (\eta)\right] \right|  \leq \frac{C |h|}{\sqrt{\ln \ln   \left( L \wedge |h|^{-1} \right)}}.
\end{equation}
Let us then fix an integer $i \in \{ 1 , \ldots, n \}$. For each realization of the random field $\eta$ and each $h \in \R^n$, we let $\tau_i(\eta , h) \in \S^{\partial \Lambda_L}$ be a boundary condition satisfying
\begin{equation} \label{eq:120105111}
    \frac{1}{\left| \Lambda_L \right|}  \sum_{v \in \Lambda_L} \left\langle  \sigma_{v,i} \right\rangle^{\tau_i(\eta , h) , h}_{\Lambda_L}  = \sup_{\tau \in \S^{\partial \Lambda_L}}  \frac{1}{\left| \Lambda_L \right|} \sum_{v \in \Lambda_L} \left\langle  \sigma_{v,i} \right\rangle^{\tau , h}_{\Lambda_L}.
\end{equation}
Using the inequality~\eqref{eq:12030211} of Proposition~\ref{prop.basicpropfreeen}, we have, for any $h, h' \in \R^n$,
\begin{equation} \label{eq:13531999}
    \left| \E \left[ \FE^{\tau_i(\eta , h) , h'}_{\Lambda_L}(\eta) \right] - \E \left[ \FE^{\per , h'}_{\Lambda_L}(\eta) \right] \right| \leq \frac{C}{L}.
\end{equation}
A combination of the inequalities~\eqref{eq:13521999} and~\eqref{eq:13531999} yields, for any $h, h' \in \R^n$ satisfying $ |h'| \leq 1/10$,
\begin{equation} \label{eq:135619999}
    \left| \E \left[ \FE^{\tau_i(\eta , h) , h'}_{\Lambda_L}(\eta) \right] - \E \left[ \FE^{\tau_i(\eta, h) , 0}_{\Lambda_L} (\eta) \right] \right| \leq\frac{C |h'|}{\sqrt{\ln \ln   \left( L \wedge |h'|^{-1} \right)}} + \frac{C}{L}. \\
\end{equation}
We then fix $|h| \leq 1/10$, and denote by
\begin{equation*}
\tilde h :=
\left(h_1 , \ldots, h_{i-1}, h_i + \sqrt{ |h| \vee L^{-1}} , h_{i+1} , \ldots, h_n\right).
\end{equation*}
We note that $|\tilde h| \leq 2 \sqrt{ |h| \vee L^{-1}} \leq 1$, and that, if $|h|$ and $1/L$ are sufficiently small (e.g., smaller than $1/400$), then $|\tilde h|$ is smaller than $1/10$ and we may apply the bound~\eqref{eq:135619999}.

We introduce the function $G: h' \mapsto -\E \left[ \FE^{\tau_i(\eta , h) ,h'}_{\Lambda_L}(\eta) \right]$, observe that the map $G$ is convex, and that its derivative with respect to the $i-$th variable satisfies
\begin{equation*}
    \frac{\partial G}{\partial h_i'}(h') = \E \left[ \frac{1}{\left| \Lambda_L \right|} \sum_{v \in \Lambda_L} \left\langle  \sigma_{v,i} \right\rangle^{\tau_i(\eta , h) , h'}_{\Lambda_L} \right].
\end{equation*}
Additionally, by~\eqref{eq:135619999} and the definition of $\tilde h$ (assuming that $|h|$ and $1/L$ are sufficiently small), we have
\begin{equation*}
    \left| G(\tilde h) - G(h) \right| \leq \frac{C \sqrt{L^{-1} \vee |h| }}{\sqrt{\ln \ln  \left( L \wedge |h|^{-1} \right)}} + \frac{C}{L} \leq  \frac{C \sqrt{L^{-1} \vee |h|}}{\sqrt{\ln \ln \left( L \wedge |h|^{-1} \right)}},
\end{equation*}
which then yields
\begin{align} \label{eq:16180511}
    \E \left[ \frac{1}{\left| \Lambda_L \right|} \sum_{v \in \Lambda_L} \left\langle  \sigma_{v,i} \right\rangle^{\tau_i(\eta , h) , h}_{\Lambda_L} \right]
     & \leq  \frac{G(\tilde h) - G(h)}{\sqrt{|h|\vee L^{-1}}} \notag \\
    & \leq \frac{C}{\sqrt{\ln \ln \left(  L \wedge |h|^{-1} \right)}}. \notag
\end{align}
Using the definition of the random boundary condition $\tau_i\left( \eta , h \right)$, we obtain
\begin{equation*}
    \E \left[ \sup_{\tau \in \S^{\partial \Lambda_L}}\frac{1}{\left| \Lambda_L \right|} \sum_{v \in \Lambda_L} \left\langle  \sigma_{v,i} \right\rangle^{\tau , h}_{\Lambda_L} \right] \leq \frac{C}{\sqrt{\ln \ln \left( L \wedge |h|^{-1}  \right)}}.
\end{equation*}
The proof of Theorem~\ref{thm.thm2} can be completed by using the same arguments as the one presented in the proof of the estimate~\eqref{eq:TV090404} of Theorem~\ref{prop:subcritical} in Section~\ref{subsec4.3.21739}, we thus omit the details.
\end{proof}

\section{Discussion and open problems} \label{sectioncommentsandopenproblems}
This work initiates the study of quantitative versions of the Aizenman--Wehr~\cite{AW1989, AW89} result on the Imry--Ma rounding phenomenon~\cite{IM75}. In this section we discuss some of the remaining problems:

\medskip

{\bf Uniqueness conjecture:} As discussed in Section~\ref{sec:2d systems}, we believe that a stronger qualitative statement than the one provided by Aizenman--Wehr~\cite{AW1989, AW89} is valid. Namely, that in two dimensions the thermal averages $\left\langle f_v \left( \sigma \right) \right\rangle_{\Lambda_L}^{\tau}$ cannot be significantly altered by changing $\tau$ in the sense of Conjecture~\ref{conj:uniqueness}. We again point out that Conjecture~\ref{conj:uniqueness} would imply as a special case the well-known belief that the two-dimensional Edwards--Anderson spin glass model has a unique ground-state pair. It would be very interesting to make additional progress in this direction.

It is also possible that analogous uniqueness statements hold in dimensions $d=3,4$ for the spin systems with continuous symmetry discussed in Section~\ref{sec:continuous symmetry results} (the case $d=2$ is covered by Conjecture~\ref{conj:uniqueness}).

\medskip

{\bf Quantitative decay rate and possible phase transitions:} It is very natural to seek the optimal rates in our quantitative results. We first discuss the general two-dimensional disordered spin systems of Section~\ref{sec:2d systems}. When the base Hamiltonian has finite range, exponential decay of correlations and uniqueness of the infinite-volume Gibbs measure follow in all dimensions in the high-temperature regime ($\beta\ll 1$), and for certain systems also in the strong disorder regime ($\lambda\gg1)$, from the disagreement percolation methods of van den Berg--Maes~\cite{van1994disagreement} and their adaptations by Gielis--Maes~\cite{gielis1995uniqueness} (suitable versions of Dobrushin's condition~\cite{D68} should also be applicable). The main interest is thus in the low-temperature (or even zero temperature) regime. We then believe that the correct rate in two dimensions should be much faster than the inverse power of log-log rate obtained in Theorem~\ref{thm1prop2.31708main}. Without assuming translation invariance, the decay cannot hold at a faster than power-law rate (e.g., as the noised observables $(f_v)$ can be identically zero at all but one vertex), but it is possible that such a rate indeed holds in general. In the translation-invariant setup, it is even possible that exponential decay holds (noting that for a faster than power-law decay one needs to perform the spatial average in~\eqref{eq:fluctuations around limiting value} over a smaller domain, say $\Lambda_{L/2}$, to avoid boundary effects), as proved for the nearest-neighbor ferromagnetic random-field Ising model ~\cite{ding20192exponential, AHP20} (though exponential decay is still open for the non-nearest-neighbor version for which only a power-law upper bound is known~\cite{AP19}).

We proceed to discuss the spin systems with continuous symmetry of Section~\ref{sec:continuous symmetry results}. Exponential decay in the high-temperature regime ($\beta\ll 1$) again follows in all dimensions~\cite{van1994disagreement, gielis1995uniqueness}. This is also expected in the strong disorder regime ($\lambda\gg1$); see~\cite{feldman2000nonexistence} for the zero-temperature case. The interest is thus in the low-temperature and weak disorder regime. One might expect exponential decay to hold in the sub-critical dimensions $d = 2,3$, and arguments have been given both in favor~\cite{3R, 180R} and, in the three dimensional case, against~\cite{83R, 99R} this possibility in the physics literature. Power-law decay would imply a transition of the Berezinskii–Kosterlitz–Thouless type~\cite{berezinskii1971destruction, berezinskii1972destruction, kosterlitz1972long, KT, FrSp} as the temperature or disorder strength varies and would thus be of great interest.

\medskip

{\bf Higher-order continuous symmetries:} The form of continuous-symmetry that the spin $O(n)$ model enjoys is that its formal Hamiltonian
\begin{equation*}
H(\sigma):=\sum_{v\sim w} |\sigma_v - \sigma_w|^2,
\end{equation*}
where $\sigma:\Zd\to\mathbb{S}^{n-1}$, satisfies $H(R\sigma) = H(\sigma)$ for the operation $R$ which rotates all spins in $\sigma$ by the same angle. One can also envision spin systems enjoying higher-order symmetries in the sense that we now explain for the $n=2$ case (analogous symmetries may exist for $n>2$ but we have not explored this possibility). We embed the circle $\mathbb{S}^1$ in $\mathbb{C}$ and write $\sigma_v = e^{i\theta_v}$ (where the angle $\theta$ is defined modulo $2\pi$). For a polynomial $P:\R^d\to\R$ and configuration $\sigma:\Zd\to\mathbb{S}^1$ define the `polynomial rotation' $R_P(\sigma)$ by $R_P(\sigma)_v = e^{i (\theta_v + P(v))}$. A Hamiltonian $H$ on configurations $\sigma:\Zd\to\mathbb{S}^{1}$ is then said to enjoy a continuous symmetry of order $k$ if $H(R_P(\sigma)) = H(\sigma)$ for all polynomials $P:\R^d\to\R$ of degree at most $k$. In particular, the case $k=0$ corresponds to the usual notion of continuous symmetry. As an example of a Hamiltonian enjoying continuous symmetries of order $1$, we propose
\begin{equation}\label{eq:Laplacian Hamiltonian}
H(\sigma):=\sum_{v} \cos((\Delta \theta)_v)
\end{equation}
where $\Delta$ is the discrete Laplacian operator: $(\Delta\theta)_v := \sum_{w\colon w\sim v} (\theta_w - \theta_v)$. Similarly, an example of a Hamiltonian enjoying continuous symmetry of order $k=2\ell -1$ is obtained by replacing $\Delta$ with $\Delta^\ell$ (the composition of $\Delta$ with itself $\ell$ times) in~\eqref{eq:Laplacian Hamiltonian} and an example enjoying continuous symmetries of order $k=2\ell$ is furnished by
\begin{equation}
H(\sigma):=\sum_{v\sim w} \cos((\Delta^\ell \theta)_v - (\Delta^\ell \theta)_w).
\end{equation}
We do not know if these spin models have received attention in the literature.

Higher-order symmetries reduce the surface tension of finite-range spin systems with smooth energy. Specifically, we believe that an analog of Proposition~\ref{propMermW} holds for a spin system having a smooth, finite-range Hamiltonian enjoying a higher-order symmetry of order $k$ with the factor $\ell^{d-2}$ replaced by $\ell^{d-2k-2}$. To prove this fact, one may follow the steps of Proposition~\ref{propMermW} with the following modifications:
\begin{itemize}[leftmargin=*]
    \item References to $\Psi$ should be replaced by corresponding references to the Hamiltonian.
    \item The spin wave in equation~\eqref{def.thetaMW} needs to be replaced by a function $\theta$ satisfying $\theta\equiv 0$ on $\Zd\setminus 2\Lambda$, $\theta\equiv \pi$ in $\Lambda$ and all discrete derivatives of order $k+1$ of $\theta$ are uniformly bounded by $C_{k+1}\ell^{-(k+1)}$. Such a function may be obtained by choosing a smooth function $\psi:\R^d\to\R$ satisfying that $\psi\equiv 0$ outside the box $B(2)$, $\psi\equiv 1$ on the box $B(1)$ (where $B(r):=\{x\in\R^d\colon \|x\|_\infty\le r\}$) and then setting $\theta(v):=\pi\psi(v / \ell)$ for $v\in\Zd$.
    \item Instead of the expression~\eqref{eq:155924} one notes that a discrete Taylor expansion may be performed to write $\theta_w = \theta_v + P_{k,v}(w-v)+\varepsilon_{k,v,w}$ where $P_{k,v}$ is a polynomial of degree at most $k$ and then the higher-order symmetry of the Hamiltonian allows to replace the expression $\theta_w - \theta_v$ on the right-hand side of~\eqref{eq:155924} by $\varepsilon_{k,v,w}$ which is of order at most $C_{k+1}\ell^{-(k+1)}$ by our assumptions on $\theta$.
\end{itemize}
The reduced surface tension allows to push the Imry--Ma phenomenon to higher dimensions. Specifically, spin systems with a finite-range smooth Hamiltonian enjoying a higher-order symmetry of order $k$ will lose their low-temperature ordered phase upon introduction of a random field of arbitrarily weak intensity, of the form in our theorems, in all dimensions $d\le 4(k+1)$. Moreover, the strategy used in this paper to obtain a quantitative decay rate can be followed to yield that (at $h=0$)
\begin{equation}
\left| \E \left[\frac{1}{\left| \Lambda_L \right|}  \sum_{v \in  \Lambda_L } \left\langle \sigma_v \right\rangle^{\tau , 0}_{\Lambda_{2L}} \right] \right|\le C L^{-2(k+1) + d/2}
\end{equation}
in dimensions $d<4(k+1)$. It is further possible that the strategy used in the proof of Theorem~\ref{thm.thm2} can be adapted to yield a bound in dimension $d=4(k+1)$.

\medskip

{\bf Comparison with the results of Aizenman--Wehr:} The seminal result of Aizenman--Wehr~\cite{AW1989, AW89} establishes rigorously the rounding of the first-order phase transitions of low-dimensional spin systems upon the addition of a quenched disorder. Our work presents a quantified version of the Aizenman--Wehr theorem, but applies in somewhat different generality than the original result. While we expect that the techniques developed in this work may be extended to a more general setup, closer to that of~\cite{AW1989, AW89}, we have not pursued this direction. In this section, we elaborate on the various assumptions made:
\begin{itemize}[leftmargin=*]
\item {\it Translation-invariance of the systems:} While the proof of Aizenman and Wehr requires to work in a translation-invariant setup, the techniques developed in this article apply to spin systems which do not satisfy this assumption.
\item {\it Distribution of the disorder:} The result of Aizenman--Wehr applies to a wide class of disorder distributions while our result is presented only for the Gaussian case.
\item {\it Range of the interaction:} The results of~\cite{AW1989, AW89} apply also to disordered systems with long-range interactions (in which case one-dimensional systems are also of interest) as long as these decay at a sufficiently fast rate. Our results for general two-dimensional disordered spin systems (Section~\ref{sec:2d systems}) allow the \emph{base system} to have arbitrary interactions as long as the bounded boundary effect condition~\eqref{def.cteCH} holds. However, we have opted to restrict to finite-range dependencies in the noised observables $(f_v)_{v \in \Zd}$ (Section~\ref{sec:general setup}).

\item {\it Uniformity of the results in the temperature and external field:} The results of Aizenman--Wehr apply not only for a fixed value of the temperature and external magnetic field (the latter is incorporated into the models there) but also uniformly when these parameters are themselves allowed to depend on the disorder $\eta$. This uniformity shows that there cannot be deviations from the proven behavior at \emph{random} critical points. In comparison, our results are stated only for a fixed value of the temperature and, in the general two-dimensional setup, without an external magnetic field (though one can be included in the base Hamiltonian). Still, a uniform version of our results may be obtained with minor modifications of the proof, as indicated in Remark~\ref{Remark uniform beta}.

\item {\it Systems with continuous symmetry:} For systems with continuous symmetries, Aizenman and Wehr allow the spins to take values in a subset of $\R^n$ which is invariant under the action of a closed \emph{connected} subgroup $G$ of $O(n)$, and require the base Hamiltonian to take the form $\sum_{x,y} J_{x-y} \Psi_{x-y}(\sigma_x,\sigma_y)$ where each $\Psi_z$ is required to be bounded by $1$, to be invariant under the same subgroup, to satisfy certain differentiability properties and the coupling constants $J_z$ are required to satisfy $\sum_{z\colon \left|z\right|\le L} J_z|z|^2\le C L^{(4-d)/2}$. The noised observables are taken to be the spins themselves and the conclusion is then that in all dimensions $d\le 4$, a suitable spatially and thermally averaged magnetization value is invariant under all elements of $G$ for all Gibbs measures of the model, almost surely. It is further mentioned that the technique should generalize to a suitable class of many-body interactions. Our results for spin systems with continuous symmetry (Section~\ref{sec:continuous symmetry results}) are presented, again for simplicity, only for the case that the spins take values in the sphere $\mathbb{S}^{n-1}$ in the setup of nearest-neighbor interactions which are invariant under all rotations in $O(n)$, and when the noised observables are fixed to be the spins themselves.
\end{itemize}

\subsubsection*{Acknowledgements} We are grateful to Michael Aizenman, Ronald Fisch, David Huse, Charles M. Newman, Thomas Spencer and Daniel L. Stein for encouragement and helpful conversations on the topics of this work. We also thank Antonio Auffinger, Wei-Kuo Chen, Izabella Stuhl and Yuri Suhov for the opportunity to present these results in online talks and for useful discussions. We would like to thank an anonymous referee for many helpful comments which helped us improve the presentation of the article.
The research of the authors was supported in part by Israel Science Foundation grants 861/15  and  1971/19  and  by  the  European  Research  Council starting  grant 678520 (LocalOrder).

\small
\bibliographystyle{abbrv}
\bibliography{references}

\end{document}